\documentclass[11pt]{article}
\usepackage{color}
\usepackage[colorlinks=true,citecolor=blue]{hyperref}
\usepackage{cite}
\usepackage{amssymb,amsmath,amsfonts,amsthm,amsthm}
\usepackage{mathtools}
\usepackage[english]{babel}
\usepackage[utf8]{inputenc}
\usepackage[margin=0.90in]{geometry}
\usepackage[margin=0.2in]{caption}
\usepackage{graphicx}
\usepackage{times,fourier}
\usepackage[caption = false]{subfig}
\usepackage{extarrows}
\usepackage[toc, title]{appendix}
\usepackage{tikz}

\usepackage{enumitem}
\usepackage{pifont}
\usepackage{pgfplots}
\usepackage{extarrows}
\usepackage{overpic}
\usepackage{caption}      
\usepackage{subfig}  

\makeatletter
\renewcommand\section{\@startsection {section}{1}{\z@}%
	{-2ex \@plus -1ex \@minus -.2ex}%
	{1ex \@plus.1ex}%
	{\normalfont\bf\sffamily}}
\renewcommand\subsection{\@startsection{subsection}{2}{\z@}%
	{-1.75ex\@plus -0.4ex \@minus -.2ex}%
	{0.6ex \@plus .1ex}%
	{\normalfont\small\bf\sffamily}}
\renewcommand\subsubsection{\@startsection{subsubsection}{3}{\z@}%
	{-0.6ex\@plus -0.2ex \@minus -.2ex}%
	{0.4ex \@plus .1ex}%
	{\normalfont\normalsize\it}}
\renewcommand\paragraph{\@startsection{paragraph}{4}{\z@}%
	{0.2ex \@plus0.2ex \@minus0.1ex}{-0.5em}%
	{\normalfont\normalsize\bfseries}}
\def\ps@headings{%
	\let\@oddfoot\@empty
	\let\@evenfoot\@empty
	\def\@evenhead{\small\sffamily\thepage\hfil\slshape\leftmark}%
	\def\@oddhead{\small\sffamily{\slshape\rightmark}\hfil\thepage}%
	\let\@mkboth\markboth
	\def\chaptermark##1{\markboth{{\ifnum \c@secnumdepth >\m@ne
				\if@mainmatter \@chapapp\ \thechapter. \ \fi \fi ##1}}{}}%
	\def\sectionmark##1{\markright {{\ifnum \c@secnumdepth >\z@
				\thesection. \ \fi ##1}}}}
\def\fbf#1{\setbox0=\hbox{$#1$}\kern-0.10\wd0
	\lower0.02em\copy0\kern-\wd0 \lower0.02em\hbox{\kern+0.04em\copy0}\kern-\wd0
	\raise0.00em\copy0\kern-\wd0 \raise0.00em\hbox{\kern-0.04em\box0}}
\makeatother

\advance\textwidth -8mm
\advance\textheight -10mm
\advance\hoffset 4mm
\advance\voffset 4mm
\advance\textheight 0mm
\advance\parskip 4pt
\numberwithin{equation}{section}
1

\newtheorem{theorem}{Theorem}[section]
\newtheorem{problem}{RH problem}[section]

\newtheorem{lemma}[theorem]{Lemma}
\newtheorem{remark}[theorem]{Remark}
\newtheorem{proposition}[theorem]{Proposition}
\newtheorem{assumption}[theorem]{Assumption}

\def\maketitle{\par\noindent{\LARGE\bf\sffamily\thetitle}\\[1.4ex]
	{\large\theauthor}\\[0.6ex]
	\textit{\thetextinfo}\\[0.2ex]
	{\small\today}\par\vglue1.4\bigskipamount}
\def\title#1{\def\thetitle{#1}}
\def\author#1{\def\theauthor{#1}}
\def\textinfo#1{\def\thetextinfo{#1}}
\def\be{\begin{equation}}
	\def\ee{\end{equation}}
\def\bse{\begin{subequations}}
	\def\ese{\end{subequations}}

\definecolor{deeppurple}{rgb}{0.5, 0, 0.7}

\def\gl{\mathrel{\mathpalette\overl@ss>}}

\def\@#1{{\mathbf{#1}}}
\def\_#1{{\mathsf{#1}}}
\let\~=\tilde
\let\==\bar
\let\^=\hat
\let\<=\langle
\let\>=\rangle

\def\max{\mathop{\rm max}\nolimits}

\def\note[#1]{\marginpar{\color{red}[#1]}}

\def\XXint#1#2#3{{\setbox0=\hbox{$#1{#2#3}{\int}$}
		\vcenter{\hbox{$#2#3$}}\kern-.5\wd0}}

\def\1{{\bf 1}}

\makeatother

\let\trueparagraph=\paragraph
\def\paragraph#1{\par\smallskip\trueparagraph{\rm\textbf{#1}}}
\advance\textheight 12mm
\advance\voffset -6mm
\allowdisplaybreaks

\begin{document}
\pagestyle{plain}
\title{\bf   Painlevé Asymptotics of the Focusing Nonlinear Schr\"odinger Equation with a Finite-Genus Algebro-Geometric Background}
\author{\large
	Ruihong Ma and Engui Fan}
\textinfo
{\normalsize\it
1:School of Mathematical Sciences, Peking University,
Beijing 100871, P.R. China\\\normalsize\it
2: School of Mathematical Sciences, Fudan University, Shanghai 200433, P.R. China \\\normalsize\it
	Authors' Email: faneg@fudan.edu.cn}
\maketitle
\kern-4ex
\begin{abstract}
We investigate the Cauchy problem for the focusing nonlinear Schr\"odinger (NLS) equation
\begin{equation}
	iq_t(x,t)+q_{xx}(x,t)+2|q(x,t)|^2q(x,t)=0,\quad x\in\mathbb{R},\quad t\ge0,\nonumber
\end{equation}
subject to initial data $  q(x,0)$ satisfying the asymptotic boundary conditions
\begin{equation}\label{eq:boundary}
q(x,0)  \sim q^{alg}(x,0) \quad \text{as} \quad x \to \pm\infty,\nonumber
\end{equation}
where $q^{alg}(x,t)$ denote finite-genus algebro-geometric quasi-periodic solutions of the focusing NLS equation.
Employing the Riemann--Hilbert (RH) approach combined with the Deift--Zhou nonlinear steepest descent method, we analyze the long-time asymptotic behavior of solutions to this Cauchy problem. Our analysis distinguishes between two cases based on the genus $n$ of the underlying hyperelliptic Riemann surface:
\begin{itemize}
	\item[(i)] \textbf{Odd genus backgrounds:} When the background solutions $q^{alg}(x,0)$ correspond to hyperelliptic curves of odd genus $n = 2s+1$ $(s \in \mathbb{N}_0)$, we identify distinct asymptotic regions in the $(x,t)$-plane characterized by the variable $\xi = x/t$, within which the leading-order asymptotics is expressed in terms of the second Painlev\'e transcendent.
	\item[(ii)] \textbf{Even genus backgrounds:} When the background solutions $q^{alg}(x,0)$ correspond to hyperelliptic curves of even genus $n = 2s$ $(s \in \mathbb{N})$, the asymptotic behavior in regions selected by $\xi$ is described in terms of parabolic cylinder functions.
\end{itemize}
Specifically, we derive the leading-order asymptotics and establish explicit error bounds for the solution $q(x,t)$ as $t \to +\infty$, uniformly for $x \in \mathbb{R}$.
\end{abstract}
\baselineskip=16pt
\tableofcontents%
\section{Introduction}\label{sec1}
We consider the Cauchy problem for the focusing nonlinear Schr\"odinger(NLS)  equation
\begin{subequations}
	\begin{align}\label{eq:fNLS}
		&iq_t(x,t)+q_{xx}(x,t)+2|q(x,t)|^2q(x,t)=0,\quad x\in\mathbb{R},\quad t\ge0,\\\label{eq:q_0}
		&q(x,0)=q_0(x),
\end{align}
\end{subequations}
where the initial data $q_0(x)$ behaves asymptotically like a finite-genus algebro-geometric solution \cite{KS17}
\begin{equation}\label{eq:a-alg-initi}
	q(x,0)\sim q^{alg}(x,0),\quad x\to\pm\infty.
\end{equation}
Here, $q^{alg}(x, t)$ represents an $n$-genus algebro-geometric solution to the focusing NLS equation \eqref{eq:fNLS} with spectral cuts $\left[\bar{E}_k, E_k\right], k\in\mathcal{K}=\{0,\dots, n\},$  explicitly given by
\begin{equation}\label{eq:q-alg}
	q^{alg}(x,t)=-\frac{i}{2}\operatorname{Im}\left(\sum_{k=0}^nE_k\right)e^{2i(f_0x+g_0t)}\frac{\Theta(\varphi(\infty)+d)\,\Theta(\varphi(\infty)-c(x,t;\phi)-d)}{\Theta(\varphi(\infty)-d)\,\Theta(\varphi(\infty)+c(x,t;\phi)+d)},
\end{equation}
where
\begin{equation}
	d=\varphi(D)+K,\quad c(x,t;\phi)=\{ c_1,\ldots,c_n\},\quad c(x,t;\phi)\,|_{1\leq k\leq n}=-\frac{xC_k^f+tC_k^g+\phi_k}{2\pi}.
\end{equation}
Here, $\Theta(z)$ is the Riemann theta function defined on the genus-$n$ Riemann surface $\mathcal{X}$ ($n\in\mathbb{N}_0$) associated with the hyperelliptic curve \eqref{genus-surface}. The Abel map $\varphi: \mathcal{X} \to \mathbb{C}^n$ is given in \eqref{eq:abel}; the vector-valued Riemann constant $K \in \mathbb{C}^n$ and the pole divisor $D$ are defined on $\mathcal{X}$. The constant vectors $C_k^f, C_k^g, \phi_k \in \mathbb{C}^n$ and the constants $f_0, g_0$ are determined by the initial data.

In order for the formulation of the Cauchy problem \eqref{eq:fNLS}-\eqref{eq:q_0} to be complete, it has to be supplemented with boundary conditions for $t>0$. These boundary conditions are the natural extensions of \eqref{eq:a-alg-initi} to $t>0$ and are given by
\begin{equation}\label{eq;com-con}
	\int_{\mathbb{R}}\left|q(x, t)-q^{alg}(x, t)\right| dx<\infty \quad \text { for all } t \geq 0.
\end{equation}

The focusing NLS equation \eqref{eq:fNLS} is an integrable nonlinear partial differential equation  amenable to the inverse scattering transform   method \cite{MJH81}. The Riemann--Hilbert (RH)   formulation of this method has proven highly effective for analyzing diverse properties of integrable systems, particularly the long-time asymptotics of Cauchy problems \cite{PS94,PD93} and the small-dispersion limit \cite{PS97,SK03}.
The nonlinear steepest descent method, developed rigorously by Deift and Zhou in \cite{PD93}, was subsequently refined in \cite{PS94,PS97} through the introduction of the so-called \emph{$g$-function mechanism}. This mechanism facilitates the core objective of the method: to construct a sequence of transformations converting the original RH problem into an explicitly solvable model RH problem. Furthermore, this approach yields not only the leading-order asymptotic term but also provides a systematic
framework for deriving rigorous error estimates via appropriate local (parametrix) RH problems.

The nonlinear steepest descent method has been successfully applied to investigate long-time asymptotics for the focusing NLS equation \cite{DB21,GGD16,SLD21,BLM,MR18} and various other integrable systems \cite{AM08,KXY24}. Specific studies include: the long-time dynamics of step-like initial data $q(x,0)=0$ for $x\leq 0$ and $q(x,0)=A\mathrm{e}^{-2iBx}$ for $x>0$ with $A>0$, $B\in\mathbb{R}$ \cite{DM11}; the focusing NLS equation with step-like oscillating background exhibiting genus-3 structure \cite{MAB22}.
Even for problems with \emph{zero boundary conditions}, where solutions decay to zero as $|x|\to\infty$, the $(x,t)$-plane contains various narrow transition regions exhibiting qualitatively distinct asymptotic behaviors \cite{PS94}. Consequently, one anticipates analogous transition regions in the more intricate setting of non-zero boundary conditions \cite{GG14,GMM17,F14,F15,F16}. Notably, a transition mechanism involving so-called \emph{asymptotic solitons} associated with spectral arc endpoints has been recently reported \cite{BM19,KV19}. Additionally, the long-time asymptotics of periodic (and, slightly more generally, algebro-geometric finite-gap) solutions of the doubly infinite Toda lattice under a short-range perturbation was studied in \cite{KSS12}; long-time asymptotics of perturbed finite-gap Korteweg-de Vries solutions was studied in \cite{MLT2012,McN2021}; Painlev\'e-XXXIV asymptotics for the defocusing NLS equation with finite-genus algebro-geometric background have been investigated in \cite{FL26}.

 The aim of the present  work is to establish long-time asymptotics for   the Cauchy problem \eqref{eq:fNLS}-\eqref{eq:q_0} and (\ref{eq;com-con}) for the
 focusing NLS equation by implementing the nonlinear steepest descent method. Our  main result  is stated as follows.
	
\begin{theorem}\label{theorem-1}
	Given a finite-genus algebro-geometric solution $q^{alg}(x,t)$ of the focusing NLS equation \eqref{eq:fNLS}, as given in \eqref{eq:q-alg}, let $q(x,t)$ be the solution of the Cauchy problem \eqref{eq:fNLS}--\eqref{eq:q_0} subject to the condition  \eqref{eq;com-con}. Then, as $t\to\infty$, we have the following asymptotics for $q(x,t)$ in the regions described below:
	\begin{enumerate}
		\item
For the odd genus $n = 2s+1$ $(s \in \mathbb{N}_0)$ background solution, when the distribution of stationary phase points is given by equation \eqref{eq:zero-j} and \eqref{h-pp}, consider the regime
$|\xi-\xi_{j}|t^{2/3} \leq C$ with fixed $C>0$,
where the critical line $\xi_{j}$ and coalescence point $z_{j}$ are given in \eqref{eq:xireal}. Then
\begin{align}
	q(x,t) = -\delta^2(\infty)e^{2i(f_0x+g_0t)}\tilde{Q}(x,t) + 2e^{2i(f_0x+g_0t)}\delta^2(\infty)\frac{t^{-1/3}u(\varpi)}{\sqrt[3]{\theta'''(z_{j})}(z-z_{j})} + \mathcal{O}(t^{-1/2}),
\end{align}
where $\delta(\infty)$ and $\tilde{Q}(x,t)$ are given in \eqref{delt-ingg} and \eqref{tile-Q}, respectively. Here $u(\varpi)$ denotes the solution of the Painlev\'{e} II equation \eqref{PE-II} with asymptotic behavior
\begin{equation}
	u(\varpi) = \mathrm{Im}\,\rho\,\mathrm{Ai}(\varpi) + \mathcal{O}\left(\varpi^{-1/4}e^{-\frac{4}{3}\varpi^{3/2}}\right), \quad \varpi\to +\infty,
\end{equation}
where $\rho$ is given in \eqref{eq:rrrho}. The Painlev\'{e} II parameter $\varpi$ is defined by
\begin{align}
	\varpi = i\frac{\overline{r(\bar{z}^{\mathrm{R}}_{j})}}{1+|r(z_{j})|^2}\delta^2(z_{j}),
\end{align}
and the local variable $\lambda$ is as in \eqref{eq:trans-z-k}.
		\item For the even genus $n = 2s$ $(s \in \mathbb{N})$ background solution, when the distribution of stationary phase points $\kappa^{\mathrm{R}}_{j}$ ($j=1,2$) is given by equations \eqref{eq:zero-j} and \eqref{eq:theta_prime}, consider the regime $|\xi|<d$ with fixed $d>0$. Then
		\begin{align}
			q(x,t) = 2e^{2i(f_0x+g_0t)}e^{2(\ln\delta(\infty)+ig(\infty))\sigma_3}\left(\sum_{j=1,2}\frac{\beta_{12}(\kappa^{\mathrm{R}}_j)}{2\sqrt{t}(z-\kappa^{\mathrm{R}}_{j})\psi_{\kappa^{\mathrm{R}}_{j}}(z)} + \hat{Q}(x,t)\right) + \mathcal{O}\left(\frac{\ln t}{t}\right),
		\end{align}
		where $\delta(\infty)$ is defined in \eqref{delt-ing}, $g(\infty)$ in \eqref{ginfty}, $\hat{Q}(x,t)$ in \eqref{eq:Nglo-solution1}, $\beta_{12}(\kappa^{\mathrm{R}}_j)$ in \eqref{beta-12-21}, and $\psi_{\kappa^{\mathrm{R}}_{j}}(z)$ in \eqref{eq:dsdss} and \eqref{eq:dsds}.
	\end{enumerate}
\end{theorem}
\begin{assumption}\label{assumption-1-1}
Our results rely on several assumptions, summarized below for the reader's convenience.
	\begin{enumerate}
		\item   Throughout the paper, we assume that the initial data contains no solitons.
		\item For some $C>0$, the initial data $q_0(x)$ is assumed to be smooth and satisfies
		\begin{equation}\label{eq:finite_perturbation}
			q_0(x) = q^{alg}(x,0) \quad \text{for} \quad \pm x > C.
		\end{equation}
	\end{enumerate}
\end{assumption}
\subsection{Notation}
We introduce some notations that will be used in this paper:
\begin{enumerate}
	\item As usual, the three Pauli matrices are defined by
	\begin{equation*}
		\sigma_1=\begin{pmatrix}
			0&1\\
			1&0
		\end{pmatrix},\quad \sigma_2=\begin{pmatrix}
			0&-i\\
			i&0
		\end{pmatrix},\quad \sigma_3=\begin{pmatrix}
			1&0\\
			0&-1
		\end{pmatrix}.
	\end{equation*}
	\item $\hat{\sigma}_3$ acts on a matrix $A$ by $e^{\hat{\sigma}_3} A=e^{\sigma_3} A e^{-\sigma_3}$.

	\item The Cauchy operator is defined by
	\begin{align}\label{C_V_E}
		C _\Sigma f(z) = \frac{1}{2\pi i} \int_{\Sigma} \frac{f(s)}{s-z} \, ds, \quad z \in \mathbb{C}\setminus\Sigma.
	\end{align}
\item The notation $D_\rho\left(z_0\right):=\left\{z:\left|z-z_0\right|<\rho\right\}$ represents the open disk of radius $\rho$ centered at $z_0$ in the complex plane.
	\item For a real parameter $a$, we define the branch logarithm $\ln _a z$ with branch cut placed along the ray $\arg z=a$. Explicitly,
	$$
	\ln _a z=\ln |z|+i \arg z, \quad \arg z \in(a, a+2 \pi] .
	$$
	In particular, the standard principal branch $\ln z$ corresponds to the choice $a=-\pi$, i.e.,
	$$
	\ln _{-\pi} z=\ln z .
	$$
\end{enumerate}
\subsection{Plan of Proof for Theorem \ref{theorem-1}}
The remainder of this paper is organized as follows. In Section~\ref{sec-2}, we provide a rigorous definition of the algebro-geometric solution $q^{alg}(x,t)$ and formulate the RH problem~\ref{RH2-2} characterizing the Cauchy problem \eqref{eq:fNLS}--\eqref{eq:q_0} via the spectral analysis of its Lax pair. In Section~\ref{sec-3}, we focus on the odd genus case. Utilizing the properties of the phase function $\theta(z)$ and the distribution of stationary phase points, we employ the nonlinear steepest descent method to investigate the Painlev\'{e} asymptotics of $q(x,t)$ in the transition regime where two real stationary phase points coalesce. In Section~\ref{sec-4}, we address the even genus case. By analyzing the admissible configurations of branch cuts $\Gamma$ and their influence on the long-time asymptotics, we determine the asymptotic behavior of the solution in the vicinity of $\xi=0$ by employing the same methodological framework.
\section{The Riemann-Hilbert Formulism}\label{sec-2}

In view of the central role of the RH problem in the inverse scattering method,  we
 try to characterize the algebro-geometric background solution  in terms of the solutions of appropriate RH problems.

\subsection{The algebro-geometric background solution}
The focusing NLS equation \eqref{eq:fNLS} admits the following Lax pair \cite{BL21}:
\begin{subequations}
\begin{align}\label{eq:Lax-x}
	&	\Phi_x+i z \sigma_3 \Phi  =Q\Phi, \\
	&\Phi_t+2 i z^2 \sigma_3 \Phi  =\tilde{Q}\Phi,
\end{align}
\end{subequations}
where $\Phi=\Phi(x, t, k)$ is a $2 \times 2$ matrix-valued function, $z \in \mathbb{C}$ is a spectral parameter and
\begin{align}
	Q =\begin{pmatrix}
		0 & q(x, t) \\
		-\bar{q}(x, t) & 0
	\end{pmatrix}, \quad \tilde{Q}=\begin{pmatrix}
	    i|q(x,t)|^2&2zq(x,t)+iq_x(x,t)\\
        -2z\bar q(x,t)+i\bar q(x,t)&-i|q(x,t)|^2
	\end{pmatrix}.
\end{align}

Let a set of points $\left\{E_k, \bar E_k\right\}_{k=0}^n$ in the complex plane be given, and let $\mathcal{X}$ be the Riemann surface of genus $n$ defined by the equation
$$
w^2=P(z)=\prod_{k=0}^n\left(z-E_k\right)\left(z-\bar{E}_k\right) \equiv z^{2(n+1)}+P_{2 n+1} z^{2 n+1}+P_{2 n} z^{2 n}+\ldots . .+P_0,
$$
with cuts along the arcs
$\Gamma=\cup_k\Gamma_k, \ \Gamma_k=\left(\bar E _k, E_k\right)$ connecting $\bar E_k$ with $ E_k$,
where $E_k=B_k+iA_k, k\in\mathcal{K}$ with oriented upwards. The correspond Riemann surface   is introduce in Appendix \ref{APP-A}

Let us introduce the background eigenfunction $\Phi^{alg}(z)$, which is a solution of the form   \cite{KS17}
$$
\Phi^{alg}(z)=e^{\left(i f_0 x+i g_0 t\right) \sigma_3} M^{alg}(z) e^{-(i f(z) x+ig(z) t) \sigma_3},
$$
where $ f_0$, $ g_0$, $f(z), g(z)$  and $M^{alg}(z)$   can  be determined as follows.

\begin{problem}$f(z)$ and $g(z)$ satisfy the following properties.

\begin{enumerate}
\item
$f(z)$ and $g(z)$ satisfy the jump conditions
\begin{align}\label{f+g+add}
	f_{+}(z)+f_{-}(z)=C_k^f, \quad
	 g_{+}(z)+g_{-}(z)=C_k^g,\ \ z\in\Gamma_k.
\end{align}
with some real constants $C_k^f,C_k^g$.
\item $f(z)$ and $g(z)$ are analytic in $\mathbb{C}\setminus\Gamma$ and such that
\begin{align}\label{properties--11}
	f(z)=z+f_0+O(1 / z),\quad g(z)=2 z^2+g_0+O(1 / z).
\end{align}

\item $f(z)$ and $g(z)$ have to satisfy the symmetry condition
\begin{equation}\label{eq:symmetry-fg}
    \overline{f(\bar z)}=f(z),\quad \overline{g(\bar z)}=g(z).\end{equation}
\end{enumerate}
\end{problem}
This  RH problem  admits a solution   of  the form
\begin{align}\label{eq:f(z)}
	f(z)&=\int_{\bar{E}_0}^z\frac{\hat{f}(s)}{w(s)}\,ds=\int_{\bar{E}_0}^z \frac{s^{n+1}+\hat{f}_n s^n+\hat{f}_{n-1} s^{n-1}+\cdots+\hat{f}_0}{w(s)} \,d s,
	\\
	g(z)&=\int_{\bar{E}_0}^z \frac{\hat{g}(s)}{w(s)}\,ds=\int_{\bar{E}_0}^z \frac{4 s^{n+2}+\hat{g}_{n+1} s^{n+1}+\hat{g}_n s^n+\cdots+\hat{g}_0}{w(s)} \,d s,
\end{align}
 where  $\hat{f}_0, \hat{f}_1, \ldots, \hat{f}_{n-1}$ and $\hat{g}_0, \hat{g}_1, \ldots, \hat{g}_{n-1}$ are   determined by the normalization conditions on $f(z)$ and $g(z)$:
	\begin{equation}\label{eq:f-g-contour}
\int_{E_k}^{\bar{E}_k} d f=0 ,\quad\int_{E_k}^{\bar{E}_k}dg=0,\quad k=0,1, \ldots,n  .
\end{equation}
which   yields
	\begin{equation}\label{hat-ff}
	\hat{f}_n=\frac{P_{2 n+1}}{2},\qquad
	\hat{g}_{n+1}=6 P_{2 n+1}, \qquad \hat{g}_n=2 P_{2 n}-\frac{3}{2} P_{2 n+1}^2.
	\end{equation}
	In the standard construction of finite gap solutions, it is typically required that $\hat{f}(z)$ and $\hat{g}(z)$ have real coefficients to ensure that the resulting solutions are real valued.
 Also note that ${\rm Im}\,f(E_k)={\rm Im}\, f(\bar E_k)=0$, one has $\operatorname{Im} f_{+}(z)=-\operatorname{Im} f_{-}(z)>0$ for $z\in\Gamma$.

The background solution $M^{alg}(z)$ satisfying following RH problem.
\begin{problem}
Construct a meromorphic function $M^{alg}(z):\mathbb{C} \setminus\Gamma$ such that:
	\begin{enumerate}
		\item $M^{alg}(z)=I+\mathcal{O}(z^{-1}),$\quad $|z|\to\infty$.
		\item 	For each $z\in\Gamma$, the boundary values $M_\pm^{alg}(z)$ satisfy the jump relation
		$$M_+^{alg}(z)=M_-^{alg}(z)J_k^{alg}(z),\quad z\in\Gamma_k,$$
		where
	$$
J^{alg}_k(z)=\begin{pmatrix}
	0 & i e^{-2i\pi c_k} \\
	i e^{2i\pi c_k} & 0
\end{pmatrix}=e^{-i t\theta(z)}
\begin{pmatrix}
	0 & i e^{-i\phi_k} \\
	i e^{i\phi_k} & 0
\end{pmatrix}e^{i t\theta(z)},
$$
and
	\begin{equation}\label{eq:theta(z)}
		\theta(z)=\theta(z ; \xi):=\left(f(z)-f_0\right) \xi+\left(g(z)-g_0\right), \quad \xi=\frac{x}{t} ,
		\end{equation}
where $ c_0=0$ and $c=\left\{c_k \mid c_k= \frac{x\left(C_k^f-2 f_0\right)+t\left(C_k^g-2 g_0\right)+\phi_k}{2 \pi},\,1\leq k\leq n\right\}$
and $\phi_k$ is real constant.
		\item The function $M^{alg}(z)$ has singularities at the endpoints of\,\,\,$\Gamma_k$ of order at most $|z-E_k|^{-1/4}$ or $|z-\bar{E}_k|^{-1/4}$.
	\end{enumerate}
\end{problem}
Solving this RH problem yields  $M^{alg}(z)$ as follows
$$
\begin{aligned}
	M^{alg}(z)= & \frac{1}{2}\left(\begin{array}{cc}
		\frac{1}{\Lambda_{11}(\infty)} & 0 \\
		0 & \frac{1}{\Lambda_{22}(\infty)}
	\end{array}\right)  \times\left(\begin{array}{cc}
		\left(\nu(z)+\nu^{-1}(z)\right) \Lambda_{11}(z) & \left(\nu(z)-\nu^{-1}(z)\right) \Lambda_{12}(z) \\
		\left(\nu(z)-\nu^{-1}(z)\right) \Lambda_{21}(z) & \left(\nu(z)+\nu^{-1}(z)\right) \Lambda_{22}(z)
	\end{array}\right) .
\end{aligned}
$$
where
\begin{equation}\label{eq:kappa}
\nu(z)=\sqrt[4]{\prod_{k=0}^n \frac{z-E_k}{z-\bar{E}_k}}, \quad z \in \mathbb{C} \backslash \Gamma,
\end{equation}
and for  $s=1,2$, we have
$$
\Lambda_{s1}(z)=\frac{\Theta\left(\varphi(z)+c+d_s\right)}{\Theta\left(\varphi(z)+d_s\right)}, \quad \Lambda_{s2}(z)=\frac{\Theta\left(-\varphi(z)+c+d_s\right)}{\Theta\left(-\varphi(z)+d_s\right)}, \quad z \in \mathbb{C} \backslash \Gamma.
$$
Here, $\Theta(z)$ is the Riemann theta function defined on the genus-$n$ Riemann surface $\mathcal{X}$  associated with the hyperelliptic curve \eqref{genus-surface}. The Abel map $\varphi: \mathcal{X} \to \mathbb{C}^n$ is given in \eqref{eq:abel}, where $d_1=\varphi(\mathcal{D})+K,d_2=-\varphi(\mathcal{D})-K$, the vector-valued Riemann constant $K \in \mathbb{C}^n$ and the pole divisor $\mathcal{D}$ are defined on $\mathcal{X}$.
\begin{remark}
In this paper, without loss of generality, we assume that all endpoints $B_k$ are equal for the cuts $\Gamma_k$.
	\end{remark}
\subsection{The basic RH problem}
We next consider the solutions of the Lax pair \eqref{eq:Lax-x} with $q(x,t)$ being the solution of Cauchy problem \eqref{eq:fNLS}--\eqref{eq:q_0}   with the finite-genus algebro geometric background $q^{alg}(x,t)$. Let
\begin{align}\nonumber
M^{\pm}(x,z)&=e^{i(f_0x+g_0t)\hat\sigma_3}M^{alg}(x,z)\\\label{eq:M-pm}
&+\int_{\pm\infty}^x \Phi^{alg}(x,z) \Phi^{alg}( s, z)^{-1} Q\left(q-q^{alg}\right)(s, t) M^{\pm}(s,z)e^{-i\left(f(z)-f_0\right)(s-x) \sigma_3} ds.
\end{align}
In particular, we construct the two Jost solutions given by
\[
\Phi^{\pm}(x,z)=M^{\pm}(x,z) e^{-i\left[\big(f(z)-f_0\big) x+\big(g(z)-g_0\big) t\right] \sigma_3},
\]
which constitute a pair of linearly independent solutions to the Lax pair \eqref{eq:Lax-x}. Furthermore, the following symmetry relation holds:
\[
M^{\pm}(x,z)=\sigma_2 \overline{M^{\pm}(x,\overline{z})}\sigma_2,\quad z\in\mathbb{R}.
\]
Thus, there exists a scattering matrix $S(z)$ independent of $(x,t)$ for $z\in\mathbb{R}\setminus\Gamma$ such that
\begin{equation}
	\Phi^{+}(x,z)=\Phi^{-}(x,z)S(z),\quad z\in\mathbb{R}\setminus\Gamma,
\end{equation}
or equivalently,
$$
M^{+}(x,z)=M^{-}(x,z) e^{-i\left(\left(f(z)-f_0\right) x+\left(g(z)-g_0\right) t\right) \hat{\sigma}_3} S(z),\quad z\in\mathbb{R}\setminus\Gamma,
$$
where
\begin{equation}
	S(z)=\begin{pmatrix}
		a(z)&-\overline{b(\bar z)}\\
b(z)        &\overline{a(\bar z)}
	\end{pmatrix},\quad {\rm det}\,S(z)=|a(z)|^2+|b(z)|^2=1,
\end{equation}
the coefficients $a(z)$ and $b(z)$ can be expressed as
$$
\begin{aligned}
	& a(z)=  \operatorname{det}\left(M_1^{+}(x,z), M_2^{-}(x,z)\right),\\[2pt]
& b(z)=-e^{-2 i\left(\left(f(z)-f_0\right) x+\left(g(z)-g_0\right) t\right)} \operatorname{det}\left(M_1^{+}(x,z), M_1^{-}(x,z)\right),
\end{aligned}
$$
where $M_1^\pm$ and $M_2^\pm$ denote  the first and second  columns of the matrix $M^\pm$ respectively.    The reflection coefficient $r(z)$ is defined by
\begin{equation}
	r(z)=-\frac{b(z)}{a\overline{(\bar z)}}.
\end{equation}
For $z\in\Gamma_k$, 
  the boundary values of $M^{\pm}(x,z)$  on the contour $\Gamma_k$ satisfy either of the following jump conditions:
\begin{equation}\label{eq:scatter-gamma}
	\begin{cases}
		M^+_\pm(x,z)=M^-(x,z)\tilde S_\pm(z),\\
		M^+(x,z)=M^-_\pm(x,z)\tilde S_\pm(z).
	\end{cases}
\end{equation}
Here the scattering matrices $\tilde S_\pm(z)$   take the form
\begin{equation}\label{eq:S-matrix}
	\tilde{S}_\pm(z)=
	\begin{pmatrix}
		a_{\pm}(z) & -\overline{b_{\pm}(\overline{z})}\\[4pt]
		b_{\pm}(z) & \overline{a_{\pm}(\overline{z})}
	\end{pmatrix}, \qquad\text{with}\qquad \det\tilde S_\pm(z)=1.
\end{equation}
The existence, analyticity and symmetries of the Jost functions and scattering data cane proven directly.
Here we list their properties as follows \cite{FL26,MAB22}.


\begin{proposition}  The Jost functions and scattering data admit the following properties.
	\begin{enumerate}
		\item  The functions $M_{1}^{-}(x,z)$, $M_{2}^{+}(x,z)$
and $\overline{a(\bar{z})}$ (resp. $M_{1}^{+}(x,z)$, $M_{2}^{-}(x,z)$ and $a(z)$) admit analytic continuation to
\[
\Omega^{+}:=\mathbb{C}^{+}\cap\{\operatorname{Im}(f(z)-f_{0})>0\}\setminus\Gamma \quad \left(\text{resp. }\Omega^{-}:=\mathbb{C}^{-}\cap\{\operatorname{Im}(f(z)-f_{0})<0\}\setminus\Gamma\right),
\]
possessing continuous boundary values on $\mathbb{R}\setminus\Gamma$ and satisfying the normalization
$$\begin{pmatrix} M_{1}^{\mp}(z),  M_{2}^{\pm}(z) \end{pmatrix}\to I, \ \ z\to\infty \ \ {\rm in} \ \ \ \mathbb{C}^{\pm}.$$

 \item As $z \rightarrow \infty$, we have
		 $$
		 a(z)=1+\mathcal{O}\left(z^{-1}\right),\quad z\in\mathbb{C}^-\cup\mathbb{R},\quad
		 \quad b(z),r(z)=\mathcal{O}\left(z^{-1}\right),\quad z\in\mathbb{R}
		 $$

	\item  	
	For $M^+_\pm(x,z)=M^-(x,z)\tilde{S}_+(z)$ on each branch cut $z\in\Gamma_k$, one has
	\begin{equation}
		a_+(z)=-ie^{ic_k}\overline{b_-(\overline{z})},\quad
		b_+(z)= ie^{ic_k}\overline{a_+(\overline{z})},
	\end{equation}
	along with
	\begin{equation}\label{eq:tilder--1}
		r_{+}(z)-r_{-}(z)=\dfrac{ic_k}{\overline{a_{+}(\overline{z})}\,\overline{a_{-}(\overline{z})}}.
	\end{equation}
	The jump matrix of $N(z)$ on $\Gamma_k$, denoted $V_k^\dagger(z)$, is given by
	\begin{equation}
		V^\dagger_k(z)=
		\begin{pmatrix}
			\frac{\overline{a_-(\bar{z})}}{\overline{a_+(\bar{z})}} & ie^{-2\pi ic_k}\\
			0 & \frac{\overline{a_+(\bar{z})}}{\overline{a_-(\bar{z})}}
		\end{pmatrix},\quad z \in \Gamma_k \cap \mathbb{C}^{+},
		\quad \begin{pmatrix}
			\frac{{a_+({z})}}{{a_-({z})}} & 0\\
			ie^{2\pi ic_k} & \frac{{a_-({z})}}{{a_+({z})}}
		\end{pmatrix},\quad z \in \Gamma_k \cap \mathbb{C}^{-}.
	\end{equation}
	For $M^+(x,z)=M^-_\pm(x,z)\tilde{S}_-(z)$, we instead have
	\begin{equation}
		a_+(z)=-ie^{-ic_k}b_-(z),\quad
		b_+(z)= -ie^{ic_k}a_-(z),
	\end{equation}
	and the relation
	\begin{equation}\label{eq:tilder--2}
		\frac{r_+(z)}{1+|r_+(z)|^2}-\frac{r_-(z)}{1+|r_-(z)|^2}= ie^{-i c_k}.
	\end{equation}
	In this case, the jump matrix $V_k^\ddagger(z)$ reads
	\begin{equation}
		V^\ddagger_k(z)=
		\begin{pmatrix}
			1 & 0 \\
			\frac{ie^{2\pi i c_k}}{\overline{a_+(\overline{z})a_-(\overline{z})}} & 1
		\end{pmatrix},\quad z \in \Sigma_k \cap \mathbb{C}^{+},\quad \begin{pmatrix}
			1 & \frac{i e^{-2\pi i c_k}}{a_{+}(z) a_{-}(z)} \\
			0 & 1
		\end{pmatrix},\quad z \in \Sigma_k \cap \mathbb{C}^{-},
	\end{equation}
	\end{enumerate}
\end{proposition}

As in the case of zero background, the analytic and asymptotic properties of $M^\pm(z)$ suggest that we introduce the 2$\times$2 matrix valued function   by
\begin{equation}\label{eq:RH-con}
	 N(z)=N(x,z):= \begin{cases}\left(\frac{M_1^{-}(z)}{\overline{a(\bar z)}}\,\,\,M_2^{+}(z)\right), & z \in \mathbb{C}^{+}, \\
			\left(M_1^{+}(z)\,\,\, \frac{M_2^{-}(z)}{a(z)}\right), & z \in \mathbb{C}^{-},\end{cases}
\end{equation}
which  satisfies the following RH problem.
	\begin{problem}\label{RH2-2}
Construct a meromorphic function $N(z):\mathbb{C}\setminus (\mathbb{R}\cup\Gamma)$such that:
	\begin{enumerate}
		\item $N(z)=I+\mathcal{O}(z^{-1}),$\quad $|z|\to\infty$.
		\item 	For each $z\in\mathbb{R}\,\cup\,\Gamma$, the boundary values $N_\pm(z)$ satisfy the jump relation
		$$N_+(z)=N_-(z)J(z).$$
		where
		\begin{equation}\label{eq:jummp-2}
		J(z)=
		\begin{cases}
			\displaystyle
			\begin{pmatrix}
				1+|r(z)|^2 & \overline{r(\bar{z})} e^{-2i t \theta(z)}\\[4pt]
				r(z) e^{2i t \theta(z)} & 1
			\end{pmatrix},
			\quad z\in \mathbb{R}\setminus\Gamma,\\[6pt]
	V^\dagger(z), \quad z\in\Gamma_k\cap\mathbb{H}^-,\quad V^\ddagger(z),z\in\Gamma_k\cap\mathbb{H}^+,
		\end{cases}
	\end{equation}
where $\mathbb{H}^{+}=\{\operatorname{Re} z>0\}$ and $\mathbb{H}^{-}=\{\operatorname{Re} z<0\}$ denote the open right and left half-planes, respectively.
		\item The function $N(z)$ has singularities at the endpoints of\,\,$\Gamma_k$ of order at most $|z-E_k|^{-1/4}$ or $|z-\bar{E}_k|^{-1/4}$.
	\end{enumerate}
\end{problem}
\begin{proof}
 	For $z\in\Gamma_k\in\mathbb{H}^-$, setting $x=t=0$, one has
	\[
	S_\pm(z)=\big(M^-(z)\big)^{-1}M_\pm^+(z).
	\]
	Combining this relation with \eqref{eq:M-pm}, we arrive at
	\begin{equation}
		S_+(z)=S_-(z)
		\begin{pmatrix}
			0 & i e^{-\phi_k} \\
			i e^{\phi_k} & 0
		\end{pmatrix}.
	\end{equation}
	By virtue of \eqref{eq:RH-con}, the jump relation across $\Gamma\cup\mathbb{C}^+$ at $x=t=0$ reads
	\begin{equation}
		\left(\frac{M_1^{-}(z)}{\overline{a_+(\bar{z})}} \quad M_{2+}^{+}(z)\right)
		=\left(\frac{M_1^{-}(z)}{\overline{a_-(\bar{z})}} \quad M_{2-}^{+}(z)\right)
		\begin{pmatrix}
			\frac{\overline{a_-(\bar{z})}}{\overline{a_+(\bar{z})}} & \gamma(z) \\
			0 & \dfrac{\overline{a_+(\bar{z})}}{\overline{a_-(\bar{z})}}
		\end{pmatrix}
	\end{equation}
	for some scalar-valued function $\gamma(z)$. Consequently,
	\begin{equation}\label{eq:c_2-undefined}
		\frac{M_{2+}^{+}(z)}{\overline{a_+(\bar{z})}}-\frac{M_{2-}^{+}(z)}{\overline{a_-(\bar{z})}}
		=\frac{\gamma(z)}{\overline{a_+(\bar{z})a_-(\bar{z})}}\, M_1^-(z).
	\end{equation}
	We now determine the explicit form of $\gamma(z)$. From the scattering relation \eqref{eq:scatter-gamma}, we write
	\begin{equation}\label{M-s}
		M_{2\pm}^+(z)=s_{12,\pm}\,M_1^-(z)+s_{22,\pm}\,M_2^-(z).
	\end{equation}
	Since $\det M^\pm(z)=1$, it follows that
	\begin{equation}\label{M-ss}
		s_{22,\pm}=\det\big(M_1^-(z),\,M_{2\pm}^+(z)\big)=\overline{a_\pm(\overline{z})},\quad z\in\Gamma\cap\mathbb{C}^+.
	\end{equation}
	Substituting \eqref{M-s} and \eqref{M-ss} yields again
	\[
	\frac{M_{2+}^{+}(z)}{\overline{a_+(\bar{z})}}-\frac{M_{2-}^{+}(z)}{\overline{a_-(\bar{z})}}
	=\frac{\gamma(z)}{\overline{a_+(\bar{z})a_-(\bar{z})}}\, M_1^-(z).
	\]
	Comparing this identity with \eqref{eq:c_2-undefined} and using the structure from the scattering matrix factorization, we conclude $\gamma(z)=ie^{-\phi_k}$. Using the same method, we can prove that for $z\in\Gamma_k\cap\mathbb{H}^+$. This completes the proof.
\end{proof}
The solution $q(x,t)$ of the initial value problem \eqref{eq:fNLS}-\eqref{eq:q_0} can be expressed in terms of the solution of the   RH problem \ref{RH2-2} as follows:
$$
q(x,t)=2ie^{2i(f_0x+g_0t)}\lim _{z \rightarrow \infty} z N_{12}(z).
$$
\subsection{The distribution of stationary phase points}
For large $t$, the asymptotic behavior of the jump matrix varies across distinct regions of the $(x,t)$-plane according to the sign of $\operatorname{Im}\theta(z)$. This sign structure is determined by the spectral parameter $\xi = x/t$ through the branch structure of $\theta(z)$ defined in \eqref{eq:theta(z)}.

The function $\theta(z)$ possesses branch cuts along the contours $\Gamma_k$ and satisfies the Schwarz reflection symmetry $\theta(z)=\overline{\theta(\bar z)}$. Consequently,   the endpoint values satisfy
\begin{equation}\label{theta-tt}
\operatorname{Im}f(E_k) = \operatorname{Im}g(E_k) = 0, \quad\operatorname{Im} \theta\left(E_k\right)=\operatorname{Im} \theta\left(\bar{E}_k\right)=0,
\end{equation}
implying that $f(E_k)$, $g(E_k)$, and $\theta(E_k)$ are all real, and the boundary values on the cuts satisfy the jump relation $\operatorname{Im} \theta_{+}(z)=-\operatorname{Im} \theta_{-}(z)>0$ for $z \in \Gamma$.

We give the asymtptotic as $z\to\infty$ as follows
\begin{align}\label{eq:large-theta}
		\theta(z)\sim 2z^2+\xi z+\theta_c(\infty,\xi)+\mathcal{O}(z^{-1}),
 \end{align}
where
\begin{equation}\theta_c(\infty,\xi)=\int_{\bar E_0}^{\infty}\left(\frac{h(z)}{w(z)}-4z-\xi\right) d z+2 A_0^2-2 B_0^2-(2 B_0+f_0)\xi-g_0
\end{equation}
is a real function of $\xi$.
These properties imply that the level set $\operatorname{Im}\theta(z)=0$ possesses two infinite branches: the real axis itself, and a secondary branch asymptotics to the vertical line $\operatorname{Re}z=-\frac{1}{4}\xi$. Henceforth, the term ``infinite branch'' refers exclusively to this latter component. The intersection points of the real axis with infinite branche is termed the real stationary phase point. Broadly speaking, this paper investigates how distinct configurations of the real stationary phase point give rise to qualitatively different asymptotic regimes.

 Differentiating $\theta(z)$ defined in \eqref{eq:theta(z)} yields
\begin{equation}
	\theta'(z) = \xi f'(z) + g'(z) = \frac{\xi \hat{f}(z) + \hat{g}(z)}{w(z)},
\end{equation}
hence, stationary phase  points satisfy
\begin{equation}\label{eq:primethet}
	  h(z;\xi) \equiv  \xi \hat{f}(z) + \hat{g}(z) = 0.
\end{equation}
Factoring $h(z;\xi)$, we obtain
\begin{equation}\label{eq:zero-j}
	h(z;\xi) =\prod_{l=1}^{r}(z-\kappa_l^{\mathrm{R}}(\xi))\prod_{s=1}^{\frac{n+2-r}{2}}(z-\kappa_s^{\mathrm{C}}(\xi))(z-\bar{\kappa}_s^{\mathrm{C}}(\xi)),\quad r\leq n+2,
\end{equation}
where $\kappa_l^{\mathrm{R}}(\xi)\in\mathbb{R}$ $(l\in\mathcal{L}=\{1,\dots,r\})$ and $\kappa_s^{\mathrm{C}}(\xi)\in\mathbb{C}^+\setminus\mathbb{R}$ denote the real and complex stationary phase points, respectively. Depending on the value of $\xi$, these stationary phase points may collide on the real axis or bifurcate into the complex plane; see Figures~\ref{figure-intial}, \ref{fff}, and
\ref{FIguree45}.

For generic $\xi\in\mathbb{R}$, the number $r$ of real stationary phase points satisfies $r\in\{1,3,\dots,n+2\}$ if $n$ is odd, and $r\in\{2,4,\dots,n+2\}$ if $n$ is even. As $\xi\to\pm\infty$, the phase points approach the $n+1$ roots of $\hat{f}(z)=0$ together with "infinite branch". At $\xi=0$, they reduce to the $n+2$ roots of $\hat{g}(z)=0$. As $\xi$ varies, complex conjugate pairs may collide on or separate from the real axis; hence the phase points can only appear or disappear in such pairs. The following lemma provides the coalescence condition.

\begin{lemma}\label{lem:collision}
Let $z_j\in\mathbb{R}$ denote a collision point of two stationary phase points on the critical line $\xi=\xi_j$, which corresponds to the region $|\xi-\xi_j|t^{2/3}<C$ for $C>0$ in the $(x,t)$-plane, where two stationary phase points coalesce at $z_j$. Suppose the phase function
	\[
	\theta(z;\xi)=\bigl(f(z)-f_0\bigr)\xi+\bigl(g(z)-g_0\bigr),
	\]
	where ${f}(z)$ and ${g}(z)$ are given in \eqref{eq:f(z)}
	satisfies the condition
	\begin{equation}\label{eq:transversality}
\partial_{\xi}\theta(z_j;\xi_j)\neq 0.
	\end{equation}
	Then the collision at $z_j$ occurs if and only if the pair $(z_j,\xi_j)\in\mathbb{R}^2$ satisfies the system
	\begin{equation}\label{eq:interior}
		\theta'(z_j;\xi_j)=0,\qquad \theta''(z_j;\xi_j)=0,\qquad \theta'''(z_j;\xi_j)\neq 0,
	\end{equation}
 where primes denote $z$-derivatives. Equivalently, the critical line is given by
	\begin{equation}\label{eq:xireal}
		\xi_j=-\frac{g'(z_j)}{f'(z_j)}=-\frac{\hat{g}(z_j)}{\hat{f}(z_j)}=-\frac{4 z_j^{n+2}+\hat{g}_{n+1} z_j^{n+1}+\hat{g}_n z_j^n+\cdots+\hat{g}_0}{z_j^{n+1}+\hat{f}_n z_j^n+\hat{f}_{n-1} z_j^{n-1}+\cdots+\hat{f}_0},
	\end{equation}
where $\{\hat f_j\}_{j=1}^n$ and $\{\hat g_j\}_{j=1}^{n+1}$ are related to the finite-genus algebro-geometric solution given by \eqref{eq:f-g-contour},	then eliminating $\xi_j$ yields the Wronskian condition
	\begin{equation}\label{eq:wronskian}
		\hat{f}'(z_j)\hat{g}(z_j)-\hat{f}(z_j)\hat{g}'(z_j)=0.
	\end{equation}
	This algebraic equation is of degree $2n+2$; consequently, there exist exactly $2n+2$ collision points $z_j$, indexed by $j\in\{1,\dots,2n+2\}$.
\end{lemma}
 \begin{proof}
 	First, for the given collision point $z_j$ satisfying condition \eqref{eq:interior}, we prove that there exist two first-order stationary phase points $\kappa_1(\xi),\kappa_2(\xi)$ in its neighborhood.
 	
 At the critical line $\xi=\xi_j$, expanding $\theta(z;\xi_j)$ near $z=z_j$ gives
 \begin{equation*}
 	\theta(z;\xi_j)=\theta(z_j;\xi_j)+\frac{1}{6}\theta'''(z_j;\xi_j)(z-z_j)^3+\frac{1}{24}\theta^{(4)}(z_j;\xi_j)(z-z_j)^4+\mathcal{O}(z^5),
 \end{equation*}
 Since $\theta'(z_j;\xi_j)=\theta''(z_j;\xi_j)=0$ and $\theta'''(z_j;\xi_j)\neq 0$. Now consider a small perturbation $\xi=\xi_j+\delta_\xi$. Near $z=z_j$, the derivative expands as
 \begin{equation*}
 	\theta'(z;\xi)=\theta'(z_j;\xi_j)+\theta''(z_j;\xi_j)(z-z_j)+\frac{1}{2}\theta'''(z_j;\xi_j)(z-z_j)^2+\partial_\xi\theta(z_j;\xi_j)\delta_\xi+\mathcal{O}(z^3),
 \end{equation*}
 which at leading order reduces to
 \begin{equation*}
 	\theta'(z;\xi)\approx\frac{1}{2}\theta'''(z_j;\xi_j)(z-z_j)^2+\partial_{\xi}\theta(z_j;\xi_j)(z_j;\xi_j)\delta_\xi.
 \end{equation*}
 For $\delta_\xi\neq 0$, let the equation $\theta'(z;\xi)=0$ yields
 \begin{equation}\label{deri-heta}
 	(z-z_j)^2\approx-\frac{2}{\theta'''(z_j;\xi_j)}\partial_{\xi}\theta(z_j;\xi_j)(z_j;\xi_j)\delta_\xi.
 \end{equation}
 Thus, there exist two simple stationary phase points
 \begin{equation}\label{rrrr}
 	\kappa_l\approx z_j\pm\sqrt{-\frac{2}{\theta'''(z_j;\xi_j)}\partial_{\xi}\theta(z_j;\xi_j)\delta_\xi},\quad l=1,2,
 \end{equation}
 which coalesce at $z_j$ as $\delta_\xi\to 0$. Evaluating the second derivative \eqref{deri-heta} and using \eqref{rrrr} at $\kappa_{l}$ gives
 \begin{equation}
 	\theta^{\prime \prime}\left(\kappa_{ l} ; \xi\right) \approx \pm \theta^{\prime \prime \prime}\left(z_j ; \xi_j\right) \sqrt{-\frac{2}{\theta^{\prime \prime \prime}\left(z_j ; \xi_j\right)} \partial_{\xi}\theta(z_j;\xi_j)\delta_\xi} \neq 0,
 \end{equation}
 Combined with \eqref{eq:transversality}, this confirms that $\kappa_{l},\,l=1,2$ are nondegenerate stationary phase points.
 
Conversely, we prove that two stationary phase points $\kappa_1(\xi)$ and $\kappa_2(\xi)$ coalesce into a single point $z_j$, where $(z_j,\xi_j)$ satisfies \eqref{eq:interior}.
 
Given two simple stationary points $\kappa_1(\xi),\kappa_2(\xi)$ coalesce at $z_j$ as $\xi\to\xi_j$ satisfying
 \begin{equation}
 	\theta'(\kappa_l(\xi);\xi)=0,\quad \theta''(\kappa_l(\xi);\xi)\neq 0,\quad l=1,2.
 \end{equation}
 Since $\theta''(z,\xi)$ is continuous, it remains nonzero in a neighborhood of $(z_j,\xi_j)$. The implicit function theorem then guarantees that each $\kappa_l(\xi)$ extends uniquely as a smooth function with
 \begin{equation}
 	\frac{d\kappa_l}{d\xi}=-\frac{\partial_\xi\theta'(\kappa_l(\xi);\xi)}{\theta''(\kappa_l(\xi);\xi)},\quad l=1,2.
 \end{equation}
 However, if $\theta''(z_j;\xi_j)\neq 0$, the implicit function theorem yields a unique smooth curve of zeros $z(\xi)$ passing through $(z_j,\xi_j)$. This contradicts the existence of two distinct stationary phase points $\kappa_1(\xi),\kappa_2(\xi)$ both tending to $z_j$. Hence the second derivative must vanish at the coalescence point:
 \begin{equation}
 	\theta''(z_j;\xi_j)=0.
 \end{equation}
This completes the proof.
 \end{proof}
\begin{remark}
The expansion order in the two variables is dictated by a dominant-balance argument. Since $\theta''(z_j;\xi_j)=0$ by hypothesis, the linear term in $\delta_z:=z-z_j$ vanishes, making $(\delta_z)^2$ the first nontrivial spatial contribution. In the critical line direction, $\delta_\xi$ already provides the first nonzero correction. The balance relation $(\delta_z)^2\sim\delta_\xi$ (equivalently $\delta_z\sim\sqrt{\delta_\xi}$) then shows that keeping $(\delta_z)^2$ and $\delta_\xi$ captures the leading-order behavior, while higher-order terms are negligible as $\delta_\xi\to0$.
\end{remark}
\begin{remark}
	In the RH analysis of the Painlev\'e~II model in Appendix~\ref{App-B}, the standard phase function is
	\begin{equation}
		\theta(\lambda;s)=\frac{4}{3}\lambda^3+s\lambda.
	\end{equation}
	The stationary phase point  $\theta'(\lambda;s)=4\lambda^2+s=0$ yields
	\begin{equation}
		\lambda_{\pm}(s)=\pm\frac{\sqrt{-s}}{2},\quad s<0.
	\end{equation}
	As $s\to 0^-$, both $\lambda_{\pm}(s)$ converge to $0$. At the coalescence point $(\lambda,s)=(0,0)$,
	\begin{equation}
		\theta''(\lambda;s)=8\lambda\quad\Rightarrow\quad\theta''(0;0)=0.
	\end{equation}
	This verifies the general principle: coalescence necessarily implies the vanishing of the second derivative.\end{remark}
As $\xi$ approaches any critical line $\xi_j$, stationary phase points collide in pairs, admitting precisely two scenarios: two real stationary points may coalesce and depart from the real axis (so that $r \mapsto r-2$ ), or a pair of complex conjugate stationary points may emerge onto the real axis (so that $r \mapsto r+2$ ). Owing to the complex conjugation symmetry, isolated transitions with $r \mapsto r \pm 1$ cannot occur. In what follows, the notation $z_j$ is reserved specifically for real--real collisions, defined as the coalescence of two distinct real simple stationary phase points  at $z_j$ as $\xi \rightarrow \xi_j$.
 \begin{remark}
We investigate the distribution of the stationary phase points $\{\kappa_j\}_{j=1}^{n+2}$, where each $\kappa_j$ corresponds to either a real stationary phase point $\kappa_l^{\mathrm{R}}$ or a complex stationary phase point $\kappa_s^{\mathrm{C}}$ as defined in \eqref{eq:zero-j}, by analyzing the asymptotic and topological constraints governing their configuration.

To characterize these points, we examine the large-$z$ asymptotic expansion of the phase derivative. Expanding the denominator as $z \to \infty$, we obtain
\[
\frac{1}{\sqrt{\prod_{k=0}^{n}(z-E_k)(z-\bar{E}_k)}} = \frac{1}{z^{n+1}}\left(1+\frac{\mathcal{B}}{z}+\frac{\mathcal{B}^2+\sum_{k=0}^{n}(B_k^2-A_k^2)}{2z^2}+O(z^{-3})\right),
\]
where $\mathcal{B} \equiv \sum_{k=0}^{n} B_k$. Simultaneously, the numerator polynomial, whose roots constitute the stationary phase points, admits the expansion
\[
\prod_{l=1}^{n+2}(z-\kappa_j) = z^{n+2}\left[1-\frac{\mathcal{K}_1}{z}+\frac{\mathcal{K}_2}{z^2}+O(z^{-3})\right],
\]
with $\mathcal{K}_1 = \sum_{l=1}^{n+2} \kappa_j$ and $\mathcal{K}_2 = \sum_{1 \leq i < j \leq n+2} \kappa_i\kappa_j$ denoting the first and second elementary symmetric polynomials of the stationary phase points, respectively.

Combining these expansions yields the asymptotic representation of the phase derivative:
\[
\frac{d\theta(z)}{dz} = 4z + 4(\mathcal{B}-\mathcal{K}_1) + \frac{4\left[\mathcal{K}_2 - \mathcal{K}_1\mathcal{B} + \frac{1}{2}\left(\mathcal{B}^2+\sum_{k=0}^{n}(B_k^2-A_k^2)\right)\right]}{z} + O(z^{-2}).
\]
Comparing this expansion order-by-order with the required asymptotic behavior \eqref{eq:large-theta}, we derive the following constraints:

At order $O(z^0)$, matching the constant term yields the first moment condition:
	\begin{equation}
	4(\mathcal{B} - \mathcal{K}_1) = \xi \quad \Longrightarrow \quad \mathcal{K}_1= \mathcal{B} - \frac{1}{4}\xi.
	\label{eq:first_moment},
\end{equation}
which constrains the sum of all stationary phase points.

At order $O(z^{-1})$, the vanishing of the coefficient imposes the second moment condition:
	\begin{equation}
	\mathcal{K}_2 - \mathcal{K}_1\mathcal{B} + \frac{1}{2}\left(\mathcal{B}^2+\sum_{k=0}^{n}\left(B_k^2-A_k^2\right)\right) = 0.
	\label{eq:second_moment}
\end{equation}
These algebraic constraints are supplemented by $n$ homological conditions arising from the vanishing of periods over the $\alpha$-cycles:
\begin{equation}
	\oint_{\alpha_j} \frac{d\theta(z)}{dz}\, dz = 0, \qquad j = 1, \ldots, n.
	\label{eq:homology}
\end{equation}
Together, equations \eqref{eq:first_moment}, \eqref{eq:second_moment}, and \eqref{eq:homology} constitute a complete system of $(n+2)$ independent constraints that uniquely determine the locations of the $n+2$ stationary phase points $\{\kappa_j\}_{j=1}^{n+2}$.
\end{remark}
In the subsequent analysis, we examine the Painlev\'{e} asymptotics of $q(x, t)$ separately for the odd- and even-genus background solutions of the Cauchy problem  \eqref{eq:fNLS}--\eqref{eq:q_0},  which exhibit different asymptotic scenarios.
\section{Painlevé Asymptotics in Odd-genus Backgrounds}\label{sec-3}
In this section, we restrict our attention to the case of odd genus. We perform the Painlev\'{e} asymptotics
analysis of the RH problem~\ref{RH2-2} for $N(z)$ to derive the asymptotic
behavior of $q(x,t)$ in the transition region $|\xi-\xi_j|t^{2/3} < C$.

For $\xi \gg 0$, we assume the stationary phase points are determined by
\begin{equation}\label{h-z-xi}
	h(z;\xi) = 4(z-\kappa_1^{\mathrm{R}})\prod_{s=1}^{\frac{n+1}{2}}\left[(z-\kappa_s^{\mathrm{C}})(z-\bar{\kappa}_s^{\mathrm{C}})\right],
\end{equation}
where $\kappa_s^{\mathrm{C}}$ and $\bar{\kappa}_s^{\mathrm{C}}$ denote
complex conjugate stationary phase points, and $\kappa_1^{\mathrm{R}}$ asymptotic to the infinite branch
denotes as shown in Figure~\ref{figure-intial}. As $\xi$ decreases, the conjugate pair $\{\kappa^{\mathrm{C}}_1,
\bar{\kappa}^{\mathrm{C}}_1\}$ coalesces toward the real axis. This collision produces
two distinct real stationary phase points, denoted by $\kappa^{\mathrm{R}}_2$ and $\kappa^{\mathrm{R}}_3$,
which then move along $\mathbb{R}$ in opposite directions, transforming \eqref{h-z-xi} into
\begin{equation}\label{h-pp}
	h(z;\xi) = 4(z-\kappa_1^{\mathrm{R}})
    (z-\kappa_2^{\mathrm{R}})(z-\kappa_3^{\mathrm{R}})\prod_{s=2}^{\frac{n-1}{2}}\left[(z-\kappa_s^{\mathrm{C}})(z-\bar{\kappa}_s^{\mathrm{C}})\right],
\end{equation}
where $\kappa^{\mathrm{R}}_1$ subsequently moves to the right until it collides
with $\kappa^{\mathrm{R}}_2$, merging at $z_1=G(z_1)$, this collision occurs at $\xi=\xi_1$, which describes the distribution of stationary phase points addressed in this section.
\begin{figure}[htp]
	\begin{center}
		\begin{tikzpicture}[scale=0.7]
			\fill[orange!50]
			(-3.5,2)                          
			to[out=-165, in=-50] (-6.5,2)
			--(-6.5,0)--(-3.5,0)-- cycle;

				\draw[dashed] 	(-3.5,2)                          
				to[out=-165, in=-50] (-6.5,2) ;

	\draw[dashed] 	(-3.5,-2)                          
to[out=--165, in=50] (-6.5,-2) ;

			\fill[orange!50]
			(-11,2.5)--(-10,2.5)                          
			to[out=-90, in=100] (-9.5,0)       
			--(-11,0)-- cycle;		
			
				\draw[dashed] 	(-10,2.5)                          
				to[out=-90, in=100] (-9.5,0) ;
				
					\draw[dashed] 	(-10,-2.5)                          
				to[out=90, in=-100] (-9.5,0) ;

			\fill[orange!50]
			(-10,-2.5)                          
			to[out=90, in=-100] (-9.5,0)       
			--(-6.5,0)--(-6.5,-2.5)--
			cycle;

				\fill[orange!50]
			(-3.5,-2.5)-- (-3.5,-2)
			to[out=165, in=50] (-6.5,-2)--(-6.5,-2.5)--(-3.5,-2.5) --cycle;

			\fill[orange!50]
			(1,2.5)--(2,2.5)                          
			to[out=-90, in=100] (2.5,0)       
			--(1,0)-- cycle;

					\draw[dashed] 	(2,2.5)                          
					to[out=-90, in=100] (2.5,0) ;
			
			\draw[dashed] 	(2,-2.5)                          
		to[out=90, in=-100] (2.5,0) ;

			\fill[orange!50]
			(2,-2.5)                          
			to[out=90, in=-100] (2.5,0)       
			--(4,0)--(4,-2.5)-- cycle;

		\fill[orange!50]                      
			(-1,0)--(-1,-2.5)--(-3.5,-2.5)--(-3.5,0)--
			cycle;	
			\draw[thick,dashed](-10.2,-2.5)--(-10.2,2.5);	
			\draw[thick,dashed](1.8,-2.5)--(1.8,2.5);	
			\node  at (-6.4,-2.5)  { \footnotesize$\bar E_0$};
			\node  at (-3.5,-2.5)  { \footnotesize$\bar E_1$};	
			
			\node  at (-6.4,2.3)  { \footnotesize$ E_0$};
			\node  at (-3.5,2.3)  { \footnotesize$ E_1$};	
				\node  at (-5.3,2)  { \footnotesize$\kappa^{\mathrm{C}}_1$};	
				\node  at (-5.3,-2)  { \footnotesize$\bar\kappa^{\mathrm{C}}                _1$};	
					\node  at (-9,-0.5)  { \footnotesize$\kappa^{\mathrm{R}}                    _1$};	
					
				\node  at (-5,0.7)  { \footnotesize${\rm Im}\,\theta(z)<0$};	
	
							\node  at (-8,0.7)  { \footnotesize${\rm Im}\,\theta(z)>0$};	
			
			\node  at (-10,-2.9)  { \footnotesize${\rm Re}\,z=-\frac{1}{4}\xi$};
			\node  at (2,-2.9){ \footnotesize${\rm Re}\,z=-\frac{1}{4}\xi$};
			\draw[thick,->](-11,0)--(-1,0);
			\draw[thick,->](1,0)--(11,0);			
			\draw[thick,](-3.5,-2)--(-3.5,2);	
			\draw[thick](-6.5,-2)--(-6.5,2);

				\draw[thick](5.5,-1.5)--(5.5,1.5);
				\draw[thick](4.,-1.5)--(4.,1.5);

							\fill[orange!50]
				(4,1.5)                          
				to[out=-45, in=-145] (5.5,1.5)       
				--(5.5,0)--(4,0)-- cycle;

				\fill[orange!50]
				(4,-2.5) --	(4,-1.5)                          
				to[out=45, in=145] (5.5,-1.5) --(5.5,0)
			 --(8,0)--(8,-1.5) to [out=45, in=145]	
		(9.5,-1.5)--(9.5,0)--(11,0)--
			(11,-2.5)--(4,-2.5)--cycle;	
				\draw[dashed] 	(4,-1.5)                          
				to[out=45, in=145] (5.5,-1.5) ;
				
					\draw[dashed] 	(8,-1.5)                          
				to[out=45, in=145] (9.5,-1.5) ;
				
					\draw[dashed] 	(8,1.5)                          
				to[out=-45, in=-145] (9.5,1.5) ;
				
					\draw[dashed] 	(4,1.5)                          
				to[out=-45, in=-145] (5.5,1.5) ;

				\node  at (4.1,-2)  { \footnotesize$\bar E_0$};
					\node  at (4.1,1.8)  { \footnotesize$ E_0$};
				
					\node  at (5.5,-2)  { \footnotesize$\bar E_1$};
				\node  at (5.5,1.8)  { \footnotesize$ E_1$};
				
					\node  at (8.1,-2)  { \footnotesize$\bar E_2$};
				\node  at (8.1,1.8)  { \footnotesize$ E_2$};

					\node  at (8.1,-2)  { \footnotesize$\bar E_2$};
				\node  at (8.1,1.8)  { \footnotesize$ E_2$};

					\node  at (9.5,-2)  { \footnotesize$\bar E_3$};
				\node  at (9.5,1.8)  { \footnotesize$ E_3$};
				
		\node  at (8.8,1.6)  { \footnotesize$ \kappa^{\mathrm{C}}_2$};
					\node  at (4.8,1.6)  { \footnotesize$ \kappa^{\mathrm{C}}_1$};
			
				\node  at (8.8,-1.7)  { \footnotesize$\bar \kappa^{\mathrm{C}}_2$};
			\node  at (4.8,-1.7)  { \footnotesize$ \bar\kappa^{\mathrm{C}}_1$};
				\node  at (2.9,-0.4)  { \footnotesize$ \kappa^{\mathrm{R}}_1$};

			\draw[thick](8,-1.5)--(8,1.5);
			\draw[thick](9.5,-1.5)--(9.5,1.5);

			\fill[orange!50]
		(8,1.5)                          
			to[out=-45, in=-145] (9.5,1.5)       
			--(9.5,0)--(8,0)-- cycle;	
			
				\fill[red] (2.5,0) circle (3pt);
				\fill[red] (-9.5,0) circle (3pt);
	\fill[red] (4.7,1.2) circle (3pt);
	\fill[red] (4.7,-1.2) circle (3pt);
	\fill[red] (8.7,1.2) circle (3pt);
		\fill[red] (8.7,-1.2) circle (3pt);
			\fill[red] (-5.5,1.47) circle (3pt);
	\fill[red] (-5.5,-1.47) circle (3pt);
			
		\end{tikzpicture}
	\end{center}
	\caption{Initial configurations of stationary phase points illustrated for the cases of genus one (left) and genus three (right)}
	\label{figure-intial}
\end{figure}
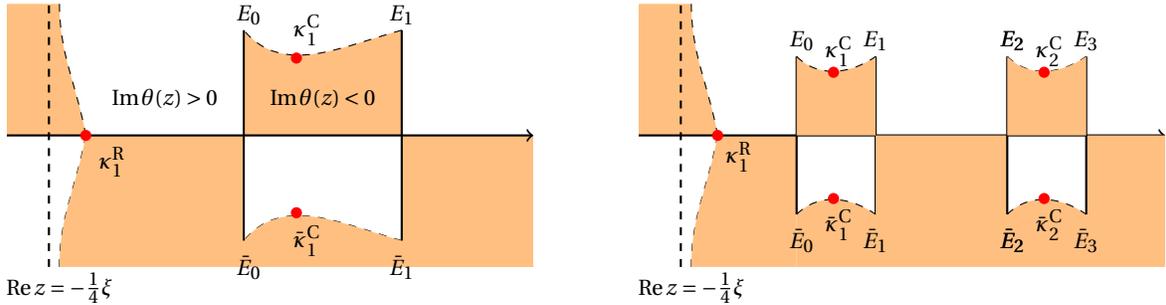

As an example, on a Riemann surface of genus one, the stationary phase points
of the phase function are given by the roots of the characteristic cubic equation:
\begin{equation}\label{eq:cubic}
	h(z) := z^3 - h_1 z^2 + h_2 z - h_3 = 0,
\end{equation}
with
\begin{align}
	h_1 &= B_0 + B_1- \frac{1}{4}\xi, \\
	h_2 &= \frac{2B_0B_1 + A_0^2 + A_1^2}{2} - \frac{(B_0 + B_1)\xi}{4}, \\
	h_3 &= -\frac{\int_{\bar{E}_1}^{E_1} \left(z^3 + h_2 z^2 + h_1 z\right) \frac{\mathrm{d}z}{w(z)}}{\int_{\bar{E}_1}^{E_1} \frac{\mathrm{d}z}{w(z)}} \in \mathbb{R}.
\end{align}
The three roots $\kappa_k$ ($k = 1, 2, 3$) are given by:
\begin{equation}\label{eq:roots}
	\kappa_k = \frac{h_1}{3} + 2\sqrt{-\frac{p^\dagger(\xi)}{3}} \cos\left(\frac{1}{3}\arccos\left[\frac{3q^\ddagger(\xi)}{2p^\dagger(\xi)}\sqrt{-\frac{3}{p^\dagger(\xi)}}\right] + \frac{2(k-1)\pi}{3}\right),
\end{equation}
where  \begin{align}
	&p^\dagger(\xi)=-\frac{\xi^2}{48}-\frac{\left(B_0+B_1\right) \xi}{12}+C^\dagger,\quad C^\dagger=\frac{3\left(A_0^2+A_1^2\right)-2\left(B_0-B_1\right)^2}{6},\\
	&	q^\ddagger(\xi) = -\frac{\xi^3}{864} - \frac{(B_0+B_1)\xi^2}{144} + \frac{\left[2(B_0+B_1)^2 + 6B_0B_1 + 3(A_0^2+A_1^2)\right]\xi}{72} + C^\ddagger,
\end{align}
and
 \begin{equation}
	C^\ddagger=-h_3(\xi) - \frac{(B_0+B_1)\left(2B_0B_1+A_0^2+A_1^2\right)}{6} + \frac{2(B_0+B_1)^3}{27}.
\end{equation}
Consider the asymptotic regime as $\xi \gg0$ and its subsequent decrease. As $\xi$ decreases from infinity, the first critical transition occurs when  $\frac{1}{3}\arccos\left(\frac{3q^\ddagger(\xi)}{2p^\dagger(\xi)}\sqrt{-\frac{3}{p^\dagger(\xi)}}\right) \to \frac{\pi}{3}$, at which point the discriminant $\Delta(\xi) = -4(p^\dagger(\xi))^3 - 27(q^\ddagger(\xi))^2$ vanishes. This condition induces the coalescence of $\kappa_2$ and $\kappa_3$ into a double root, while $\kappa_1$ remains simple. Explicitly, we obtain the limiting behavior
\[
\kappa_1=\frac{h_1}{3}-2 \sqrt{-\frac{p^\dagger(\xi)}{3}}\mapsto-\frac{1}{4}\xi+\frac{B_0+B_1}{3}, \quad\kappa_2=\kappa_3=\frac{h_1}{3}+\sqrt{-\frac{p^\dagger(\xi)}{3}}\mapsto \frac{B_0+B_1}{3},
\]
yielding the strict ordering
\[
-\frac{1}{4}\xi < \kappa_1 < B_0 < \operatorname{Re}\kappa_2 = \operatorname{Re}\kappa_3 < B_1.
\]
Continuing the decrease of $\xi>0$, the infinite branch of $\operatorname{Im}\theta(z)$ translates rightward until it intersects the finite branch. This transition corresponds to the second vanishing of the discriminant, occurring when $\frac{1}{3}\arccos\left(\frac{3q^\ddagger(\xi)}{2p^\dagger(\xi)}\sqrt{-\frac{3}{p^\dagger(\xi)}}\right) \to 2\pi$, whereby $\kappa_1$ and $\kappa_2$ merge. The resulting configuration comprises one simple root $z_3 = \frac{h_1}{3} + 2\sqrt{-\frac{p^\dagger(\xi)}{3}}$ and one double root $z_1 = z_2 = \frac{h_1}{3} - \sqrt{-\frac{p^\dagger(\xi)}{3}}$. For the boundary conditions satisfying $C^\dagger > 3B_0^2$, when $\xi \in (6|B_0| + \frac{2C^\dagger}{|B_0|}, +\infty)$, the coalescence occurs with the double root satisfying $z_1 = z_2 < B_0$, as shown in Figure~\ref{fff}.
\begin{figure}
	\subfloat[After the complex-complex collision]{\label{a}
		\begin{minipage}[t]{0.55\linewidth}
			\centering
			\includegraphics[width=2.7in]{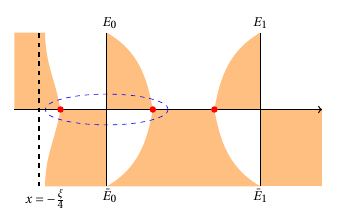}
		\end{minipage}
	}
	\subfloat[Befor the real-real collision]{\label{b}  
		\begin{minipage}[t]{0.4\linewidth}
			\centering
			\includegraphics[width=2.7in]{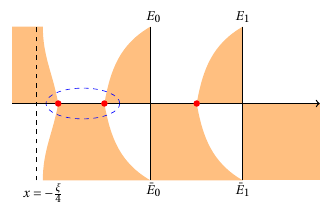}
		\end{minipage}
	}
	\caption{ In the transition region $|\xi-\xi_1|t^{2/3}<C$, two stationary phase points coalesce when $C^\dagger > 3B_0^2$.}
	\label{fff}
\end{figure}

\begin{remark}
Herein, the assumption on stationary phase point distribution in \eqref{h-pp}
is validated by the following construction. For the odd genus case, the numbers of branch cuts in the left and right half-planes are equal. Based on the properties of $\theta(z)$
\eqref{theta-tt}--\eqref{eq:large-theta} and the configuration of
endpoints $E_k$, the coalescence of the infinite branch with
$\Gamma_0$ as $\xi\gg0$ decrease exhibits two possible distinct
scenarios (Figures~\ref{ttt} and \ref{TTT}, among others),
depending on the choice of branch configuration $\{B_k,A_k\}_{k=0}^n$.

We restrict our analysis to the second scenario (Figure~\ref{fff},\ref{TTT}),
in which real-real coalescence occurs after complex-complex coalescence,
with exactly three real stationary phase points existing prior to the
coalescence event. At the same time, the real stationary phase point $\kappa^R_3$,
which moves rightward as $\xi$ decreases, may lie to the left or right
of the cut $\Gamma_1$; in both configurations, the contribution is
$\mathcal{O}(t^{-1/2})$ to the long-time asymptotics. Without loss of
generality, we consider the case where $\kappa^R_3$ lies to the left
of \,$\Gamma_1$ under the choice of branch configuration $\{B_k,A_k\}_{k=0}^n$ as shown in the first figure in Figure~\ref{TTT}.

Following this coalescence, the region where ${\rm Im}\,\theta(z)>0$ is restricted to the finite domain enclosed by $E_0$, $\bar E_0$, $\bar E_1$, and $E_1$, whereas the infinite branch coinciding with $\kappa^R_3$ migrates rightward with decreasing $\xi$ as shown in the second and third figure in Figure~\ref{TTT}.
\end{remark}
\begin{figure}[htp]
	\subfloat{\label{a}
		\begin{minipage}[t]{0.3\linewidth}
			\centering
			\includegraphics[width=2in]{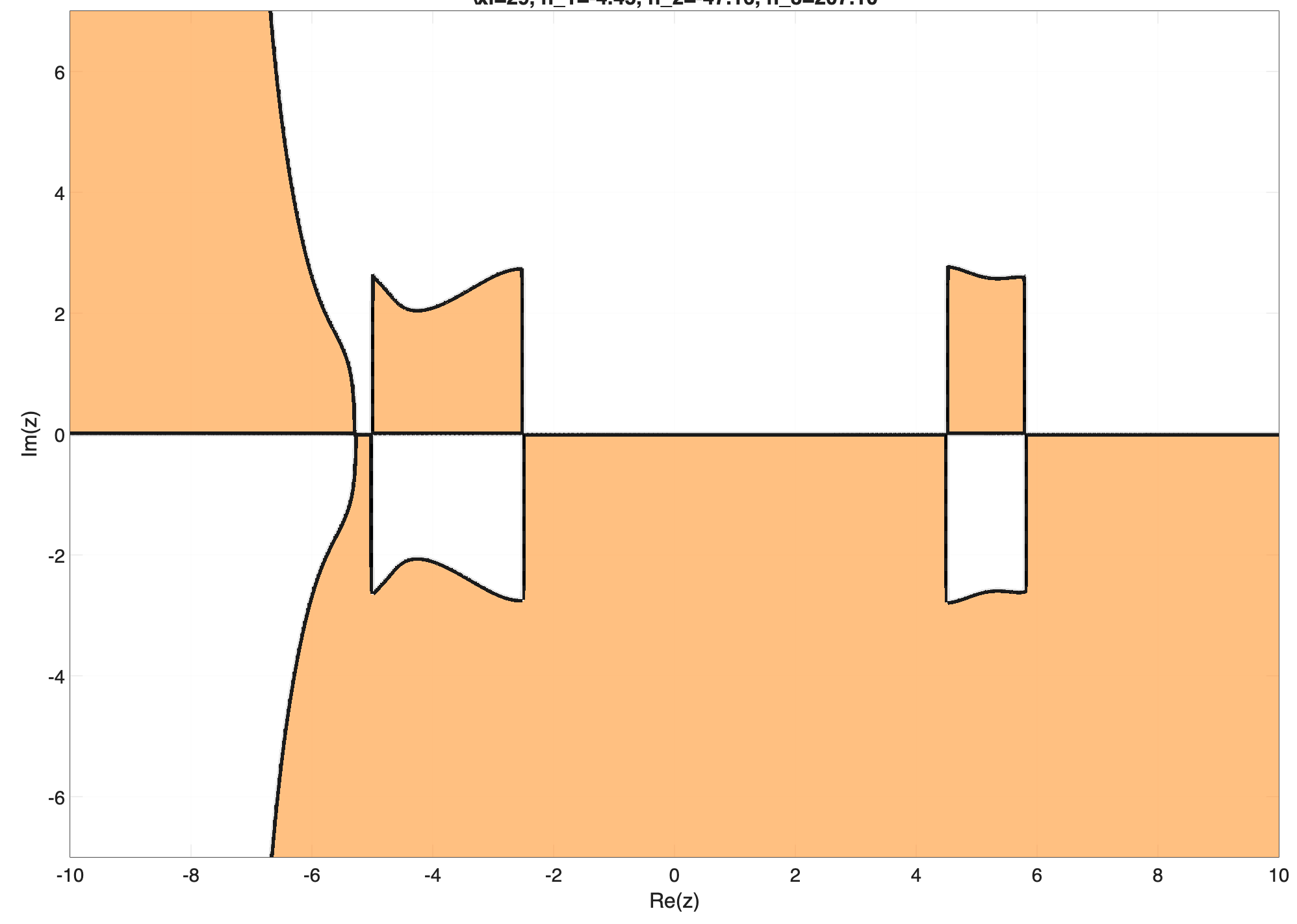}
		\end{minipage}
	}\quad
	\subfloat{\label{b}  
		\begin{minipage}[t]{0.3\linewidth}
			\centering
			\includegraphics[width=2in]{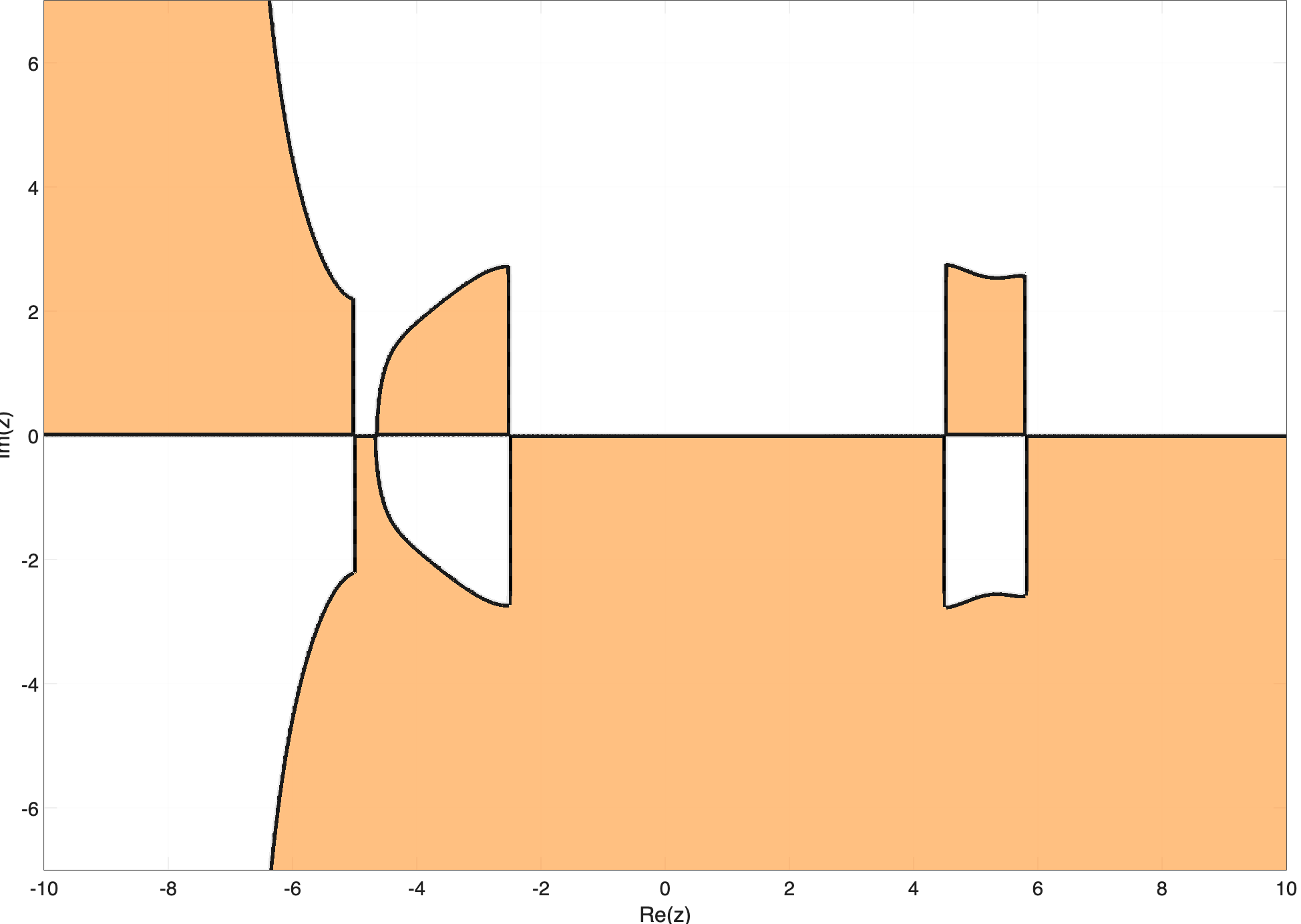}
		\end{minipage}
	}
	\subfloat{\label{c}\begin{minipage}[t]{0.4\linewidth}
			\centering
			\includegraphics[width=2in]{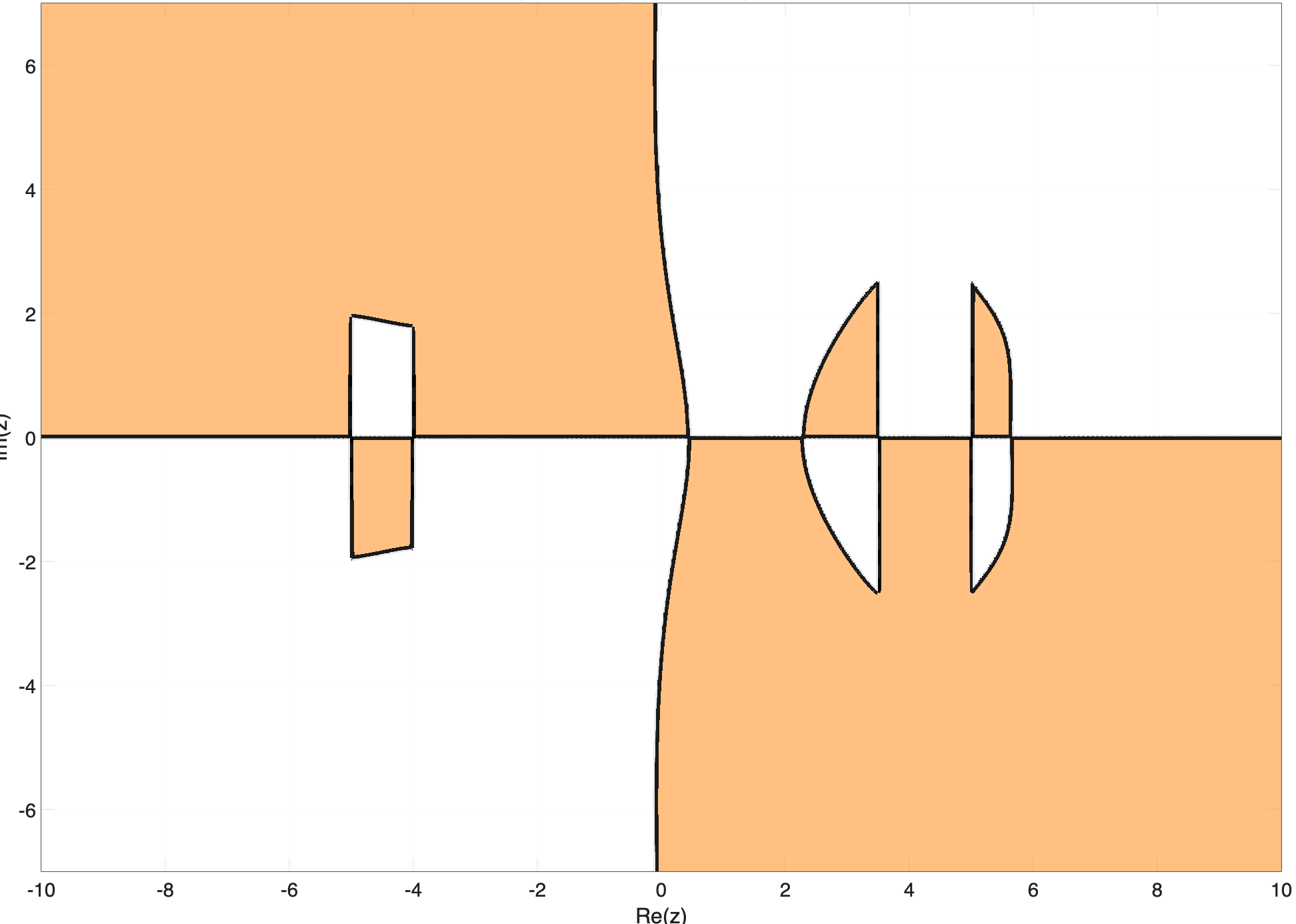}
	\end{minipage}}
\caption{The branch cut coalesces at infinity precisely when complex critical
	points remain off the real axis, yielding signature regions with
	${\rm Im}\,\theta(z)<0$ (orange) and ${\rm Im}\,\theta(z)>0$ (white) as $t\to\infty$.}
\label{ttt}
\end{figure}
\subsection*{Conjugation}
The first step in our analysis is to introduce factorizations of the jump
matrix that extend analytically off the real axis, following the approach
in \cite{CJ16,KTR08}. To achieve a properly normalized RH problem, we introduce the function
\begin{equation}\label{eq:ddeltaa}
	\delta(z)=\tilde{\nu}^{-1}(z)\exp\!\left\{\frac{w(z)}{2\pi i}\left(\int_{I}\frac{\log\!\left(1+|r(s)|^2\right)}{w(s)(s-z)}\,ds+\int_{\Gamma\setminus\tilde\Gamma}\frac{\log\frac{\overline{a_-(\bar{s})}}{\overline{a_+(\bar{s})}}}{w(s)(s-z)}\,ds+\int_{\bar{\Gamma}\setminus\bar{\tilde\Gamma}}\frac{\log\frac{\overline{a_+(\bar{s})}}{\overline{a_-(\bar{s})}}}{w(s)(s-z)}\,ds\right)\right\},
\end{equation}
where $\tilde\Gamma:=\bigcup_{k>[\frac{n}{2}]}\Gamma_k$ and the integration contour is given by
$I = \left(-\infty,\kappa^{\mathrm{R}}_1\right)\cup\left(\kappa^{\mathrm{R}}_2,B_{k_0+1}\right)\cup\left(\kappa^{\mathrm{R}}_3,B_{k_0+2}\right)$.
The function $\tilde{\nu}(z)=\prod_{k>\left[\frac{n}{2}\right]}\sqrt[4]{\frac{z-\bar{E}_k}{z-E_k}}$
is defined for parameters satisfying
$B_{k_0}<\kappa^{\mathrm{R}}_1<\kappa^{\mathrm{R}}_2<B_{k_0+1}<\kappa^{\mathrm{R}}_3<B_{k_0+2}$
with some $k_0\in\mathcal{K}$, and the logarithm assumes its principal branch.
The function $\delta(z)$ characterized by \eqref{eq:ddeltaa} satisfies the
following scalar RH problem.
\begin{figure}[htp]
	\subfloat{\label{a}
		\begin{minipage}[t]{0.3\linewidth}
			\centering
			\includegraphics[width=2in]{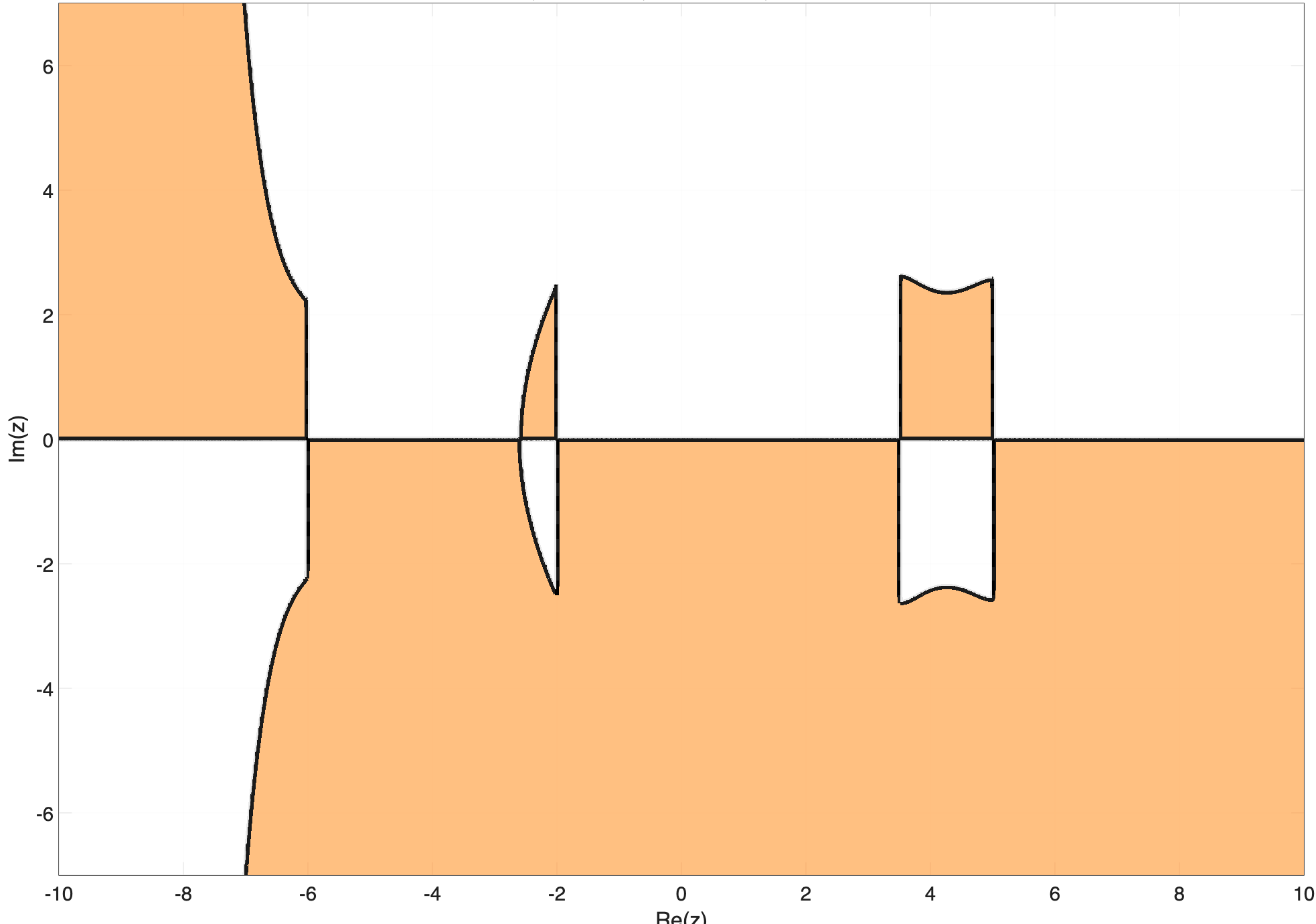}
		\end{minipage}
	}\quad
	\subfloat{\label{b}  
		\begin{minipage}[t]{0.3\linewidth}
			\centering
			\includegraphics[width=2in]{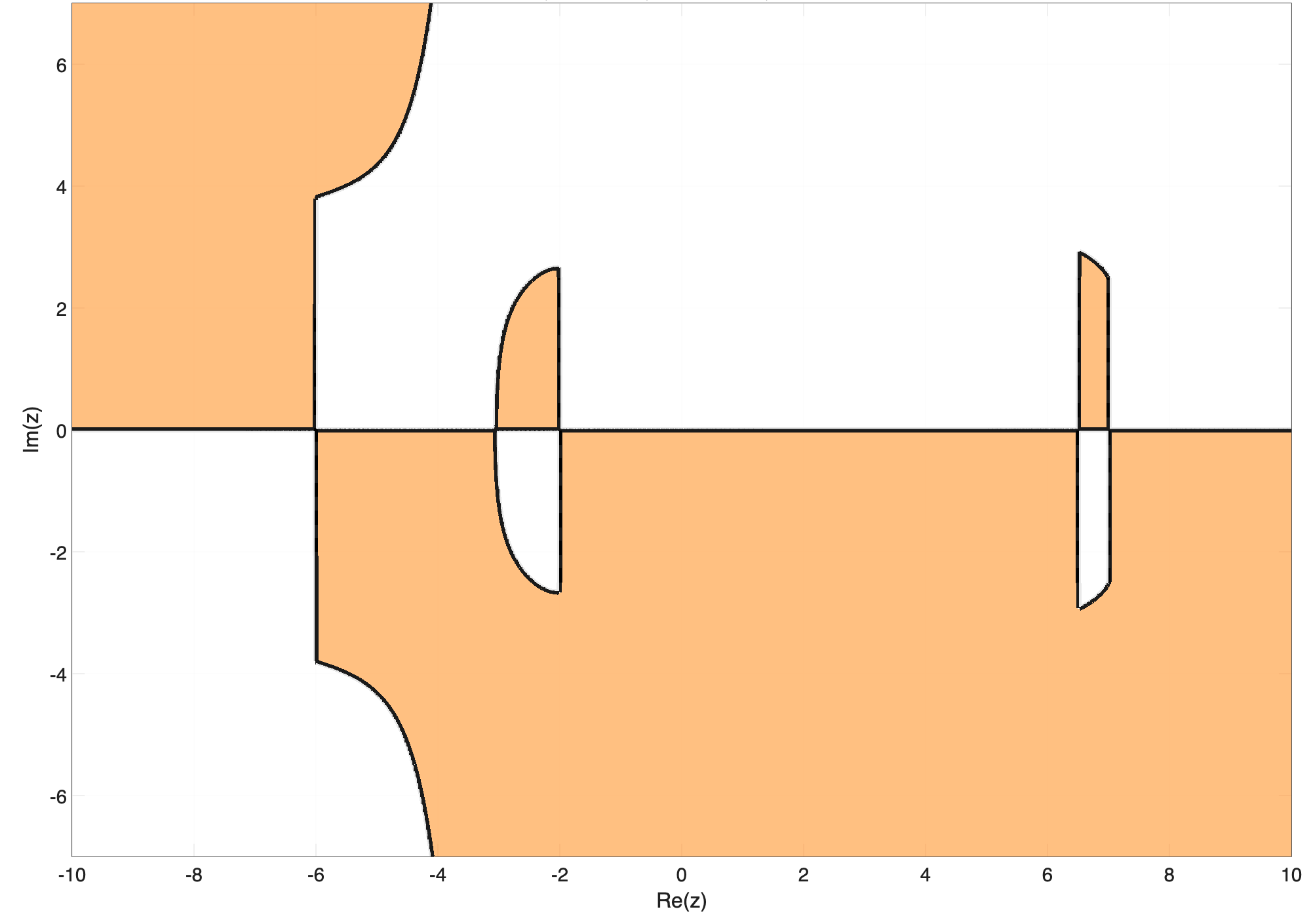}
		\end{minipage}
	}
	\subfloat{\label{c}\begin{minipage}[t]{0.4\linewidth}
			\centering
			\includegraphics[width=2in]{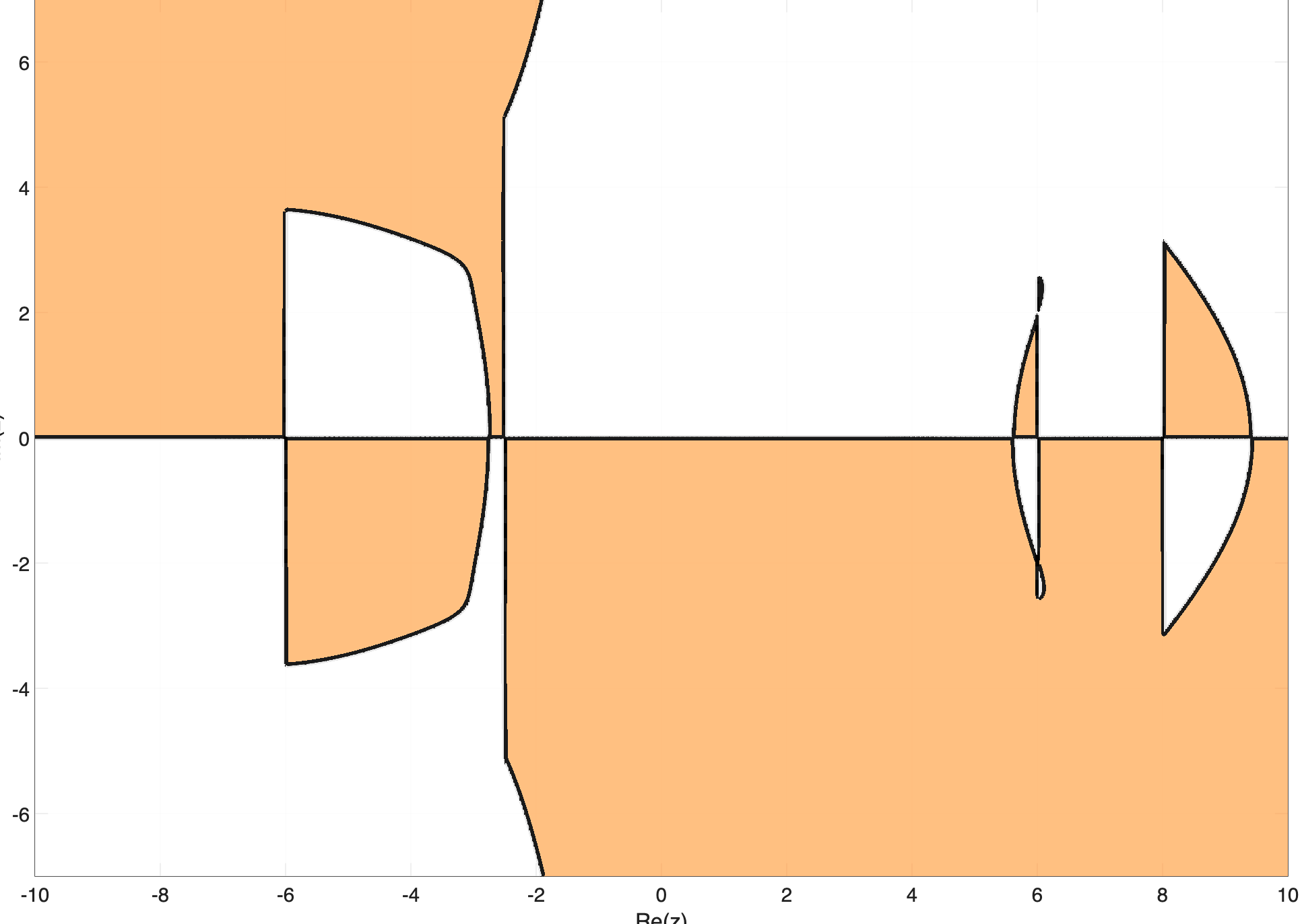}
	\end{minipage}}
	\caption{The branch cut coalesces with real stationary phase point on the real axis, yielding signature regions with
	${\rm Im}\,\theta(z)<0$ (orange) and ${\rm Im}\,\theta(z)>0$ (white) as $t\to\infty$.}
\label{TTT}
\end{figure}
\begin{problem}
	Construct a meromorphic function $\delta(z): \mathbb{C}\setminus\left( I\,\cup\,\Gamma \right)\to SL_2(\mathbb{C})$ such that:
	\begin{enumerate}
		\item[(i)]  $\delta(z)=\delta(\infty)\left(1+\frac{I_1}{z}+O\left(\frac{1}{z^2}\right)\right)
		$ as $z\to\infty$, where
		\begin{align}\label{delt-ingg}
			&\delta(\infty)=\exp\left\{-\frac{1}{2\pi i}\left[
			\int_{I}\frac{s^n\log(1+|r(s)|^2)}{w(s)}\,ds
			+\int_{\Gamma\setminus\tilde\Gamma}\frac{s^n\log\frac{\overline{a_-(\bar s)}}{\overline{a_+(\bar s)}}}{w(s)}\,ds
			+\int_{\bar\Gamma\setminus\bar{\tilde\Gamma}}\frac{s^n\log\frac{\overline{a_+(\bar s)}}{\overline{a_-(\bar s)}}}{w(s)}\,ds
			\right]\right\},\\
			&I_1=-\frac{1}{2\pi i}\left[
			\int_{I}\frac{s^n(s+\hat{f}_n)\log(1+|r(s)|^2)}{w(s)}\,ds
			+\int_{\Gamma\setminus\tilde\Gamma}\frac{s^n(s+\hat{f}_n)\log\frac{\overline{a_-(\bar s)}}{\overline{a_+(\bar s)}}}{w(s)}\,ds
			+\int_{\bar\Gamma\setminus\bar{\tilde\Gamma}}\frac{s^n(s+\hat{f}_n)\log\frac{\overline{a_+(\bar s)}}{\overline{a_-(\bar s)}}}{w(s)}\,ds
			\right].
		\end{align}
		where $\hat{f}_n$ is given in \eqref{hat-ff}.
		\item[(ii)]The boundary values $\delta_\pm(z)$ satisfy
		\begin{equation}
			\delta_+(z)=\delta_-(z)\begin{cases}
			\bigl(1+|r(z)|^2\bigr),&z\in I,\\
\log\frac{\overline{a_-(\bar{s})}}{\overline{a_+( \bar{s})}},&z\in (\Gamma\,\setminus\,\tilde\Gamma)\cap\mathbb{C}^+,\\
\log\frac{\overline{a_+(\bar s)}}{\overline{a_-(\bar s)}},&z\in (\Gamma\,\setminus\,\tilde\Gamma)\cap\mathbb{C}^-,\\
i,&z\in\tilde\Gamma.			\end{cases}
		\end{equation}
		\item[(iii)] $\delta(z)$ obeys the symmetry
		\begin{equation}
			\delta(z)=\left(\overline{\delta(\bar z)}\right)^{-1}.
		\end{equation}
		\item[(iv)] $\delta(z)$ and $\delta^{-1}(z)$ are bounded and analytic functions of $z \in \mathbb{C}\setminus\left( (I\setminus(\cup_{k<k_0}B_{k}))\,\cup\,\Gamma \right)$.
	\end{enumerate}
\end{problem}
\begin{proof}
	Properties $(ii)$, $(iii)$, and $(iv)$ can be directly deduced from the definition of $\delta(z)$ in \eqref{eq:ddeltaa}. To establish property $(i)$, let $|z|>\max\{|s|:s\in I\cup\Gamma\}$ and expand
	\begin{equation}\label{eq:cauchy}
		\frac{1}{s-z}=-\sum_{k=0}^{\infty}\frac{s^k}{z^{k+1}}.
	\end{equation}
	Let $\delta(z)=e^{E(z)}$ as $z\to\infty$, substituting \eqref{eq:cauchy} into the exponent of \eqref{eq:ddeltaa} yields
	\begin{equation*}\label{eq:Ez}
		E(z)=-\frac{1}{2\pi i}\sum_{k=0}^{\infty}\frac{w(z)}{z^{k+1}}M_k,
		\quad 
		M_k:=\int_{I}\frac{s^k\log(1+|r(s)|^2)}{w(s)}\,ds
		+\int_{\Gamma\setminus\tilde\Gamma}\frac{s^k\log\frac{\overline{a_-(\bar s)}}{\overline{a_+(\bar s)}}}{w(s)}\,ds
		+\int_{\bar\Gamma\setminus\bar{\tilde\Gamma}}\frac{s^k\log\frac{\overline{a_+(\bar s)}}{\overline{a_-(\bar s)}}}{w(s)}\,ds,
	\end{equation*}
	the terms $z^n,\dots,z^1$ appear for $k=0,\dots,n-1$. To ensure $\delta(z)=e^{E(z)}$ has a finite nonzero limit at infinity, the moment conditions
	\begin{equation}\label{eq:moment}
		M_0=M_1=\cdots=M_{n-1}=0
	\end{equation}
	must hold. Under \eqref{eq:moment}, only $k\ge n$ contribute to the $z^0$ and $z^{-1}$ coefficients. Keeping $k=n$ and $k=n+1$, we obtain
	\begin{equation}\label{eq:Eexp}
		E(z)=E_0+\frac{E_1}{z}+O(z^{-2}),
	\end{equation}
	with
	\begin{equation}\label{eq:E0E1}
		E_0=-\frac{M_n}{2\pi i},\qquad 
		E_1=-\frac{\hat{f}_nM_n+M_{n+1}}{2\pi i}.
	\end{equation}
	Therefore,
	\begin{equation}\label{eq:delta-asymp}
		\delta(z)=e^{E(z)}=\delta(\infty)\left(1+\frac{I_1}{z}+O(z^{-2})\right),
		\qquad z\to\infty,
	\end{equation}
The following result admits a complex proof for property $(i)$.
\end{proof}

 \begin{figure}[htp]
	\begin{center}
		\begin{tikzpicture}[scale=0.7]

				\fill[orange!50]
			(-9,0)--(-9,3)--(-5.5,3) to[out=-90, in=95] (-5,0)
			-- cycle;			
			\draw[dashed](-5.5,3) to[out=-90, in=95] (-5,0);
				\fill[orange!50]
			(-1,0) to[out=70, in=-160] (1,3)--(1,0)
			-- cycle;			
			\draw[dashed]	(-1,0) to[out=70, in=-160] (1,3);
			
				\fill[orange!50]
		(-5.5,-3)  to[out=90, in=-95] (-5,0)--(-1,0)to[out=-70, in=160] (1,-3)
			-- cycle;
			\draw[dashed]	(-5.5,-3)  to[out=90, in=-95] (-5,0);
			\draw[dashed](-1,0)to[out=-70,
			 in=160] (1,-3);
			
			 	\fill[orange!50]
			 (3.3,0) to[out=80, in=-150] (5,3)--(5,0)
			 -- cycle;			
			 \draw[dashed]	 (3.3,0) to[out=80, in=-150] (5,3);
			
			 	\fill[orange!50]
			 (1,-3)--(1,0)--(3.3,0)
			  to[out=-80, in=150] (5,-3)
			 -- cycle;			
			 \draw[dashed](3.3,0)
			 to[out=-80, in=150] (5,-3);
			 	\fill[orange!50]
 (5,-3)--(5,0)--(9,0)--(9,-3)
			 -- cycle;

			\node  at (3.1,-0.5)  { \footnotesize$\kappa^{\mathrm{R}}_3$};
				\node  at (-4.5,-0.5)  { \footnotesize$\kappa^{\mathrm{R}}_1$};
				\node  at (-1.3,-0.5)  { \footnotesize$\kappa^{\mathrm{R}}_2$};
			
				\draw[dashed,blue] (3.3,0) ellipse (1.5cm and 1cm);
				\draw[thick](4.3,0.8)--(2.3,-0.8);
				\draw[thick](4.3,-0.8)--(2.3,0.8);
				\node  at (4.7,-1.3)  { \footnotesize$\bar\Sigma^1_{\kappa^{\mathrm{R}}_3}$};
				\node  at (2.,-1.3)  { \footnotesize$\bar\Sigma^2_{\kappa^{\mathrm{R}}_3}$};
				
					\node  at (4.7,1.3)  {
				\footnotesize$\Sigma^1_{\kappa^{\mathrm{R}}_3}$};
				\node  at (2.,1.3)  { \footnotesize$\Sigma^2_{\kappa^{\mathrm{R}}_3}$};

			\draw[dashed,blue] (-3,0) ellipse (2.9cm and 1.3cm);

								\draw[thick](-1,0)--(0.5,1.5);
					\draw[thick](-1,0)--(0.5,-1.5);
	\draw[thick](-5,0)--(-6.3,1.5);
					\draw[thick](-5,0)--(-6.3,-1.5);

						\node  at (2.5,2.5)  { \footnotesize$\Omega^1_{\kappa^{\mathrm{R}}_2}$};
					
						\node  at (2.5,-2.5)  { \footnotesize$\bar\Omega^1_{\kappa^{\mathrm{R}}_2}$};
						
							\node  at (-7.5,1.5)  { \footnotesize$\Omega^2_{\kappa^{\mathrm{R}}_1}$};
											\node  at (-7.5,-1.5)  { \footnotesize$\Omega^2_{\kappa^{\mathrm{R}}_1}$};
			
				\draw[thick, ->](-9,0)--(9,0);
			\draw[thick](1,-3)--(1,3);
			\draw[thick](5,-3)--(5,3);

	\fill[blue] (-3,0) circle (3pt);

								\fill[red] (3.3,0) circle (3pt);
						\fill[red] (-1,0) circle (3pt);
						\fill[red] (-5,0) circle (3pt);
		\end{tikzpicture}
	\end{center}
	\caption{ The coalescence of two stationary phase points $\kappa^\mathrm{R}_0,\kappa^\mathrm{R}_1<B_0$ as $\xi\gg0$.}
	\label{fig:zs}
\end{figure}
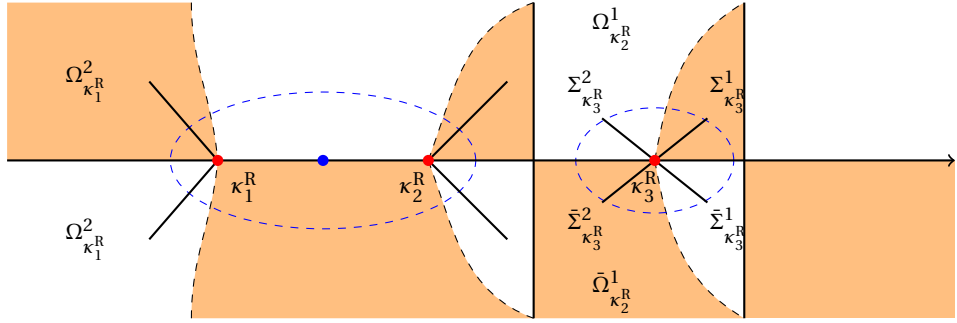
Denote by $\Sigma_{\kappa^\mathrm{R}_l} = \Sigma_{\kappa^\mathrm{R}_l}^t \cup \bar{\Sigma}_{\kappa^\mathrm{R}_l}^t$ ($t=1,2,l\in\mathcal{L}$) the curves oriented counterclockwise, which can be chosen freely subject only to the following conditions: they pass through $\kappa^\mathrm{R}_{l}$, do not intersect $\Gamma_k$, and lie entirely in regions where $\theta(z)$ has the appropriate sign. Denote by$$
\begin{aligned}
\Omega_{\kappa_2^{\mathrm{R}}}^1&=\left\{z\in\mathbb{C}\setminus\mathcal{U}_0:0\leq\arg(z-\kappa_2^{\mathrm{R}})\leq\varphi_0\right\},\\
	\Omega_{\kappa_1^{\mathrm{R}}}^2&=\left\{z\in\mathbb{C}:\pi-\varphi_0\leq\arg(z-\kappa_1^{\mathrm{R}})\leq\pi\right\},
\end{aligned}
$$
where $\mathcal{U}_0=\{z\in\mathbb{C}:|z-\kappa^{\mathrm{R}}_3|<\epsilon\}$ with $\epsilon$ sufficiently small to ensure $\mathcal{U}_0$ does not intersect $\Gamma$ and
 $\varphi_0\in(0,\pi/2)$ is chosen such that $\Gamma_k\cap\mathbb{C}^+\subset	\Omega_{\kappa_2^{\mathrm{R}}}^1$ for $k>k_0$ and $\Gamma_k\cap\mathbb{C}^+\subset	\Omega_{\kappa_1^{\mathrm{R}}}^2$ for $k\leq k_0$, the sectors in $\mathbb{C}$ separated by $\mathbb{R}$ and the corresponding curves $\Sigma^t_{\kappa^\mathrm{R}_l} \cup \bar{\Sigma}_{\kappa^\mathrm{R}_l}^t$, as shown in Figure~\ref{fig:zs}.

Define
$$\Sigma^{(1)}=\Gamma\,\bigcup\,\left(\Sigma^{2}_{\kappa^\mathrm{R}_1}\cup \bar\Sigma^{2}_{\kappa^\mathrm{R}_1}\cup\Sigma_{\kappa^\mathrm{R}_2}^1\cup\bar\Sigma_{\kappa^\mathrm{R}_2}^1
\cup\Sigma_{\kappa^\mathrm{R}_3}\right)$$
We define a new function
\begin{equation}\label{tran-1}
N^{(1)}(z)=e^{-ig(\infty)\sigma_3}\delta(\infty )^{-\sigma_3} N(z)\delta^{\sigma_3}(z) G(z)e^{g(z)\sigma_3},
\end{equation}
where
\begin{equation}\label{eq:G(zz)}
	G(z)= I+\begin{cases}
    V_4= \frac{\overline{r(\bar z)} e^{-2 i t \theta(z)}}{1+|r(z)|^2}\sigma_+
		,  z \in \Omega^{2}_{\kappa^\mathrm{R}_1}\cup\Omega^1_{\kappa^\mathrm{R}_2}\cap \Omega^1_{\kappa^\mathrm{R}_3},\quad  V^{-1}_2=	- r(z) e^{2 i t \theta(z)}\sigma_-,z\in\Omega^{2}_{\kappa^\mathrm{R}_3},\\
V_3=\frac{r(z) e^{2 i t \theta(z)}}{1+|r(z)|^2}\sigma_-,  z \in \bar\Omega^{2}_{\kappa^\mathrm{R}_1}\cup\bar\Omega^1_{\kappa^\mathrm{R}_2} \cup\bar\Omega^1_{\kappa^\mathrm{R}_3} , 	\quad \,\,\,V_1= {\overline{r(\bar z)}} e^{-2 i t \theta(z)}\sigma_+,\quad z\in\bar\Omega^2_{\kappa^\mathrm{R}_3},\\
		I,  \text { elsewhere, }\end{cases}
\end{equation}
	and		\begin{equation}
		\sigma_-=	\begin{pmatrix}
				0&0\\1&0
			\end{pmatrix},\qquad \sigma_+=\begin{pmatrix}
			0&1\\
			0&0
			\end{pmatrix}.
		\end{equation}
The function $g(z)$ is defined via the Cauchy integral
\begin{equation}\label{eq:hhhh(z)}
	g(z)=\frac{w(z)}{2\pi i}\int_{\tilde\Gamma}\frac{\mathcal{H}(s)}{w(s)(s-z)}\,ds,\qquad z\in\mathbb{C}\setminus\tilde\Gamma,
\end{equation}
where the function $\mathcal{H}(z)$ is prescribed on $\tilde\Gamma$ according to
\begin{equation}\label{eq:mathcal-H-def}
	\mathcal{H}(z)=
	\begin{cases}
		\frac{i\ln\delta^{2}(z)}{w_{+}(z)},\,\,\, \qquad z\in\Gamma_{[\frac{n}{2}]+1}\cap\mathbb{C}^+,\qquad\quad
\frac{2\alpha+i\ln\delta^{2}(z)}{w_{+}(z)},\qquad z\in(\tilde\Gamma\setminus\Gamma_{[\frac{n}{2}]+1})\cap\mathbb{C}^+,\\
		\frac{-i\ln\left(-i\overline{\delta^{2}(\bar z)}\right)}{w_{+}(z)},\,\,\, z\in(\tilde\Gamma\setminus\Gamma_{[\frac{n}{2}]+1})\cap\mathbb{C}^-,\quad
			\frac{2\alpha-i\ln\left(-i\overline{\delta^{2}(\bar z)}\right)}{w_{+}(z)},\,\,\, z\in(\tilde\Gamma\setminus n\Gamma_{[\frac{n}{2}]+1})\cap\mathbb{C}^-.
	\end{cases}
\end{equation}
Here $\alpha=\alpha(\xi) \in \mathbb{R}$ denotes a real parameter yet to be determined, while $w_{+}(z)$ represents the boundary value of $w(z)$ taken from the positive side of $\tilde{\Gamma}$. Consequently, when $\alpha$ is chosen real, the prescription \eqref{eq:mathcal-H-def} necessarily enforces the Schwarz symmetry condition $\mathcal{H}(z)=\overline{\mathcal{H}(\bar{z})}$ on $\tilde{\Gamma}$; this assertion will be rigorously established in the following lemma.

\begin{lemma}
There exists a unique choice of real constants $\alpha=\alpha(\xi)$, such that the function $g(z)$ defined in \eqref{eq:hhhh(z)} possesses the following properties:
	\begin{enumerate}
		\item It satisfies the symmetry condition $g(z)=\overline{g(\bar{z})}$.
		\item As $z\to\infty$, $g(z)$ admits the asymptotic expansion		\begin{equation}
		g(z)=g(\infty)\left(1+\frac{I_2}{z}+\mathcal{O}(z^{-2})\right),
		\end{equation}
		where $g(\infty)$ is a finite real constant given by
		\begin{align}\label{ginfty}
		g(\infty)=-\frac{1}{2\pi i}\int_{\tilde\Gamma}\frac{s^n\mathcal H(s)}{w(s)}\,ds, \quad
		I_2=\hat{f}_n+\frac{\displaystyle\int_{\tilde\Gamma}\frac{s^{n+1}\mathcal H(s)}{w(s)}\,ds}
		{\displaystyle\int_{\tilde\Gamma}\frac{s^n\mathcal H(s)}{w(s)}\,ds}.
		\end{align}
		\item $e^{ig(z)\sigma_3}$ is bounded and analytic for $z\in\mathbb{C}\setminus\tilde\Gamma_k$.
		\item For each $z \in \tilde\Gamma$, the boundary values $g_\pm(z)$ satisfy the jump relations
		\begin{align}\label{eq:h-condition}
			&g_{+}(z)+g_{-}(z)=
			\begin{cases}
				2\alpha+i\ln\delta^{2}(z)
,&z\in\tilde\Gamma_k\cap\mathbb{C}^+,\\
				2\alpha-i\ln (-\overline{\delta^2(\bar z)}),&z\in\tilde\Gamma_k\cap\mathbb{C}^-.
			\end{cases}
		\end{align}
	\end{enumerate}
\end{lemma}
\begin{proof}
	Expanding the Cauchy kernel as in \eqref{eq:cauchy} and substituting into \eqref{eq:hhhh(z)}, we obtain
	\begin{equation}
		g(z)=-\frac{1}{2\pi i}\sum_{k=0}^{\infty}\frac{w(z)}{z^{k+1}}\widetilde M_k,
		\qquad 
		\widetilde M_k:=\int_{\tilde\Gamma}\frac{s^k\mathcal H(s)}{w(s)}\,ds.
	\end{equation}
	Under the moment conditions $\widetilde M_0=\cdots=\widetilde M_{n-1}=0$, only the terms with $k\ge n$ contribute. Using $w(z)=z^{n+1}+\hat f_nz^n+\cdots+a_0$, we have
	\begin{align*}
		g(z)&=-\frac{1}{2\pi i}\left[\left(1+\frac{\hat f_n}{z}+O(z^{-2})\right)\widetilde M_n+\left(\frac{1}{z}+O(z^{-2})\right)\widetilde M_{n+1}+O(z^{-2})\right]\\
		&=-\frac{\widetilde M_n}{2\pi i}-\frac{\hat f_n\widetilde M_n+\widetilde M_{n+1}}{2\pi i\,z}+O(z^{-2}).
	\end{align*}
Writing $g(z)=g(\infty)\bigl(1+\frac{I_2}{z}+O(z^{-2})\bigr)$ and comparing coefficients completes the proof of the lemma.
\end{proof}
It follows that $N^{(1)}(z)$ satisfies the following RH problem.
\begin{problem}
	Construct a meromorphic function $N^{(1)}(z) :\mathbb{C}\setminus \Sigma^{(1)}\to SL_2(\mathbb{C})$  such that:
	\begin{enumerate}
		\item $N^{(1)}(z)=I+\mathcal{O}(z^{-1}),$\quad $|z|\to\infty$.
		\item 	For each $z\in\Sigma^{(1)}$, the boundary values $N^{(1)}_\pm(z)$ satisfy the jump relation
		$$N^{(1)}_+(z)=N^{(1)}_-(z)J^{(1)}(z),$$
		where
  \begin{equation}\label{eq:J-pele}
        J^{(1)}(z)=
        \begin{cases}
            I+V_3(z)\delta^{2}(z)\sigma_-, \qquad z\in\Sigma^{1}_{\kappa^R_2}\cup\Sigma^{2}_{\kappa^R_1}\cup
            \Sigma^{1}_{\kappa^R_3}           , \quad   I+V_1(z)\delta^{-2}(z)\sigma_+,  z\in\bar\Sigma^{2}_{\kappa^R_3}.            \\
            I+V_4(z)^{-1}\delta^{-2}(z)\sigma_+, \,\,\,z\in\bar\Sigma^{1}_{\kappa^R_2}\cup\bar\Sigma^{2}_{\kappa^R_1}\cup\bar\Sigma^{1}_{\kappa^R_1}            ,\quad    I+V_2(z)\delta^2(z)\sigma_-,  z\in\Sigma^{2}_{\kappa^R_3},\\            \begin{pmatrix}
                0 & ie^{-2\pi ic_k-2i\alpha}\\
                ie^{2\pi ic_k+2i\alpha} & 0
            \end{pmatrix},  z\in\tilde\Gamma_k\cap\mathbb{C}^+,\, \begin{pmatrix}
                0 & -ie^{-2\pi ic_k-2i\alpha}\\
                -ie^{2\pi ic_k+2i\alpha} & 0
            \end{pmatrix},  z\in\tilde\Gamma_k\cap\mathbb{C}^-,\\
\delta(z)^{\hat\sigma_3}J(z),z\in(\kappa^R_{1},\kappa^R_{2}).
        \end{cases}
    \end{equation}

	\item The function $N^{(1)}(z)$ has singularities at the endpoints of\,\, $\Gamma_k$ of order at most $|z-E_k|^{-1/4}$ or $|z-\bar{E}_k|^{-1/4}$.
	\end{enumerate}
\end{problem}
We will construct the solution $N^{(1)}(z)$ by seeking a solution of the form
\begin{equation}\label{eq:trans-2}
N^{(1)}(z)= \begin{cases}\mathcal{E}(z) N^{alg}(z), & z \in \mathbb{C} \backslash U_{loc}, \\ \mathcal{E}(z) N^{loc}(z) N^{alg}(z), & z \in U_{loc},\end{cases}
\end{equation}
where $U_{loc}$ denotes the union of open disks
\begin{equation}\label{eq:UU}
	U_{loc}:=D_{\rho_0}(z^{R}_{j})\cup D_\epsilon(\kappa_3^{\mathrm{R}}),
\end{equation}
where the parameter $\rho_0$ is chosen sufficiently small such that
$D_{\rho_0}(z^{R}_{j})$ contains the points $\kappa^{\mathrm{R}}_{1}$,
$z_{j}$, $\kappa^{\mathrm{R}}_{2}$ but does not intersect
$\Gamma_{k_0}$ or $\Gamma_{k_0+1}$. Here, $D_\epsilon(\kappa^{\mathrm{R}}_{3})$
denotes a disk with radius $\epsilon\equiv\epsilon(\xi)$ chosen small enough
to ensure it does not intersect $\Gamma_{k_0+2}$. Moreover, $N^{alg}(z)$ satisfies the RH problem~\ref{RH-global},
$N^{loc}(z)$ is defined via the RH problem~\ref{RH-ee}, and
the small-norm problem $\mathcal{E}(z)$ is specified in the RH problem~\ref{RHee}.

\subsection{Solution of the Model problem}
The matrix $N^{(1)}(z)$ is meromorphic on $\mathbb{C}\setminus\Sigma^{(1)}$
and satisfies the jump condition \eqref{eq:J-pele} across $\Sigma^{(1)}$.
Away from the real stationary phase points, the jump matrix approaches the
identity exponentially fast. Specifically, recalling $\theta(z)$ defined
in \eqref{eq:theta(z)}, we have the uniform estimate
\begin{equation}\label{eq:j-2-2}
	\sup_{z\in\Sigma^{(1)}\setminus\tilde\Gamma}\left\|J^{(1)}(z)-I\right\|=\mathcal{O}(e^{-ct}),\quad c=c(\xi)>0,
\end{equation}
as $t\to\infty$. Furthermore, on $\Gamma\setminus\tilde\Gamma$, the jump
simplifies to the identity matrix $J^{(1)}(z)=I$, since
\begin{align}
	J^{(1)}(z)=V_4^{-1}(z)V^\dagger(z)V_4(z),
\end{align}
and by virtue of \eqref{eq:tilder--1} and \eqref{eq:tilder--2}, the
off-diagonal entries vanish identically:
\begin{align}
	\delta^{-2}(z)e^{2it\theta(z)}\left(r_-(z)-r_+(z)+\frac{ie^{i\phi_k}}{\overline{a_+(\bar z)a_-(\bar z)}}\right)\equiv 0, \quad z\in\Gamma\setminus\tilde\Gamma.
\end{align}
Thus, we arrive at the following normalized RH problem.
\begin{problem}\label{RH-global}
	Construct a meromorphic function $N^{alg}(z) :\mathbb{C}\setminus\tilde\Gamma\to SL_2(\mathbb{C})$ such that:
	\begin{enumerate}
		\item $N^{alg}(z)=I+\mathcal{O}(z^{-1}),$\quad $|z|\to\infty$.
		\item 	For each $z\in\tilde\Gamma_k$, the boundary values $N_\pm^{alg}(z)$ satisfy the jump relation
		$$N^{alg}_+(z)=N^{alg}(z)J^{alg}(z),$$
		where $J^{alg}(z)=J^{(1)}(z)|$ for $_{z\in\tilde\Gamma}$.
		\item  The function $N^{alg}(z)$ has singularities at the endpoints of $\Gamma_k$ of order at most $|z-E_k|^{-1/4}$ or $|z-\bar{E}_k|^{-1/4}$.
	\end{enumerate}
\end{problem}
The RH problem \ref{RH-global} can be solved explicitly by
\begin{equation}
	N^{alg}(z)=\tilde N^{alg}(\infty)^{-1} \tilde N^{alg}(z),
\end{equation}
where $\tilde N^{alg}(z)=\left(\tilde N_{i j}^{(alg)}(z)\right)_{i, j=1,2}$ and
$$
\begin{aligned}
	& \tilde N_{s1}^{alg}(z)=\frac{1}{2}\left(\hat\nu(z)+\hat\nu(z)^{-1}\right) \frac{\Theta(\varphi(z)+\tilde c+d_s)}{\Theta(\varphi(z)+d_s)} ,\\
	& \tilde N_{s2}^{alg}(z)=\frac{1}{2}\left(\hat\nu(z)-\hat\nu(z)^{-1}\right) \frac{\Theta(\varphi(z)-\tilde c-d_s)}{\Theta(\varphi(z)-d_s)},\quad s=1,2.
\end{aligned}
$$
Herein, we define the auxiliary function
\[
\hat\nu(z)=\sqrt[4]{\prod_{k>\left[\frac{n}{2}\right]}\frac{z-E_k}{z-\bar E_k}},
\]
and denote by $\Theta$ the Riemann theta function as specified in \eqref{eq:THeze}. Furthermore, we introduce the parameter vector $\tilde c=\{\tilde c_k\,|\,\tilde c_k=-c_k-\frac{\alpha}{\pi},\, k>\left[\frac{n}{2}\right]+1\}$. The Abel map $\varphi\colon \mathcal{X} \to \mathbb{C}^{\left[\frac{n}{2}\right]}$ is defined by \eqref{eq:abel}, and we set $d_1=\varphi(\mathcal{D})+K,\,d_2=-\varphi(\mathcal{D})-K,$ where $K\in\mathbb{C}^{\left[\frac{n}{2}\right]}$ denotes the vector-valued Riemann constant and $\mathcal{D}$ signifies the pole divisor on the Riemann surface $\mathcal{X}$.

Additionally, the matrix function $N^{alg}(z)$ exhibits the following asymptotic behavior at infinity:
\begin{equation}\label{eq:Nglo-solution}
	N^{alg}(z)=I+\frac{N_1^{alg}}{z}+\mathcal{O}\left(z^{-2}\right),\qquad z\to\infty,
\end{equation}
and the reconstruction formula for $\tilde Q(x)$ is given by
\begin{equation}\label{tile-Q}
\tilde Q(x)=(N_1^{alg})_{12}=-\frac{i}{2}\operatorname{Im}\left(\sum_{k>\left[\frac{n}{2}\right]}E_k\right)\cdot\frac{\Theta\bigl(\varphi(\infty^+)+d_1\bigr)\,\Theta\bigl(\varphi(\infty^+)-\tilde{c}-d_1\bigr)}{\Theta\bigl(\varphi(\infty^+)-d_1\bigr)\,\Theta\bigl(\varphi(\infty^+)+\tilde{c}+d_1\bigr)}.
\end{equation}

\subsection{Solution of the Local Problem}
Within $U_{loc}$, the bound \eqref{eq:j-2-2} is merely pointwise
and lacks the uniformity necessary for $\mathcal{E}(z)$ to satisfy a
uniformly small-norm RH problem. Therefore, we introduce the local
parametrix $N^{loc}(z)$ matching the jumps of $N^{(1)}(z)$ on
$\Sigma^{loc} := \Sigma^{(1)} \cap U_{loc}$. The
construction of $N^{loc}(z)$ is formulated as follows.
\begin{problem}\label{RH-ee}
		Construct a meromorphic function $N^{loc}(z) :\mathbb{C}\setminus\Sigma^{loc}\to SL_2(\mathbb{C})$ such that:
	\begin{enumerate}
		\item $N^{loc}(z)=I+\mathcal{O}(z^{-1}),$\quad $|z|\to\infty$.
		\item 	For each  $z \in \Sigma^{loc}$, the boundary values $N^{loc}_\pm(z)$ satisfy the jump relation
		$$N^{loc}_+(z)=N^{loc}_-(z)J^{(loc)}(z),$$
		where $J^{l o c}(z)=\left.J^{(1)}(z)\right|_{z \in \Sigma^{ loc}}$.
	\end{enumerate}
\end{problem}
Consider the case $z\in D_\epsilon(\kappa^{\mathrm{R}}_{3})$, the jump matrix is estimated by
\begin{equation}
    \bigl\|J^{(1)}(z)-I\bigr\|_{L^\infty(\Sigma^{\mathrm{R}}_{\kappa_3})}=\mathcal{O}(t^{-1/2}),\qquad z\in\Sigma^{(1)}\cap D_\epsilon(\kappa^{R}_{3}).
\end{equation}

Consider the case $z\in D_{\rho_0}(z_{j})$, as in the mKdV case \cite{PD93}, the signature table dictates a factorization of $J(z)$ on $\mathbb{R} \backslash I$ that enables deformation of the RH problem for $N^{(1)}(z)$. The resulting triangular factors decay exponentially to $I$ as $t \rightarrow+\infty$, uniformly away from $ D_{\rho_0}(\kappa^{\mathrm{R}}_{3})$. Within this neighborhood, $e^{i t \theta(z)}$ admits the approximation
\begin{equation}\label{eq:t-lambda}
	t \theta(z)=\frac{4}{3}\lambda^3+\varpi \lambda+\mathcal{O}\left(t^{-1 / 3}\lambda^4\right),
\end{equation}
where $\lambda$ is the scaled spectral parameter
\begin{equation}\label{eq:trans-z-k}
	\varpi=2\frac{\theta^\prime(z_j)}{\sqrt[3]{\theta^{\prime\prime\prime}(z_j)}}t^{2 / 3}, \quad \lambda=\sqrt[3]{\frac{\theta^{\prime\prime\prime}(z_j)t}{8}}(z-z_j).
\end{equation}
The coefficients above are selected so that the scaled phase factor matches that of RH problem \ref{RH2-2} for a familiar equation. Indeed, the right-hand side of \eqref{eq:t-lambda} displays the phase factor from RH problem \ref{RH2-2}, associated with
\begin{align}
	&u_{{\varpi}{\varpi}}={\varpi}u+2u^2v,\\
	&v_{{\varpi}{\varpi}}={\varpi}v+2v^2u.
\end{align}
Thus, asymptotics in this region are expressible via Painlevé II transcendents in ${\varpi}$.

 In view of the structural properties exhibited by the preceding expansion, we define appropriate scaled spectral variables. These are designed to ensure correspondence between the coefficients of exponential terms and those appearing in the Painlevé II model RH problem in Appendix \ref{App-B}. Using the transform $z\mapsto\lambda$ in \eqref{eq:trans-z-k}, the RH problem \ref{RH-ee} for $N^{loc}(z)$ transforms into an RH problem for $N^{loc}(\lambda)$ with jump contour $L_j=L^t_{j}\cup \bar{L}^t_{j} (t=1,2)$, as shown in Figure \ref{ddd}. The basic estimate
 \begin{equation}\label{eq:tras}
 	\left\|\left(e^{2i t \theta(z)} r\left(z\right)-\mathrm{e}^{\frac{8}{3}i \lambda^3+2 i{\varpi}\lambda} r(z_j)\right)\right\|_{ \Sigma^{(1)}\setminus U_{loc}} \leq Ct^{-\frac{1}{3}},\quad C>0,
 \end{equation}
 which, in particular, implies that
 \begin{equation}
 	N_1^{loc}(\lambda)=\tilde N_1^{loc}(\lambda)+\mathcal{O}(t^{-1/3}),
 \end{equation}
 where $N_1^{loc}(\lambda)$ and $\tilde N_1^{loc}(\lambda)$ are the coefficients in the large $\lambda$ expansion of
 \begin{align}
 	N^{loc}(\lambda)=I+\frac{N_1^{loc}({\varpi})}{\lambda}+\mathcal{O}(\lambda^{-2}),\quad\tilde N^{loc}(\lambda)=I+\frac{\tilde N_1^{loc}({\varpi})}{\lambda}+\mathcal{O}(\lambda^{-2}),
 \end{align}
and $\tilde N^{loc}(\lambda)$ is the solution of the following RH problem.
\begin{problem}\label{RHtilde-N}
	Construct a meromorphic function $\tilde N^{loc}(\lambda) :\mathbb{C}\setminus L_{j}\to SL_2(\mathbb{C})$  such that:
	\begin{enumerate}
		\item $\tilde N^{loc}(\lambda)=I+\mathcal{O}(z^{-1}),$\quad $|z|\to\infty$.
		\item 	For each  $z \in L_j$, the boundary values $\tilde N^{loc}_\pm(\lambda)$ satisfy the jump relation
		$$\tilde N^{loc}_+(\lambda)=\tilde N^{loc}_-(\lambda)\tilde J^{loc}(\lambda),$$
where
		\begin{align}\label{eq:J(z)-loc}
			&\tilde J^{loc}(\lambda)=I+ \begin{cases}
		\frac{{r(z_j)}e^{2i\left(\frac{4}{3}\lambda^3+{\varpi}\lambda\right)}}{1+|r(z_j)|^2}\delta^{-2}(z_j^\mathrm{R})\sigma_-&z \in \bar L^{1}_{j}\cup\bar L_{j}^2, \\
			- \frac{\overline{r(\bar z_{j})}e^{-2i\left(\frac{4}{3}\lambda^3+{\varpi}\lambda\right)}}{1+|r(z_{j})|^2}\delta^2(z_{j})\sigma_+, & z \in L^{1}_{j}\cup L_{j}^2.  \end{cases}
		\end{align}
	\end{enumerate}
\end{problem}
\begin{figure}
	{
		\begin{minipage}{15cm}\centering
			\includegraphics[scale=0.276]{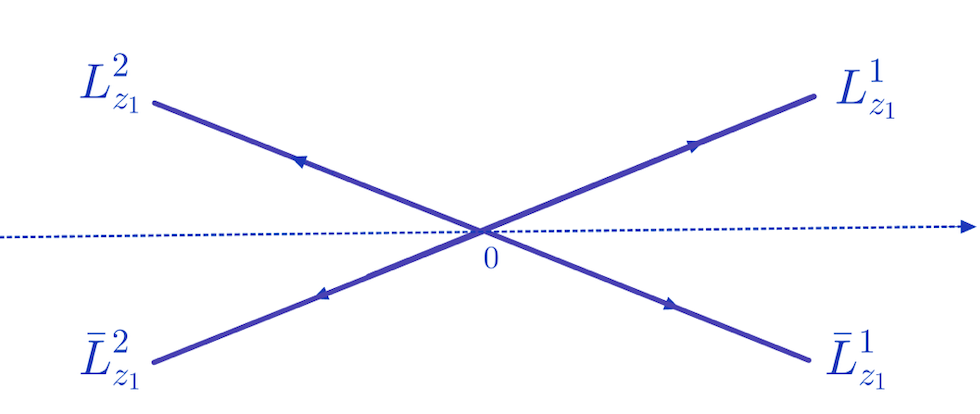}
	\end{minipage}}
	\caption{The jump contour is given by $L_{j}=\bigcup_{k=1}^{4}L_{j}^{k}$.}
	\label{ddd}
\end{figure}
In order to determine $\tilde{N}^{loc}(\lambda)$, we transform the jump matrix \eqref{eq:J(z)-loc} of RH problem \ref{RHtilde-N} into a form standard for studying the asymptotic behavior of Painlev\'{e} II transcendents. Following the approach in \cite{AAA06}, we find that the solution $\tilde{N}^{loc}(\lambda)$ can be expressed in terms of
\begin{equation}
 \tilde N^{loc}(\lambda) =\sigma_1 e^{-i\left(\frac{\pi}{4}+\frac{\varphi_0}{2}\right) \hat{\sigma}_3} N^{P}(\lambda)  \sigma_1,\quad\varphi_0=\arg r(z_j)
\end{equation}
where $N^{P}(\lambda)$ satisfies a standard Painlev\'{e} II model RH problem given in Appendix \ref{App-B} with parameter
\begin{equation}\label{eq:rrrho}
	\rho=i\frac{\overline{r(\bar{z}_{j})}e^{-2i\left(\frac{4}{3}\lambda^3+{\varpi}\lambda\right)}}{1+|r(z_{j})|^2}\delta^2(z_j).\end{equation}
Finally, as $\lambda\to\infty$, $N^{loc}(z)$ can be described by the following equation
\begin{equation}
	N^{loc}(z)=I+\frac{2N^{loc}_1({\varpi})}{\sqrt[3]{\theta^{\prime\prime\prime}(z_j)t}(z-z_j)}+\mathcal{O}\left(t^{-\frac{2}{3}+\epsilon}\right),
\end{equation}
where $N^{loc}_1({\varpi})$ is given in \eqref{eq:Pe-e} with
\begin{equation}\label{asym-u-q}
	u({\varpi})=2\left(\tilde N^{loc}(\lambda)\right)_{12}={\rm Im}\rho Ai({\varpi})+\mathcal{O}\left(\varpi^{-\frac{1}{4}} e^{-\frac{4}{3} \varpi^{3 / 2}}\right).\end{equation}.
 \subsection{Solution of the Small norm Problem}
Using the functions $N^{alg}(z)$ and $N^{loc}(z)$ defined by RH problems \ref{RH-global} and \ref{RH-ee}, respectively, we implicitly define an unknown function $\mathcal{E}(z)$ that is analytic in $\mathbb{C}\setminus\Sigma^{\mathcal{E}}$, where
\begin{equation}
	\Sigma^{\mathcal{E}}=\left(\Sigma^{(2)} \setminus\left(\Sigma^{loc} \cup \tilde\Gamma\right)\right) \cup \partial U_{loc},
\end{equation}
and we orient $\partial\Sigma^{\mathcal{E}}$ clockwise. It is straightforward to show that $\mathcal{E}(z)$ satisfies the following RH problem.
 \begin{problem}\label{RHee}
 	Construct a meromorphic function $\mathcal{E}(z):\mathbb{C}\setminus\Sigma^{\mathcal{E}}\to SL_2(\mathbb{C})$  such that:
 	\begin{enumerate}
 		\item $\mathcal{E}(z)=I+\mathcal{O}(z^{-1}),$\quad $|z|\to\infty$.
 		\item 	For each $z\in\Sigma^{\mathcal{E}}$, the boundary values $\mathcal{E}_\pm(z)$ satisfy the jump relation
 		$$\mathcal{E}_+(z)=\mathcal{E}_-(z)J_{\mathcal{E}}(z)$$
 		where
 		\begin{equation}\label{err-11}
 			J_{\mathcal{E}}(z)= \begin{cases}N^{alg}(z) J^{(1)}(z)\left(N^{alg}(z)\right)^{-1}, & z \in \Sigma^{(2)} \backslash {U}_{loc}, \\
 				N^{alg} (z)N^{loc}(z)\left(N^{alg}(z)\right)^{-1}, & z \in  \partial U_{loc},
 			\end{cases}
 		\end{equation}
 	\end{enumerate}
 \end{problem}
 Starting from \eqref{eq:j-2-2} and using\eqref{eq:tras} for $z\in\mathbb{C}\setminus\Sigma^{loc}$, and the boundedness of $N^{alg}(z)$ for $z\in D_{z^R_{j}}$, one finds that
 \begin{align}\label{eq:L-P}
 	\left|J_{\mathcal{E}}(z)-I\right|=\begin{cases}
 		\mathcal{O}\left(e^{-ct}\right),& z\in\Sigma^{(2)}\setminus 	U_{z_j^\mathrm{R}},\\
 		\mathcal{O}(t^{-1/3}), &z\in \partial 	U_{z_j^\mathrm{R}},
 	\end{cases}
 \end{align}
This uniformly vanishing bound on $J_{\mathcal{E}}(z)-I$ establishes RH problem \ref{RHee} as a small-norm RH problem, for which there is a walg known existence and uniqueness theorem \cite{PZ03,PDZ94,XZ89}.  In order to reconstruct the solution $q(x, t)$ of \eqref{eq:fNLS} we need the large $z$ behavior of the solution of RH problem \ref{RH-EE}. This will include the large $z$ expansion of $\mathcal{E}(z)$ which we give here.
 $$
\mathcal{E} (z)=I+\frac{1}{2 \pi i} \int_{\Sigma^{\mathcal{E}}} \frac{(I+\eta(s))\left(J_{\mathcal{E}}(s)-I\right)}{s-z} d s=I+\frac{\mathcal{E}_1}{z}+\mathcal{O}(z^{-2})
 $$
 where $\eta \in L^2\left(\Sigma^{\mathcal{E}}\right)$ is the unique solution of
 $
 \left(1-C_{\Sigma^\mathcal{E}}\right) \eta=C_{\Sigma^\mathcal{E}} I
 $ and
\begin{equation}
	\mathcal{E}_1=\frac{2iN^{loc}_1({\varpi})}{\sqrt[3]{\theta^{\prime\prime\prime}(z_{j}^\mathrm{R})t}(z-z_{j}^\mathrm{R})}+\mathcal{O}\left(t^{-1}\right).
\end{equation}

Inverting the sequence of transformations \eqref{tran-1} and \eqref{eq:trans-2}, the solution of RH problem \ref{RH2-2} is given by
\begin{equation}
	N(z)=\delta^{\sigma_3}(\infty)\mathcal{E}(z)N^{alg}(z)\delta^{-\sigma_3}(z)G^{-1}(z),\quad z\in\mathbb{C}\setminus U_{z^R_{1}}.
\end{equation}
Expanding as $z\to\infty$, we obtain
\begin{equation}
	N(z)=\delta^{\sigma_3}(\infty)\left(I+\frac{\mathcal{E}_1}{z}+\mathcal{O}(z^{-2})\right)\left(I+\frac{N_1^{alg}}{z}+\mathcal{O}(z^{-2})\right)\left(I+\frac{\delta_1}{z}+\mathcal{O}(z^{-2})\right)\delta^{-\sigma_3}(\infty),
\end{equation}
and consequently, by \eqref{tile-Q}, \eqref{err-11} and \eqref{asym-u-q},
\begin{align}\nonumber
	q(x,t) &= 2ie^{2i(f_0x+g_0t)}\lim_{z\to\infty}zN_{12}(z) = 2e^{2i(f_0x+g_0t)}\delta^2(\infty)\left(\left(\mathcal{E}_1\right)_{12}+\left(N_1^{alg}\right)_{12}\right) \\
	&= 2e^{2i(f_0x+g_0t)}\delta^2(\infty)\frac{u(\varpi)}{\sqrt[3]{\theta'''(z_{j})t}(z-z_{j}^\mathrm{R})} - 2e^{2i(f_0x+g_0t)}\delta^2(\infty)\tilde Q(x,t).
\end{align}
This completes the proof of Theorem \ref{theorem-1}.
\section{Long-time Asymptotics in Even-genus Backgrounds}\label{sec-4}
In this section, we restrict our attention to the case of even genus, with the distribution of stationary phase points determined by
\begin{equation}\label{h-z-xii}
	h(z;\xi) = 4(z-\kappa_1^{\mathrm{R}})(z-\kappa^{\mathrm{R}}_2)\prod_{s=1}^{\frac{n}{2}}\left[(z-\kappa_s^{\mathrm{C}})(z-\bar{\kappa}_s^{\mathrm{C}})\right],
\end{equation}
Among these, the number of branch cuts $\Gamma\setminus\Gamma_{\frac{n}{2}}$ in the left and right half-planes is equal, while the remaining cut $\Gamma_{\frac{n}{2}}$ is located either in the left half-plane or in the right half-plane.

The distinction between the two distributions becomes apparent in a neighborhood of $\xi=0$, specifically in the regime $|\xi|<d$ for some $d>0$,  where the infinite branch coalesces with the first positive stationary point. When an even number of branch cuts $\Gamma_k$ lies in the left half-plane, the first positive stationary point increases monotonically with decreasing $\xi$ as shown in Figure~\ref{FIguree45}. Conversely, when an odd number of cuts $\Gamma_k$ lies in the left half-plane, the first positive stationary point decreases monotonically with decreasing $\xi$ as shown in Figure~\ref{ffuu}. In this section, we restrict our attention to the latter case, where an odd number of cuts $\Gamma_k$ lies in the left half-plane.
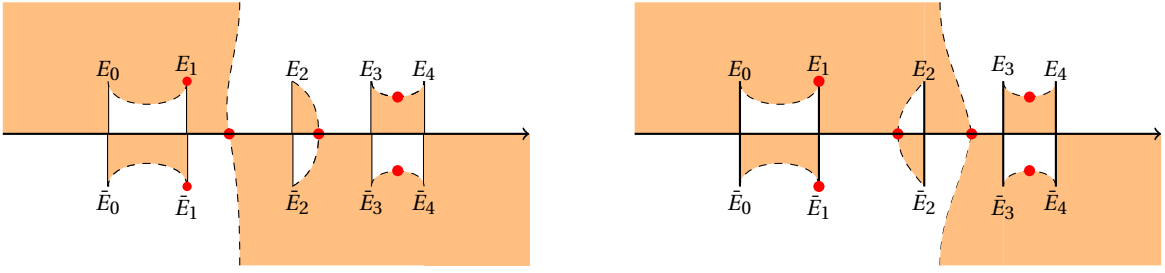
\begin{figure}[htp]
	\begin{center}
		\begin{tikzpicture}[scale=0.7]			
			\fill[orange!50](1,2.5)--(1,0)--(3,0)--(3,1)to[out=-90, in=-90]
			(4.5,1)--(4.5,0)--(6,0)
			to[out=90, in=-130]
			(6.5,1)--(6.5,2.5)--(1,2.5)--cycle;
			
				\fill[orange!50](3,0)--(3,-1)to[out=90, in=90]
			(4.5,-1)--(4.5,0)--cycle;

			\draw[dashed](3,1)to[out=-90, in=-90]
			(4.5,1);
			\draw[dashed](6,0)
			to[out=90, in=-130]
			(6.5,1);
				\fill[orange!50](6,0)
				to[out=-90, in=130]
				(6.5,-1)--(6.5,0)
				--cycle;
			\draw[dashed](6,0)
			to[out=-90, in=130]
			(6.5,-1);
			
				\fill[orange!50](6.5,2.5)--(6.8,2.5)
			to[out=-90, in=100]
			(7.4,0)--(6.5,0)--cycle;
			\draw[dashed](6.8,2.5)
			to[out=-90, in=100]
			(7.4,0);
			
				\draw[dashed](6.8,-2.5)
			to[out=90, in=-100]
			(7.4,0);
			
				\fill[orange!50](6.8,-2.5)
				to[out=90, in=-100]
				(7.4,0)--(8,0)--(8,-2.5)--cycle;
				
					\fill[orange!50](11,-2.5)--(11,0)--(9,0)--(9,-1)
				to[out=90, in=90]
				(8,-1)--(8,-2.5)--cycle;
				\draw[dashed](9,-1)
				to[out=90, in=90]
				(8,-1);

				\fill[orange!50](8,0)--(8,1)
			to[out=-90, in=-90]
			(9,1)--(9,0)--cycle;
			\draw[dashed](8,1)
			to[out=-90, in=-90]
			(9,1);

				\draw[thick](-9,-1)--(-9,1);
					\draw[thick](-7.5,-1)--(-7.5,1);
								\draw[thick](-5.5,-1)--(-5.5,1);
							\draw[thick](-4,-1)--(-4,1);
							\draw[thick](-3,-1)--(-3,1);
							
					\fill[orange!50](-4,1)to[out=-90, in=-90]
			(-3,1)--(-3,0)--(-4,0)--cycle;
			
			\draw[dashed](-4,1)to[out=-90, in=-90]
			(-3,1);
								\fill[orange!50](-5.5,1)to[out=-30, in=90]
						(-5,0)--(-5.5,0)--cycle;
							
					\draw[dashed](-5.5,1)to[out=-30, in=90]
					(-5,0);	\draw[dashed](-4,1)to[out=-90, in=-90]
					(-3,1);
							\draw[dashed](-5.5,-1)to[out=30, in=-90]
						(-5,0);	
							
									\draw[dashed](4.5,-1)to[out=90, in=90]
							(3,-1);

		\fill[orange!50]
		(-11,2.5)--(-6.5,2.5)
		to[out=-90, in=100] (-6.7,0)
		--(-7.5,0)--(-7.5,1)to[out=-90, in=-90] (-9,1)--(-9,0)--(-11,0)
		--(-11,2.5)-- cycle;			
			\draw[dashed](-7.5,-1)to[out=90, in=90] (-9,-1);

		\fill[orange!50]
	(-6.5,-2.5)
		to[out=90, in=-80] (-6.7,0)
		--(-5.5,0)--(-5.5,-2.5)--(-5.5,-1)to[out=30, in=-90](-5,0)-- (-4,0)--(-4,-2.5)-- (-4,-1)to[out=90, in=90]
		(-3,-1)--(-3,-2.5)--
		cycle;			
	
	\fill[orange!50]
(-3,-2.5)--(-1,-2.5)--(-1,0)--(-3,0)--
	cycle;			
	
			\draw[dashed](-5.5,-1)to[out=30, in=-90]
		(-5,0);	
			\draw[dashed](-4,-1)to[out=90, in=90]
			(-3,-1);

		\draw[dashed]	(-6.5,2.5)
		to[out=-90, in=100] (-6.7,0) ;				
		\draw[dashed]	(-7.5,1)to[out=-90, in=-90] (-9,1) ;

			\fill[orange!50](-7.5,0)--
	(-7.5,-1)to[out=90, in=90] (-9,-1)--(-9,0)--cycle;
	
		\draw[dashed](-7.5,-1)to[out=90, in=90] (-9,-1);
	
			\draw[dashed]	(-6.5,-2.5)
			to[out=90, in=-80] (-6.7,0) ;

									\draw[thick](9,-1)--(9,1);
							\draw[thick](8,-1)--(8,1);
							\draw[thick](6.5,-1)--(6.5,1);
							\draw[thick](4.5,-1)--(4.5,1);
							\draw[thick](3,-1)--(3,1);

				\fill[red] (-7.5,1) circle (2.5pt);			
				\fill[red] (-7.5,-1) circle (2.5pt);
								\fill[red] (-6.7,0) circle (3pt);
								\fill[red] (-5,0) circle (3pt);
									\fill[red] (-3.5,0.7) circle (3pt);
										\fill[red] (-3.5,-0.7) circle (3pt);
							\fill[red] (4.5,1) circle (3pt);		
								\fill[red] (4.5,-1) circle (3pt);				
										\fill[red] (6,0) circle (3pt);				
											\fill[red] (7.4,0) circle (3pt);				
								\fill[red] (8.5,0.7) circle (3pt);				
										
										\fill[red] (8.5,-0.7) circle (3pt);				
						
			\draw[thick,->](-11,0)--(-1,0);
			\draw[thick,->](1,0)--(11,0);			
			
				\node  at (-9,1.2)  { \footnotesize$E_0$};
					\node  at (-7.5,1.3)  { \footnotesize$E_1$};
						\node  at (-5.4,1.2)  { \footnotesize$E_2$};
						\node  at (-4,1.2)  { \footnotesize$E_3$};
							\node  at (-3,1.2)  { \footnotesize$E_4$};
						
							\node  at (-9,-1.3)  { \footnotesize$\bar E_0$};
						\node  at (-7.5,-1.4)  { \footnotesize$\bar E_1$};
						\node  at (-5.4,-1.3)  { \footnotesize$\bar E_2$};
						\node  at (-4,-1.3)  { \footnotesize$\bar E_3$};
						\node  at (-3,-1.3)  { \footnotesize$\bar E_4$};

							\node  at (9,1.2)  { \footnotesize$E_4$};
						\node  at (8,1.3)  { \footnotesize$E_3$};
						\node  at (6.5,1.2)  { \footnotesize$E_2$};
						\node  at (4.5,1.3)  { \footnotesize$E_1$};
						\node  at (3,1.2)  { \footnotesize$E_0$};
						
							\node  at (9,-1.3)  { \footnotesize$\bar E_4$};
						\node  at (8,-1.4)  { \footnotesize$\bar E_3$};
						\node  at (6.5,-1.3)  { \footnotesize$\bar E_2$};
						\node  at (4.5,-1.4)  { \footnotesize$\bar E_1$};
						\node  at (3,-1.3)  { \footnotesize$\bar E_0$};
						
		\end{tikzpicture}
	\end{center}
		\caption{The infinite cut coalesces with the real stationary phase point as $|\xi|<d$ for even left-half-plane cuts.}
	\label{FIguree45}
\end{figure}

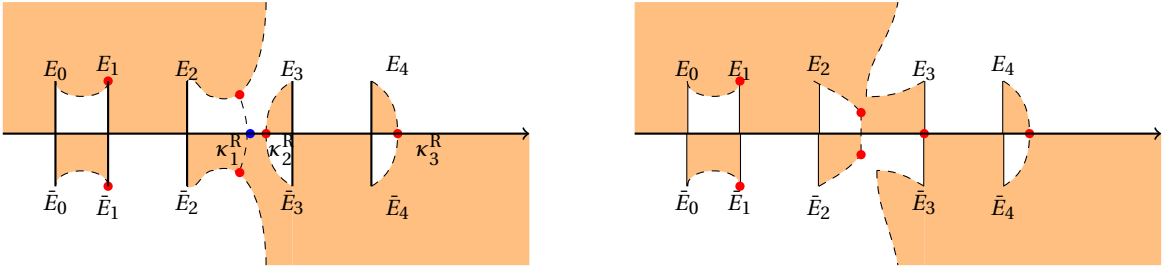
\begin{figure}[htp]
	
	\begin{center}
		\begin{tikzpicture}[scale=0.7]
			\draw[thick](8,-1)--(8,1);
			\draw[thick](6.5,-1)--(6.5,1);
			\draw[thick](4.5,-1)--(4.5,1);
			\draw[thick](3,-1)--(3,1);
					\draw[thick](2,-1)--(2,1);

						\fill[orange!50](1,2.5)--(1,0)--(2,0)--(2,1)to[out=-90, in=-90] (3,1)--(3,0)--(4.5,0)--(4.5,2.5)
			--cycle;		
			\draw[dashed](2,1)to[out=-90, in=-90] (3,1);
			
				\fill[orange!50](4.4,2.5)--(4.4,1)to[out=-30, in=130] (5.3,0.4)--
				(5.3,0)--(6.5,0)--((6.5,1)to[out=-150, in=0] (5.4,0.7)to[out=90, in=-100] (6,2.5);
			
				\draw[dashed](4.4,1)to[out=-30, in=130] (5.3,0.4);

		\draw[dashed](6.5,1)to[out=-150, in=0] (5.4,0.7);
					\draw[dashed](5.4,0.7)to[out=90, in=-100] (6,2.5);

				\fill[orange!50](4.5,-1)to[out=30, in=-130] (5.3,-0.4)--
				(5.3,0)--(4.5,0)--cycle;
			\draw[dashed](4.5,-1)to[out=30, in=-130] (5.3,-0.4);
				\draw[dashed](5.3,0.4)--(5.3,-0.4);

				\fill[orange!50](2,-1)to[out=90, in=90] (3,-1)--(3,0)--(2,0)
			--cycle;		
			\draw[dashed](2,-1)to[out=90, in=90] (3,-1);
			
				\fill[orange!50](6.5,-2.5)--
			(6.5,-1)to[out=150, in=0] (5.6,-0.7)to[out=-90, in=100] (6,-2.5)
			--cycle;		
			
			\draw[dashed]	(6.5,-1)to[out=150, in=0] (5.6,-0.7);
			
			\draw[dashed](5.6,-0.7)to[out=-90, in=100] (6,-2.5);

				\fill[orange!50](6.5,-2.5)--
			(6.5,-0)--(8,0)--(8,-1)
			to[out=30, in=-90] (8.5,0)--(11,0)--(11,-2.5)
			--cycle;		
			\draw[dashed](8,-1)
			to[out=30, in=-90] (8.5,0);
			
				\fill[orange!50](8,1)to[out=-30, in=90] (8.5,0)--(8,0)
			--cycle;		
			
		\draw[dashed](8,1)to[out=-30, in=90] (8.5,0);

			\fill[orange!50]
			(-11,2.5)--(-6,2.5)	to[out=-90, in=-20] (-7,0.7)
			to[out=180, in=30]
(-7.5,1)--(-7.5,0)--(-9,0)--(-9,1)to[out=-90, in=-90] (-10,1)--(-10,0)--(-11,0)
			--(-11,2.5)-- cycle;			
				\draw[dashed](-6,2.5)	to[out=-90, in=-20] (-7,0.7);
				\draw[dashed](-7,0.7)
				to[out=180, in=30]
				(-7.5,1);
				\draw[dashed](-9,1)to[out=-90, in=-90] (-10,1);

				\fill[orange!50]
		(-6,-2.5)	to[out=90, in=20] (-7,-0.7)
			to[out=180, in=-30]
			(-7.5,-1)--(-7.5,0)--(-6,0)to[out=-90, in=150] (-5.5,-1)--(-5.5,-2.5)-- cycle;			
			\draw[dashed]	(-6,-2.5)	to[out=90, in=20] (-7,-0.7);
				\draw[dashed](-6,0)to[out=-90, in=150] (-5.5,-1);
				\draw[dashed](-7,-0.7)
				to[out=180, in=-30]
				(-7.5,-1);

				\fill[orange!50]
			(-6,0)to[out=90, in=-150] (-5.5,1)--(-5.5,0)-- cycle;			
			\draw[dashed]	(-6,0)to[out=90, in=-150] (-5.5,1);
		\fill[orange!50]
(-5.5,-2.5)--(-5.5,0)--(-4,0)
--(-4,-1)to[out=30, in=-90]
(-3.5,0)--(-1,0)--(-1,-2.5)--
 cycle;			
 \draw[dashed](-4,-1)to[out=30, in=-90]
 (-3.5,0);
				\fill[orange!50]
		(-4,1)to[out=-30, in=90]
			(-3.5,0)--(-4,0)--
			cycle;			
			\draw[dashed]	(-4,1)to[out=-30, in=90]
			(-3.5,0);
				\fill[orange!50]
(-9,0)--(-9,-1)to[out=90,in=90]
			(-10,-1)--(-10,0)--cycle;			
			\draw[dashed](-9,-1)to[out=90,in=90]
			(-10,-1);
			
					\draw[dashed](-6.5,-0.7)to[out=70,in=-70]
			(-6.5,0.7);
				\fill[blue] (-6.3,0) circle (2.5pt);			
			
				\fill[red] (-9,1) circle (2.5pt);			
			\fill[red] (-9,-1) circle (2.5pt);			
			\fill[red] (-6.5,0.74) circle (2.5pt);			
			\fill[red] (-6.5,-0.74) circle (2.5pt);			
			\fill[red] (-6,0) circle (2.5pt);			
				\fill[red] (-3.5,0) circle (2.5pt);			
			
				\fill[red] (3,1) circle (2.5pt);			
				\fill[red] (3,-1) circle (2.5pt);			
				\fill[red] (5.3,0.4) circle (2.5pt);			
				\fill[red] (5.3,-0.4) circle (2.5pt);	
					\fill[red] (6.5,0) circle (2.5pt);	
					\fill[red] (8.5,0) circle (2.5pt);

				\node  at (2,1.2)  { \footnotesize$E_0$};
			\node  at (8,1.3)  { \footnotesize$E_4$};
			\node  at (6.5,1.2)  { \footnotesize$E_3$};
			\node  at (4.5,1.3)  { \footnotesize$E_2$};
			\node  at (3,1.2)  { \footnotesize$E_1$};

				\node  at (2,-1.3)  { \footnotesize$\bar E_0$};
			\node  at (8,-1.4)  { \footnotesize$\bar E_4$};
			\node  at (6.5,-1.3)  { \footnotesize$\bar E_3$};
			\node  at (4.5,-1.4)  { \footnotesize$\bar E_2$};
			\node  at (3,-1.3)  { \footnotesize$\bar E_1$};

					\node  at (-10,-1.3)  { \footnotesize$\bar E_0$};
			\node  at (-9,-1.4)  { \footnotesize$\bar E_1$};
			\node  at (-5.5,-1.3)  { \footnotesize$\bar E_3$};
			\node  at (-3.5,-1.4)  { \footnotesize$\bar E_4$};
			\node  at (-7.5,-1.3)  { \footnotesize$\bar E_2$};
			
	\node  at (-6.7,-0.3)  { \footnotesize$\kappa^{\mathrm{R}}_1$};
	\node  at (-5.7,-0.3)  { \footnotesize$\kappa^{\mathrm{R}}_2$};

	\node  at (-2.9,-0.3)  { \footnotesize$\kappa^{\mathrm{R}}_3$};

				\node  at (-10,1.2)  { \footnotesize$ E_0$};
			\node  at (-9,1.3)  { \footnotesize$ E_1$};
			\node  at (-5.5,1.2)  { \footnotesize$ E_3$};
			\node  at (-3.5,1.3)  { \footnotesize$ E_4$};
			\node  at (-7.5,1.2)  { \footnotesize$ E_2$};

				\draw[thick](-9,-1)--(-9,1);
			\draw[thick](-7.5,-1)--(-7.5,1);
			\draw[thick](-5.5,-1)--(-5.5,1);
			\draw[thick](-4,-1)--(-4,1);
			\draw[thick](-10,-1)--(-10,1);
			\draw[thick,->](-11,0)--(-1,0);
			\draw[thick,->](1,0)--(11,0);			
		\end{tikzpicture}
	\end{center}
	\caption{The infinite cut coalesces with the real stationary phase point as $|\xi|<d$ for odd left-half-plane cuts.}
	\label{ffuu}
\end{figure}
Starting from the distribution of stationary phase points as depicted in the first picture of Figure~\ref{ffuu}, we construct an $n+1$ genus Riemann surface by adding the branch cut $[\bar{E}_{00},E_{00}]$ as shown in Figure~\ref{hh}; this operation generates a real stationary phase point $\kappa^{R}_1$. Subsequently, as $\xi$ decreases, $\kappa^{R}_1$ moves along the real axis and coalesces with the existing real stationary phase point $\kappa^{R}_2$. This collision produces a pair of complex conjugate stationary phase points $\bar{E}_{01}$ and $E_{01}$, allowing us to further increase the genus by adding the corresponding branch cut $[\bar{E}_{01},E_{01}]$ as shown in Figure~\ref{hh}. Consequently, we obtain an $n+2$ genus Riemann surface characterized by these two additional stationary phase points, whose existence is proved in \cite{BL21}. We introduce the regularized phase
function
\begin{equation}\label{eq:theta_prime}
	d\theta(z)=\frac{(z-\kappa_1^{\mathrm{R}})(z-\kappa_2^{\mathrm{R}})(z-E_{00})(z-\bar E_{00})(z-E_{01})(z-\bar E_{01})\prod_{s=1}^{\frac{n-2}{2}}\left[(z-\kappa_s^{\mathrm{C}})(z-\bar{\kappa}_s^{\mathrm{C}})\right]}{w(z)\sqrt{(z-E_{00})(z-\bar E_{00})(z-E_{01})(z-\bar E_{01})}}dz.
\end{equation}
\begin{remark}
Herein, the assumption regarding the stationary phase point distribution in \eqref{h-z-xii}
is validated by the following construction. Consider the case of odd genus in the
left half-plane: for $|\xi|<d$, the coalescence of the infinite branch with
$\Gamma_{\frac{n}{2}}$ admits two possible topologically distinct scenarios
(Figures~\ref{tttt} and \ref{TTTT}, among others), depending on the choice of branch
configuration $\{B_k,A_k\}_{k=0}^n$.

The first scenario corresponds to the case where the complex stationary phase points on the cuts $\Gamma_{\frac{n}{2}}$ and $\Gamma_{\frac{n}{2}+1}$ coalesce \emph{before} the infinite branch coalesces with $\Gamma_{\frac{n}{2}}$ (Figure~\ref{tttt}). In this process, the coalescence generates two real stationary phase points moving in opposite directions as $\xi$ decreases: the left real stationary phase point coalesces with the infinite branch, while the right stationary phase point moves across $\Gamma_{\frac{n}{2}+1}$, contributing $\mathcal{O}(t^{-1/2})$ on both sides of $\Gamma_{\frac{n}{2}+1}$. Without loss of generality, we consider the case where the stationary phase point lies to the right of\, $\Gamma_{\frac{n}{2}+1}$.
\begin{figure}[htp]
	\subfloat{\label{a}
		\begin{minipage}[t]{0.3\linewidth}
			\centering
			\includegraphics[width=2in]{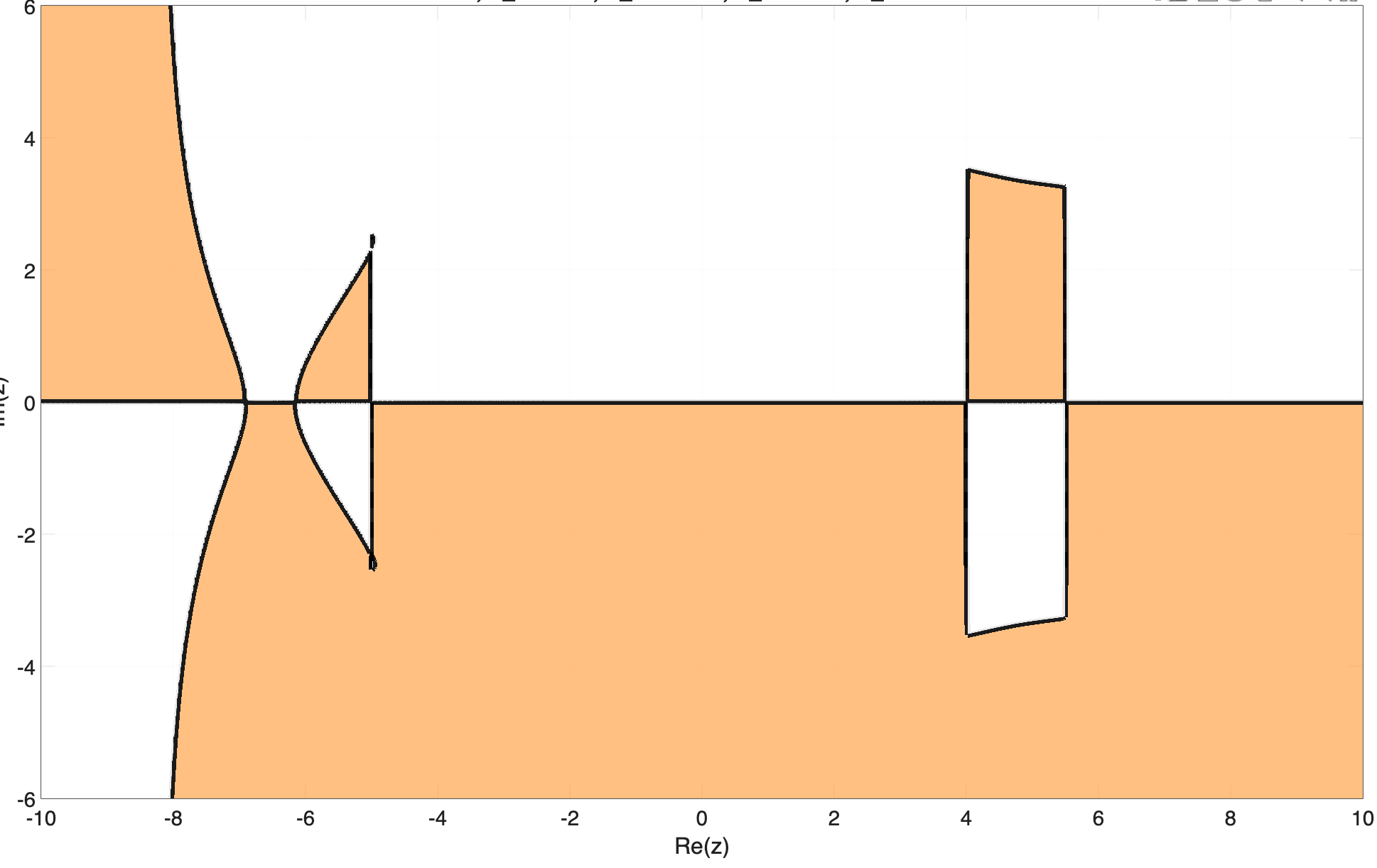}
		\end{minipage}
	}\quad
	\subfloat{\label{b}  
		\begin{minipage}[t]{0.3\linewidth}
			\centering
			\includegraphics[width=2in]{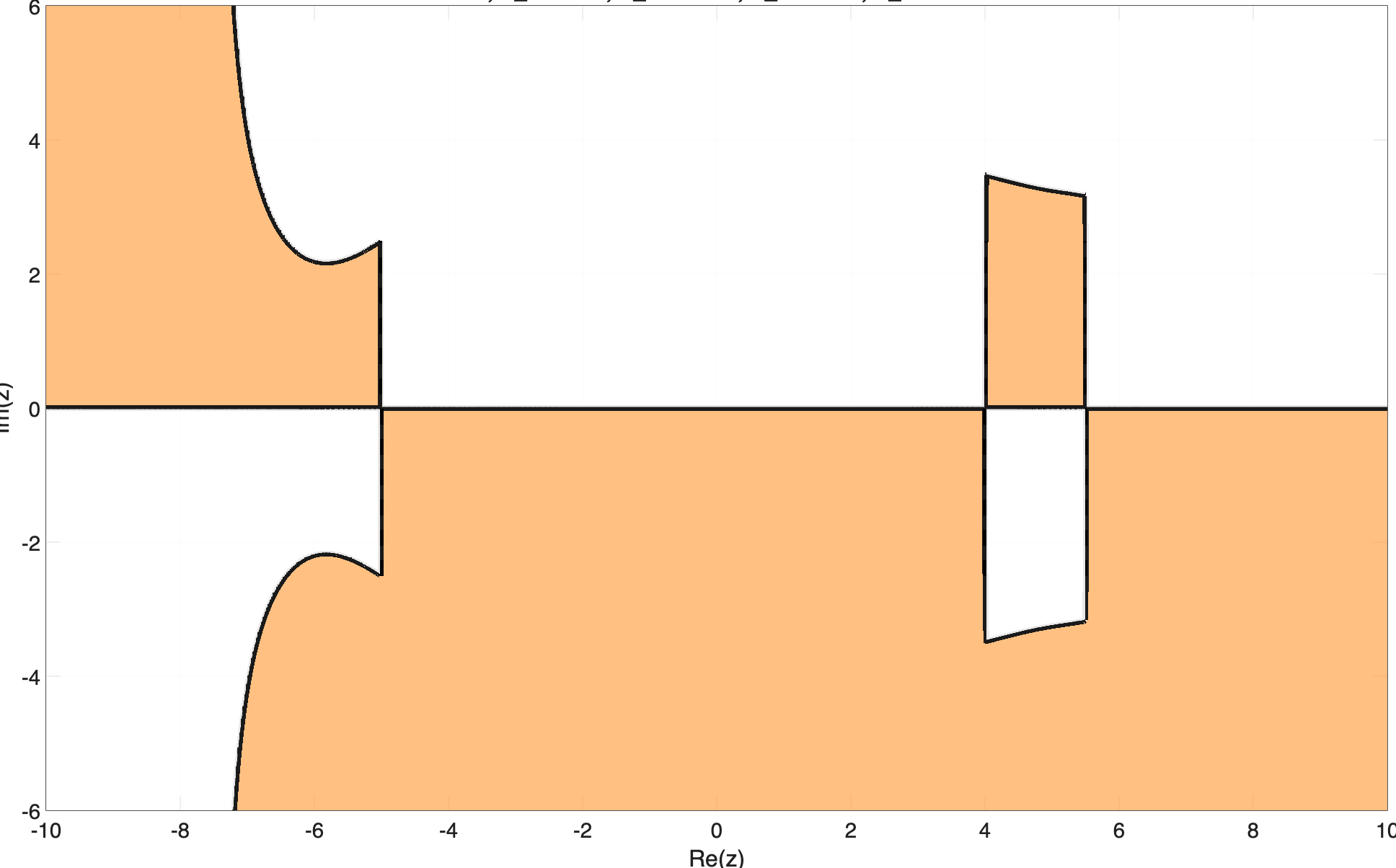}
		\end{minipage}
	}
	\subfloat{\label{c}\begin{minipage}[t]{0.4\linewidth}
			\centering
			\includegraphics[width=2in]{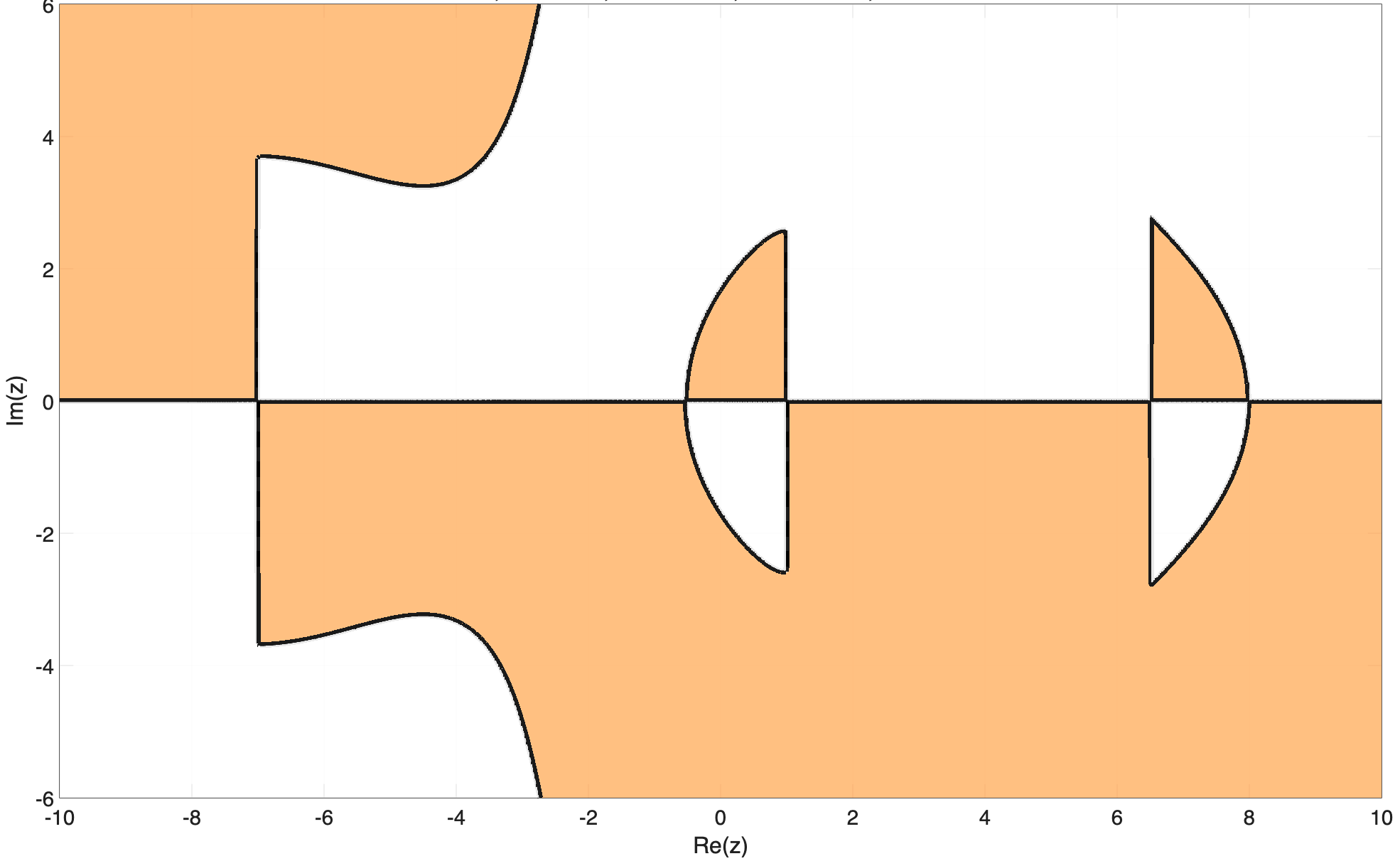}
	\end{minipage}}
	\caption{Complex-complex coalescence precedes infinity branch coalescence for the finite branch when $|\xi|<d$.}
	\label{tttt}
\end{figure}

The second scenario corresponds to the infinite branch coalescing with
$\Gamma_{\frac{n}{2}}$ \emph{before} the two complex stationary phase points coalesce
(Figure~\ref{TTTT}). In this paper, we restrict our analysis to the first scenario.
\begin{figure}[htp]
	\subfloat{\label{a}
		\begin{minipage}[t]{0.3\linewidth}
			\centering
			\includegraphics[width=2in]{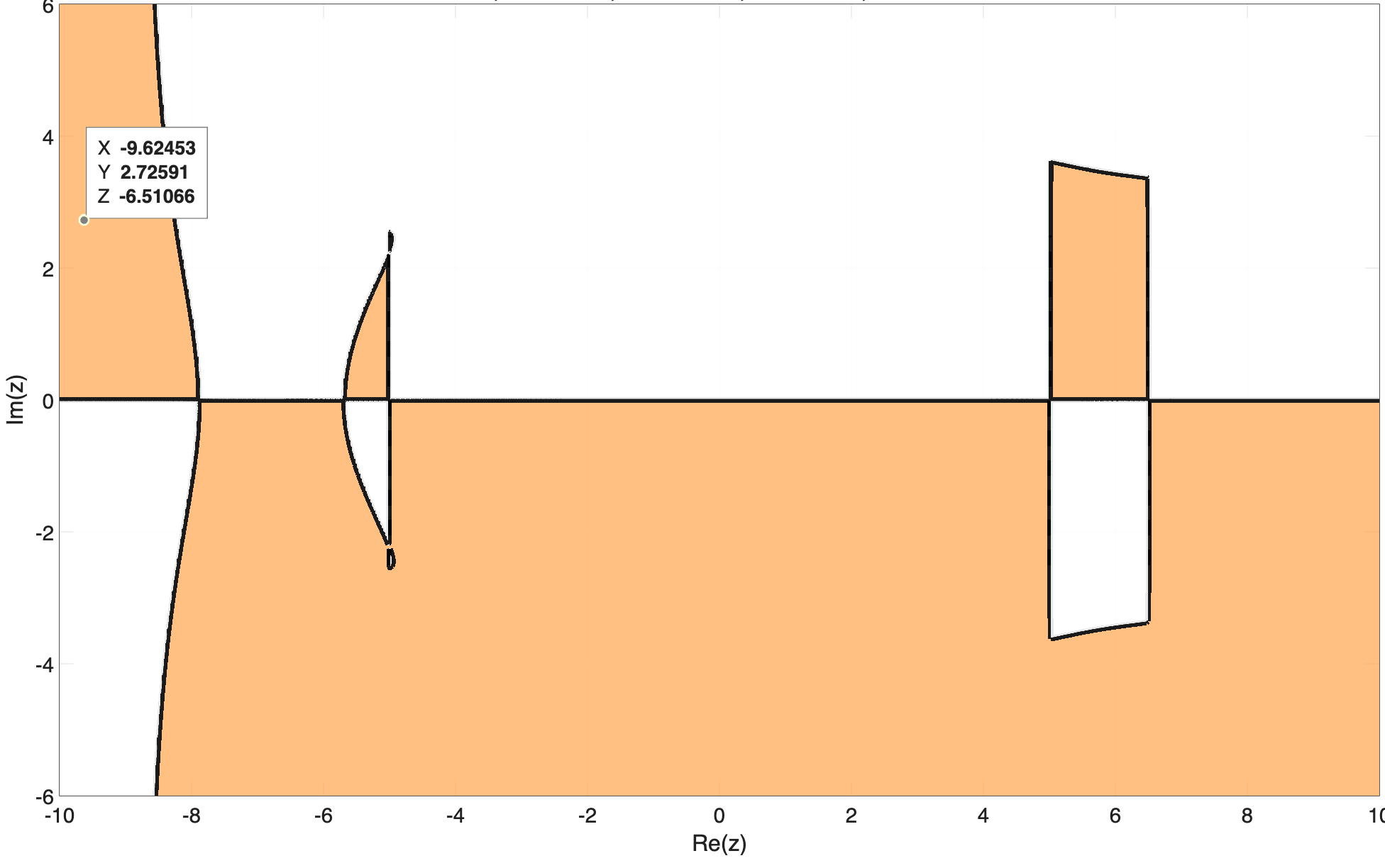}
		\end{minipage}
	}\quad
	\subfloat{\label{b}  
		\begin{minipage}[t]{0.3\linewidth}
			\centering
			\includegraphics[width=2in]{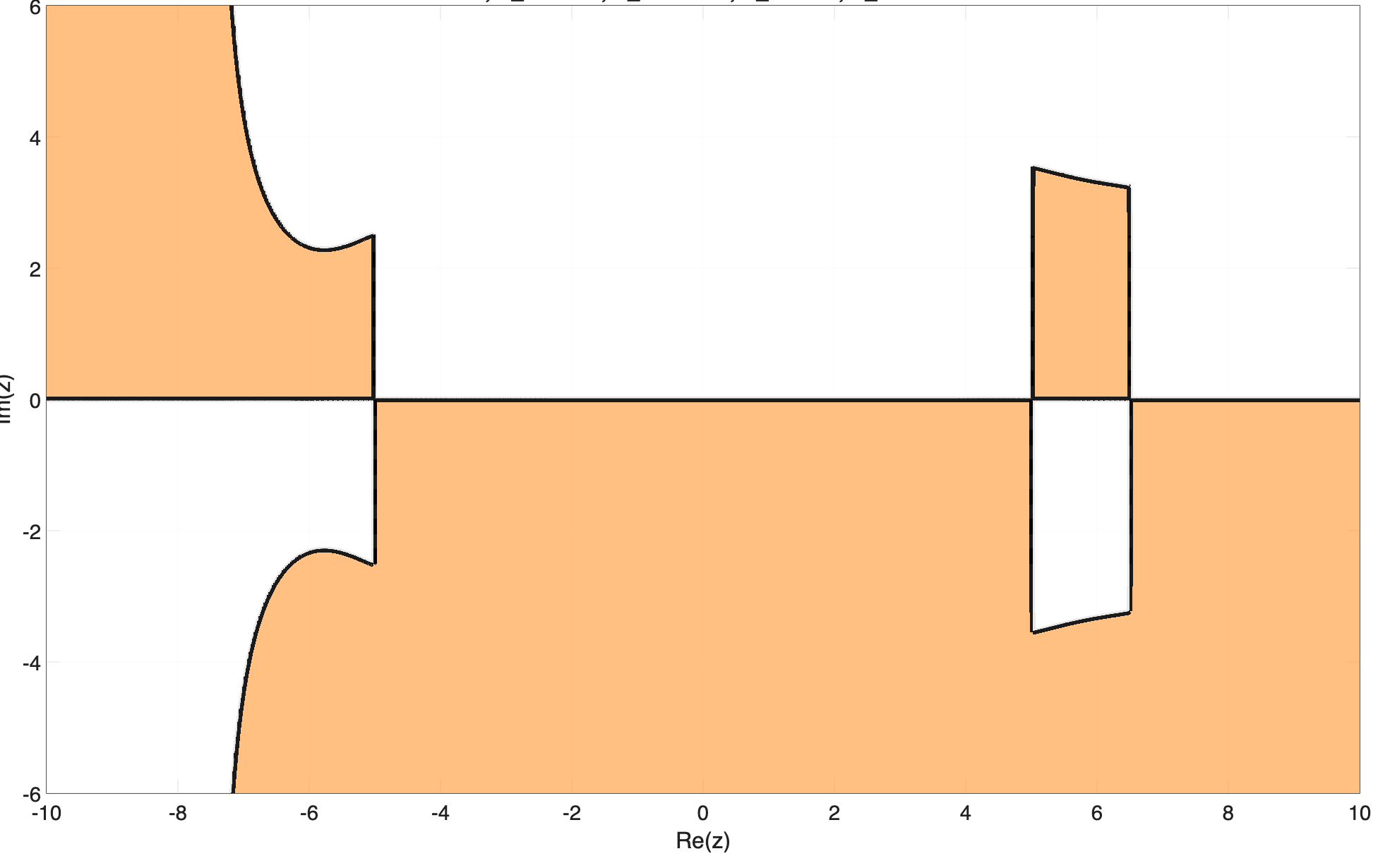}
		\end{minipage}
	}
    \subfloat{\label{c}\begin{minipage}[t]{0.4\linewidth}
			\centering
			\includegraphics[width=2in]{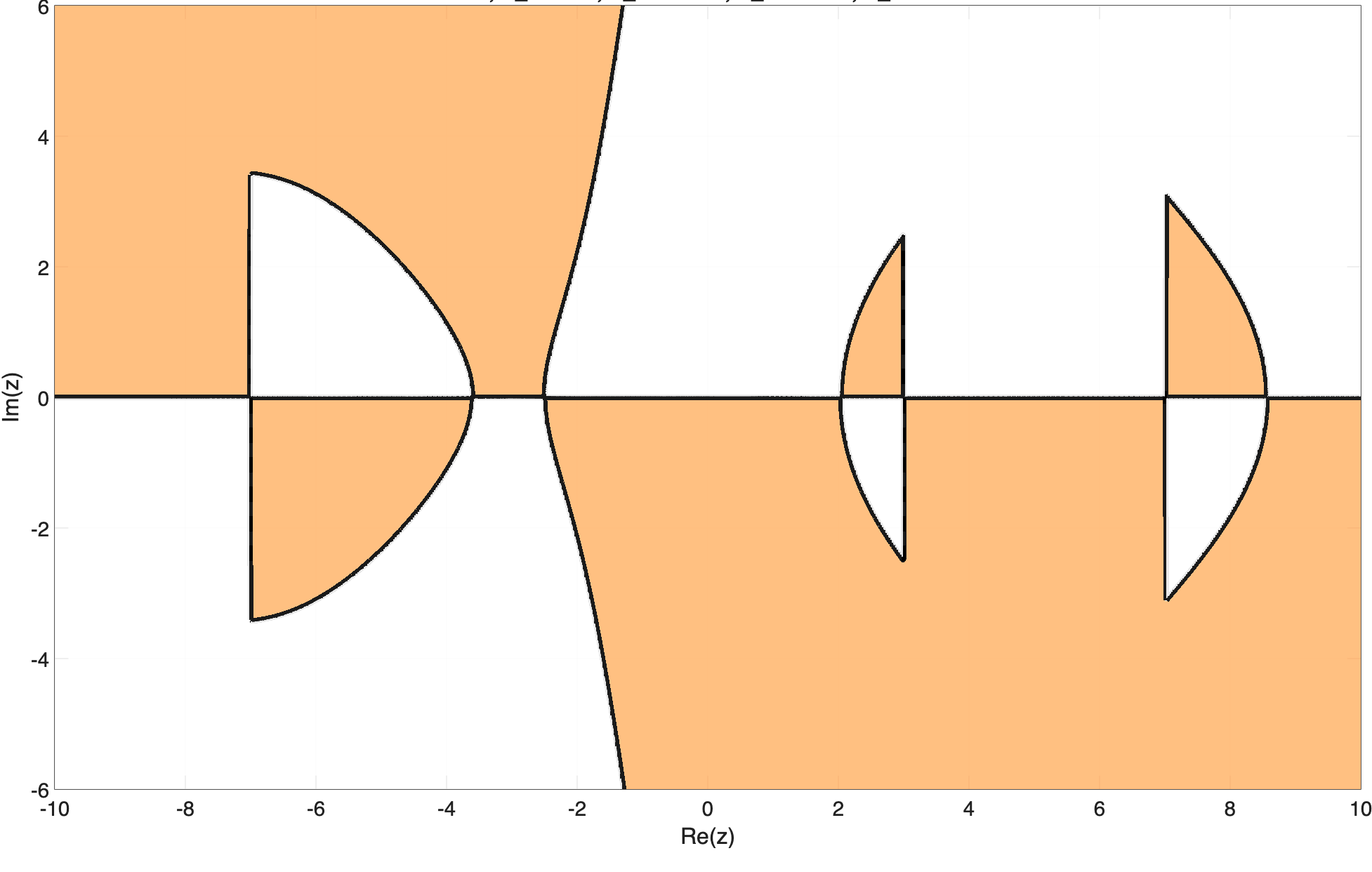}
		\end{minipage}}
	\caption{The infinity branch coalescence occurs prior to the complex-complex coalescence with the finite branch as $|\xi|<d$.}
\label{TTTT}
\end{figure}
\end{remark}
We now introduce the following interpolation transformation to construct a regular RH problem:
\begin{equation}\label{eq:N-1-two}
	N^{(1)}(z)=\delta^{-\sigma_3}(\infty)N(z)\delta^{\sigma_3}(z)G(z),
\end{equation}
where $\delta(z)$ is given by
\begin{equation}\label{eq:ddeltaaa}
	\delta(z)=\tilde{\nu}(z)\exp\left\{\frac{w(z)}{2\pi i}\int_{I}\frac{\log\left(1+|r(s)|^2\right)}{w(s)(s-z)}\,\mathrm{d}s\right\},\quad I=(-\infty,\kappa^{\mathrm{R}}_1)\cup(B_{\frac{n}{2}+1},\kappa^{\mathrm{R}}_2),
\end{equation}
with $\tilde{\nu}(z)=\prod_{k\ge\frac{n}{2}}\sqrt[4]{\frac{z-\bar{E}_k}{z-E_k}}$, and where the branch of the logarithm is taken to be principal. The function $\delta(z)$ defined by \eqref{eq:ddeltaaa} satisfies the following scalar RH problem.
\begin{problem}
Construct a scalar function $\delta(z): \mathbb{C}\setminus\left(I\cup(\tilde\Gamma_0=\bigcup_{k\ge \frac{n}{2}}\Gamma_k)\right)\to\mathbb{C}$ such that:
\begin{enumerate}
	\item[(i)] $\delta(z)=\delta(\infty)\left(1+\frac{I_1}{z}+O\left(\frac{1}{z^2}\right)\right)$ as $z\to\infty$, where
	\begin{align}\label{delt-ing}
		\delta(\infty)=\exp\left\{-\frac{1}{2\pi i}\int_{I}\frac{s^n\log(1+|r(s)|^2)}{w(s)}\,ds\right\},\quad
		I_1=-\frac{1}{2\pi i}\int_{I}\frac{s^n(s+\hat f_n)\log(1+|r(s)|^2)}{w(s)}\,ds.
	\end{align}
	\item[(ii)] For each $z\in I\cup\tilde\Gamma_0$, the boundary values $\delta_\pm(z)$ satisfy
	\begin{equation}
		\delta_+(z)=\delta_-(z)\begin{cases}
			1+|r(z)|^2, & z\in I,\\
			i&z\in \tilde\Gamma_0.
		\end{cases}
	\end{equation}
	\item[(iii)] $\delta(z)$ obeys the symmetry
	\begin{equation}
		\delta(z)=\overline{\delta(\bar{z})}^{-1}.
	\end{equation}
	
	\item[(iv)] $\delta(z)$ and $\delta^{-1}(z)$ are bounded analytic functions on $\mathbb{C}\setminus (I \cup_{k<\frac{n}{2}}B_{k})\cup\tilde\Gamma_0)$.
\end{enumerate}
\end{problem}
The function $G(z)$ defined  in \eqref{eq:N-1-two} as
\begin{equation}\label{eq:G(z)}
	G(z)=I+ \begin{cases}V^{-1}_2=	- r(z) e^{2 i t \theta(z)}\sigma_-,\,\,\,\,\,\,\,\,  z \in \Omega_{\kappa^{\mathrm{R}}_1}^1\cup\Omega_{\kappa^{\mathrm{R}}_2}^1,\quad
	V^{-1}_4=\frac{\overline{ r(\bar z)} e^{-2 i t \theta(z)}}{1+|r(z)|^2}\sigma_+,\, z \in \Omega_{\kappa^{\mathrm{R}}_1}^2 \cup\Omega_{\kappa^{\mathrm{R}}_2}^1,\\
	V_1= {\overline{r(\bar z)}} e^{-2 i t \theta(z)}\sigma_+,  \quad\,\,\,\,\,\,\, z \in \bar\Omega_{\kappa^{\mathrm{R}}_1}^1\cup\bar\Omega_{\kappa^{\mathrm{R}}_2}^1,\quad
 V_3=\frac{ r(z)e^{2 i t \theta(z)}}{1+|r(z)|^2}\sigma_-
, \quad\,\,\, z \in \bar\Omega_{\kappa^{\mathrm{R}}_1}^1\cup\bar\Omega_{\kappa^{\mathrm{R}}_2}^1, \\
		I,  \text { elsewhere, }\end{cases}
\end{equation}
where$$
\begin{aligned}
	\Omega^1&=\left\{z\in\mathbb{C}\setminus\mathcal{U}:0\leq\arg(z-\kappa_1^{\mathrm{R}})\leq\varphi_0\right\},\\
	\Omega^2&=\left\{z\in\mathbb{C}\setminus\mathcal{U}:\pi-\varphi_0\leq\arg(z-\kappa_1^{\mathrm{R}})\leq\pi\right\},
\end{aligned}
$$
where $\mathcal{U}=\{z\in\mathbb{C}:|z-\kappa^{\mathrm{R}}_2|<\epsilon\}$ with $\epsilon$ sufficiently small to ensure $\mathcal{U}$ does not intersect $\Gamma$, and $\varphi_0\in(0,\pi/2)$ is chosen such that $\Gamma_k\cap\mathbb{C}^+\subset\Omega^1$ for $k>\frac{n}{2}$ and $\Gamma_k\cap\mathbb{C}^+\subset\Omega^2$ for $k\leq \frac{n}{2}$.
Define
$$\Sigma^{(1)}=\Gamma\cup\Sigma_{\kappa_1^{\mathrm{R}}}\cup\Sigma_{\kappa_2^{\mathrm{R}}}\cup\gamma_{(\bar E_{01},E_{01})}\cup\gamma_{(E_{01},E_{00})}\cup\gamma_{(\bar E_{00},\bar E_{01})},$$
where $\gamma_{(\bar E_{01},E_{01})}$ denotes the curve segment passing through $\kappa^{\mathrm{R}}_{1}$, while $\gamma_{(E_{01},E_{00})}$ and $\gamma_{(\bar E_{00},\bar E_{01})}$ denote the oriented contours from $E_{01}$ to $E_{00}$ and from $\bar E_{00}$ to $\bar E_{01}$, respectively, as shown in Figure~\ref{hh}. This leads to the following RH problem.
\begin{figure}[htp]
	{
		\begin{minipage}{16.5cm}\centering
			\begin{overpic}[width=13cm, percent]{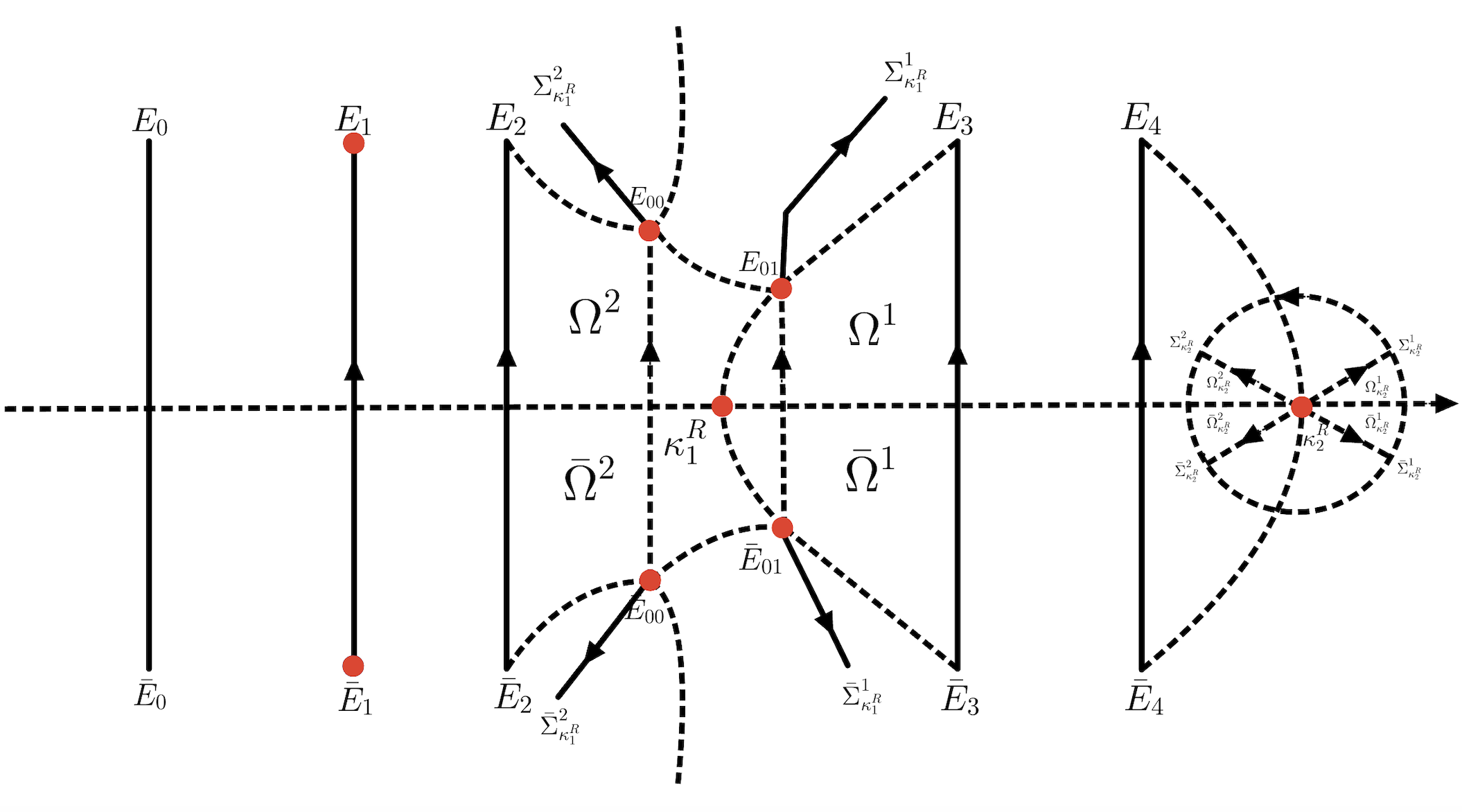}
				\put(30, 50){
					\begin{tikzpicture}[overlay]
					\end{tikzpicture}
				}
			\end{overpic}
	\end{minipage}}
	\caption{Distribution scenarios of genus~$6$ and stationary phase points}
	\label{hh}
\end{figure}
\begin{problem}
	Construct a meromorphic function $N^{(1)}(z) :\mathbb{C}\setminus \Sigma^{(1)}\to SL_2(\mathbb{C})$  such that:
	\begin{enumerate}
		\item $N^{(1)}(z)=I+\mathcal{O}(z^{-1}),$\quad $|z|\to\infty$.
		\item 	For each  $z \in \Sigma^{(1)} \cap U$, the boundary values $N^{(1)}(z)$ satisfy the jump relation
		$$N_+^{(1)}(z)=N_-^{(1)}(z)J^{(1)}(z)$$
		where
		\begin{align*}
		&J^{(1)}(z)=I+\begin{cases}
			V_1^{-1}(z)V_r(z)\tilde V_3(z)	,\quad z\in \gamma_{(\kappa_{1}^\mathbb{R},E_{01})},\quad\,\,\,\,\,\,\,\,\,\,\,
				\sigma_2V_1^{-1}(z)V_r(z)\tilde V_3(z)	\sigma_2,\quad z\in \gamma_{(\bar E_{01},\kappa_{1}^\mathbb{R})}, \\
			\tilde V_1(z)\tilde V_2(z)\tilde V_1(z),\quad\,\,\,\,\,\,\,  z\in\gamma_{(E_{01},E_{00})},\,\,\qquad
			\sigma_2	\tilde V_1(z)\tilde V_2(z)\tilde V_1(z)\sigma_2,\quad\,\,\,\,\,\,\,  z\in\gamma_{(\bar E_{00},\bar E_{01})},\\
				V_1\delta^{-2}(z)\sigma_+, \, z \in \bar\Sigma^1_{\kappa_1^{\mathrm{R}}}\cup\bar\Sigma^1_{\kappa_2^{\mathrm{R}}}, \quad\,\,\,\,\,\,\,\,	V_2\delta^2(z)\sigma_-, \,z \in \Sigma^1_{\kappa_1^{\mathrm{R}}}\cup\Sigma^1_{\kappa_2^{\mathrm{R}}},\\
				V_4\delta^{-2}(z),\sigma_+ \, z \in\Sigma^2_{\kappa_1^{\mathrm{R}}}\cup\Sigma^2_{\kappa_2^{\mathrm{R}}},  \qquad	V_3\delta^{2}(z)\sigma_-, \,\,z \in \bar\Sigma^2_{\kappa_1^{\mathrm{R}}}\cup\bar\Sigma^2_{\kappa_2^{\mathrm{R}}},
		\end{cases}\\
	&	J^{(1)}(z)=	\begin{cases}
	\tilde	V^\ddagger_k(z)=	\begin{pmatrix}
			0 &\frac{ i\overline{a_{+}(\bar z) a_{-}(\bar z)} }{ \delta_+(z) \delta_-(z)}e^{-2\pi ic_k} \\
			\frac{i \delta_+(z)\delta_-(z)  }{\overline{a_{+}(\bar z) a_{-}(\bar z)}} e^{2\pi ic_k}& 0
		\end{pmatrix},\quad z\in \tilde\Gamma_k\cup\mathbb{C}^+,\\
	\tilde V^\dagger_k(z)=\begin{pmatrix}
			0&i\delta^{-2}(z)e^{-2\pi ic_k}\\
			i\delta^{2}(z)e^{2\pi ic_k}&0
		\end{pmatrix},\qquad\qquad z\in\Gamma\cap\mathbb{C}^+\setminus\tilde\Gamma_k,\\
		\sigma_2\tilde V^\dagger_k(z)\sigma_2,\quad z\in\Gamma\cap\mathbb{C}^-\setminus\tilde\Gamma_k,\qquad \sigma_2\tilde	V^\ddagger_k(z)\sigma_2,\quad z\in \tilde\Gamma_k\cup\mathbb{C}^-,
	\end{cases}
	\end{align*}
	and
	\begin{align*}
		&	\tilde V_1(z)=	-\frac{1+|r(z)|^2}{\overline{r(\bar z)}\delta^2(z)}e^{2it\theta(z)}\sigma_-,\quad \tilde V_3=(1+|r(z)|^2)r(z)\delta^{-2}(z)e^{2it\theta(z)}\sigma_-,\\
		&\tilde V_2(z)=-\frac{\overline{r(\bar z)}}{1+|r(z)|^2}\delta^2(z)e^{-2it\theta(z)}\sigma_++	\frac{1+|r(z)|^2}{\overline{r(\bar z)}\delta^2(z)}e^{2it\theta(z)}\sigma_-,\quad V_r=(1+|r(z)|^2)^{\sigma_3}.
	\end{align*}
	\item The function $N^{(1)}(z)$ has singularities at the endpoints of $\,\,\Gamma_k$ of order at most $|z-E_k|^{-1/4}$ or $|z-\bar{E}_k|^{-1/4}$.
	\end{enumerate}
\end{problem}
However, for the nonlinear steepest descent method, the cuts connecting the branch
points to their complex conjugates must lie in the level set where
$\operatorname{Im}\theta(z)=0$. We therefore deform the vertical cut from
$\bar E_{01}$ to $E_{01}$ into the curve segment $\gamma_{(\bar E_{01},E_{01})}$,
and similarly deform the vertical cut connecting $\bar E_{00}$ and $E_{00}$ into
the composite curve
$
	\gamma_{(E_{01},E_{00})}\cup\gamma_{(\bar E_{00},\bar E_{01})}\cup \gamma_{(\bar E_{01},E_{01})}.
$
Since both deformed cuts then proceed along $\gamma_{(\bar E_{01},E_{01})}$, this
segment ceases to be a cut, and the homology basis on the left of
Figure~\ref{figur-cb} transforms into the canonical homology basis $\{a_k,b_k\}$
displayed on the right.

We now construct a new $n+2$ sheeted Riemann surface by adding the cuts
$\gamma_{(\bar E_{00},E_{00})}$ and $\gamma_{(\bar E_{01},E_{01})}$ between
$\Gamma_{\frac{n}{2}-1}$ and $\Gamma_{\frac{n}{2}}$. The resulting ordered set
of cuts is
\begin{equation}
\Gamma^\dagger=	\{\Gamma_0,\dots,\Gamma_{\frac{n}{2}-1}, \Gamma_{\frac{n}{2}}=\gamma_{(\bar E_{01},E_{01})}, \Gamma_{\frac{n}{2}+1}=\gamma_{(\bar E_{00},E_{00})}, \Gamma_{\frac{n}{2}+2},\dots,\Gamma_{n+2}\}.
\end{equation}
\begin{figure}[htp]
	\begin{minipage}{4cm}
		\includegraphics[scale=0.28]{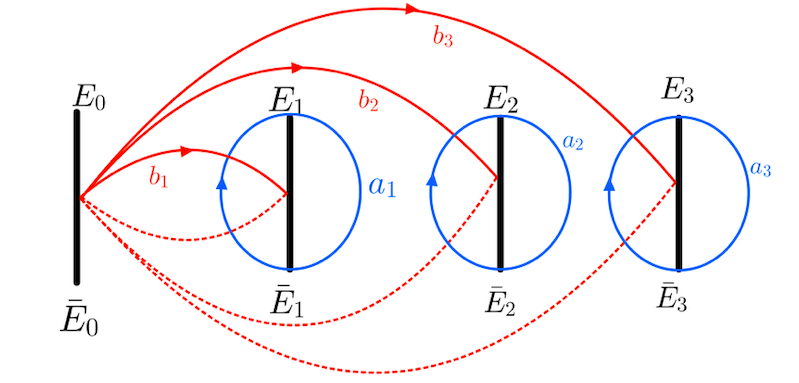}
		\end{minipage}
	\hspace{4.5cm} 
	\begin{minipage}{10cm}    
		\includegraphics[scale=0.26]{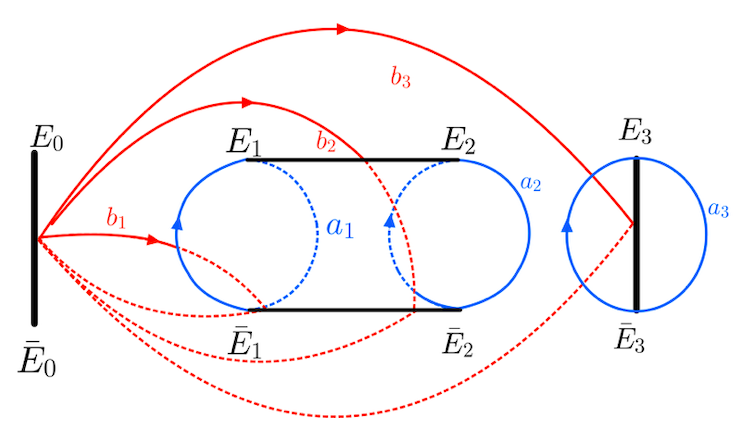}
	\end{minipage}
	\caption{The standard (left) and appropriate (right) choice of canonical homology base $\{a_{k},b_{k}\}_{k=1}^3$ for genus $n=3$}
	\label{figur-cb}
\end{figure}
Define
\begin{equation}
	\Sigma^{(2)}=\Gamma\cup\bigcup_{k=1}^8 L_k \cup \bigcup_{k=1}^8 \bar{L}_k,
\end{equation}
as shown in Figure~\ref{hhh}. The conditions in \eqref{eq:large-theta} imply that
$\theta(z)$ is independent of the choice of contour in $\mathbb{C}\setminus\Sigma^{(2)}$.
Using the relation
\begin{equation}
	\theta'_+(z)=\theta'_-(z)\times\begin{cases}
		1, & z\in\Sigma^{(2)}\setminus(\gamma_{(\bar E_{01},E_{01})},\\
		-1, & z\in\gamma_{(\bar E_{01},E_{01})},
	\end{cases}
\end{equation}
we find that $\theta_+(z)+\theta_-(z)$ is constant on the respective arcs, satisfying
\begin{align}\label{eq:th+}
	\theta_{+}(z)+\theta_{-}(z)=
	\begin{cases}
		2\beta^-, & z\in\gamma_{(\bar E_{00},\bar E_{01})}\cup\gamma_{(E_{01},E_{00})},\\
		c_k, & z\in\Gamma_k,
	\end{cases}
\end{align}
while $\theta_+(z)-\theta_-(z)$ is constant on $\gamma_{(\bar E_{01},E_{01})}$ with
\begin{align}\label{eq:th-}
	\theta_{+}(z)-\theta_{-}(z)=2\beta^+,\quad z\in\gamma_{(\bar E_{01},E_{01})}.
\end{align}
Here the constants $\beta^\pm$ are defined by
\begin{align}\label{con-a-b}
	\beta^-=\theta(E_{00})=\theta(\bar E_{00}),\quad
	\beta^+=\frac{\theta_{+}(E_{01})-\theta_{-}(E_{01})}{2}=\frac{\theta_{+}(\bar E_{01})-\theta_{-}(\bar E_{01})}{2}.
\end{align}
The conditions \eqref{eq:th+}, \eqref{eq:th-} and \eqref{eq:f-g-contour}, combined
with the Schwarz symmetry $\theta(z)=\overline{\theta(\bar{z})}$, ensure that
$\operatorname{Im}\theta(z)=0$ on $\mathbb{R}$ and that
$\operatorname{Im}\theta(E_{k_0-2})=\operatorname{Im}\theta(E_{00})=\operatorname{Im}\theta_+(E_{01})=\operatorname{Im}\theta_-(E_{01})=0$.
Consequently, the constants $\beta^\pm$ defined in \eqref{con-a-b} are indeed real.

The purpose of the following transformation is to make the jumps across the curves $\gamma_{({\bar E_{00}},\bar E_{01})}\cup\gamma_{(E_{01},E_{00})}$ and $\gamma_{(\bar E_{01},E_{01})}$ constant in $z$. We define $N^{(2)}(z)$ by
\begin{equation}
	N^{(2)}(z)=\mathrm{e}^{-ig(\infty)\sigma_3}N^{(1)}(z)G^{(1)}(z)\mathrm{e}^{ig(z)\sigma_3},
\end{equation}
where
\begin{align}
	G^{(1)}(z)=\begin{cases}
		\tilde{V}_1,  z\in U^1,\quad
		\tilde{V}_1^{-1}, z\in U^2,\\
		\tilde{V}_3,  z\in U^4,\quad
		\tilde{V}_4, \,\,\,\, z\in U^3,\\
		\sigma_2\overline{G^{(1)}(\bar{z})}\sigma_2,  z\in\bar{U}^1\cup\bar{U}^2\cup\bar{U}^3\cup\bar{U}^4.
	\end{cases}
\end{align}
\begin{figure}[htp]
	{
		\begin{minipage}{16.5cm}\centering
			\begin{overpic}[width=13cm, percent]{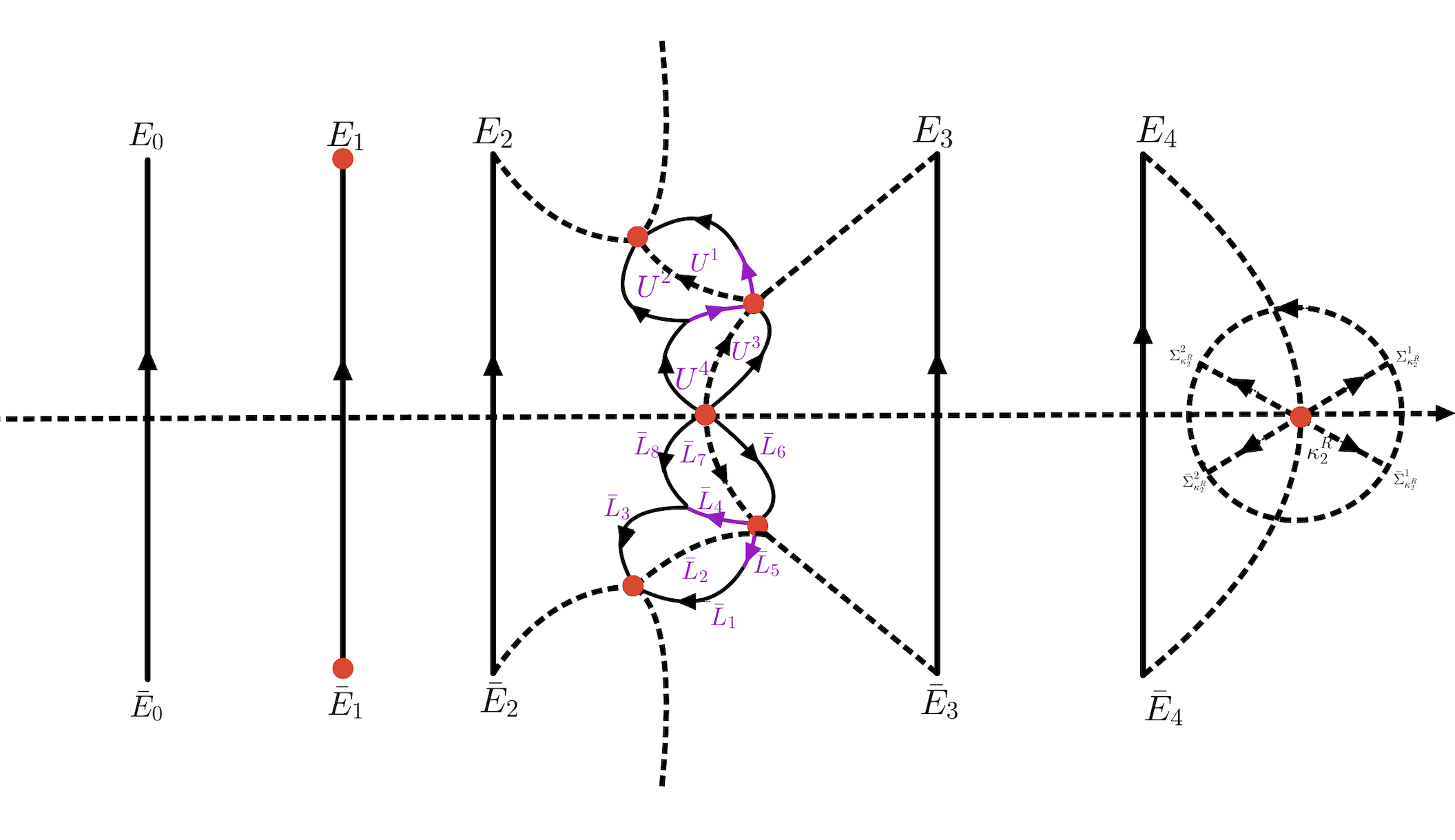}
				\put(30, 50){
					\begin{tikzpicture}[overlay]
					\end{tikzpicture}
				}
			\end{overpic}
	\end{minipage}}
	\caption{The distribution of stationary phase points for $n=6$, where the red region indicates regions with ${\rm Im}\,\theta(z)<0$.}
	\label{hhh}
\end{figure}
and $\{U^j\}_{j=1}^4$ denote open sets of the form displayed in Figure \ref{hhh}.

The function $g(z)$ is defined by
\begin{equation}\label{eq:hhh(z)}
	g(z)=\frac{w(z)}{2\pi i}\int_{\mathcal{Z}}\frac{\mathcal{H}(s)}{s-z}\,ds,\quad z\in\mathbb{C}\setminus\mathcal{Z},
\end{equation}
where $\mathcal{Z}=\gamma_{(\bar E_{00},\bar E_{01})}\cup\gamma_{(E_{01},E_{00})}\cup\gamma_{(\bar E_{01},E_{01})}\cup\Gamma$ and
\begin{equation}
	\mathcal{H}(z)=\begin{cases}
		\frac{2\alpha_1+g_1(z)}{w_{+}(z)}, &\gamma_{(E_{01},E_{00})},\\
			\frac{2\alpha_2+g_2(z)}{w_{+}(z)}, & z\in\tilde\Gamma\cap\mathbb{C}^+,\\
				\frac{2\alpha_3+g_3(z)}{w_{+}(z)}, & z\in\left(\Gamma\setminus\tilde\Gamma\setminus\Gamma_0\right)\cap\mathbb{C}^+,\\
		\frac{2\alpha_4+g_4(z)}{w(z)}, & z\in\gamma_{(\kappa_1^{\mathrm{R}},E_{01})},\\
			\frac{g_3(z)}{w_{+}(z)}, & z\in\Gamma_0\cap\mathbb{C}^+,\\
	\end{cases}
\end{equation}
	with $\alpha_t\in\mathbb{R}$ $(t=1,2,3,4)$ satisfying $\mathcal{H}(z)=\overline{\mathcal{H}(\bar z)}$ and
\begin{align}
&	g_1(z)=-i\ln\left(i\,\frac{\overline{r(\bar{z})}\,\delta^2(z)}{1+|r(z)|^2}\right),\qquad g_2(z)=	i\ln\left(\frac{ \delta_+(z)\delta_-(z)}{\overline{a_{+}(\bar z) a_{-}(\bar z)}}\right),\\
&	g_3(z)=	i\ln\delta^{2}(z),\qquad\qquad\quad
	g_4(z)=i\ln\left(1+|r(z)|^2\right).
\end{align}
It follows that $g(z)$ satisfies the following lemma.
\begin{lemma}\label{lemma-g}
	There exists a unique choice of real constants $\alpha_t=\alpha_t(\xi)$, $t=1,2,3,4$, such that the function $g(z)$ defined in \eqref{eq:hhh(z)} possesses the following properties:
	\begin{enumerate}
		\item It satisfies the symmetry condition $g(z)=\overline{g(\bar{z})}$.
		\item As $z\to\infty$,
		\begin{equation}
			g(z)=g(\infty)\left(1+\frac{I_2}{z}+\mathcal{O}(z^{-2})\right),
		\end{equation}
		where $g(\infty)$ is a finite real constant given by
		\begin{align}\label{ginfty}
			g(\infty)= \exp\left\{-\frac{1}{2\pi i}\int_{\mathcal{Z}} \frac{\mathcal{H}(s)}{w(s)}\,ds\right\}, \quad
			I_2=\frac{1}{2\pi i} \int_{\mathcal{Z}} \frac{s\,\mathcal{H}(s)}{w(s)}\,ds.
		\end{align}
		\item $e^{ig(z)\sigma_3}$ is bounded and analytic for $z\in\mathbb{C}\setminus\mathcal{Z}$.
		\item For each $z \in \mathcal{Z}$, the boundary values $g_\pm(z)$ satisfy the jump relations
		\begin{align}\label{eq:h-condition}
			&g_{+}(z)+g_{-}(z)=
			\begin{cases}
				2 \alpha_1- i\ln \left(i \frac{\overline{r(\bar z)}\delta^2(z)}{1+|r(z)|^2}\right), & z \in \gamma_{(E_{01},E_{00})}, \\
				2\alpha_1+i \ln \left(-i \frac{{r( z)}\overline{\delta^{2}(\bar z)}}{1+|\overline{r(\bar z)}|^2}\right), & z \in \gamma_{(\bar E_{00}, \bar E_{01})},\\
				2\alpha_2+	i\ln\left(\frac{ \delta_+(z)\delta_-(z)}{\overline{a_{+}(\bar z) a_{-}(\bar z)}}\right),&z\in\tilde\Gamma\cap\mathbb{C}^+,\\
				2\alpha_2-	i\ln\left(-\frac{ \overline{\delta_+(\bar z)\delta_-(\bar z)}}{{a_{+}(z) a_{-}( z)}}\right),&z\in\tilde\Gamma\cap\mathbb{C}^-,\\
					2\alpha_3+	i\ln \delta^{2}(z)
					,&z\in\left(\Gamma\setminus\tilde\Gamma\right)\cap\mathbb{C}^+,\\
				2\alpha_3-i\ln (-\overline{\delta^{2}(\bar z)}),&z\in\left(\Gamma\setminus\tilde\Gamma\right)\cap\mathbb{C}^-,
			\end{cases}\\
			&g_{+}(z)-g_{-}(z)=
			\begin{cases}
				2\alpha_4+i \ln \left(1+|r(z)|^2\right), & z \in \gamma_{(\kappa^R_{1}, E_{01})}, \\
				2\alpha_4-i \ln \left(1+|\overline{r(\bar z)}|^2\right), & z \in \gamma_{(\bar E_{01}, \kappa^R_{1})}.
			\end{cases}
		\end{align}
	\end{enumerate}
\end{lemma}
It follows that $N^{(2)}(z)$ satisfies the following RH problem.
\begin{problem}
	Construct a meromorphic function $N^{(2)}(z) :\mathbb{C}\setminus \Sigma^{(2)}\to SL_2(\mathbb{C})$  such that:
	\begin{enumerate}
		\item $N^{(2)}(z)=I+\mathcal{O}(z^{-1}),$\quad $z\to\infty$.
		\item 	For each $z\in\Sigma^{(2)}$, the boundary values $N_\pm^{(2)}(z)$ satisfy the jump relation
		\begin{equation}\label{eq:J(z)-aasy2}
			N^{(2)}_+(z)=N_-^{(2)}(z)J^{(2)}(z),
		\end{equation}
		where
\begin{align*}
	&J^{(2)}(z) =\begin{cases}
		I+	e^{-ig(z)\hat\sigma_3}V_4^{-1}, z \in L_1\cup L_3,\quad ie^{-2i\left(t\beta^-+\alpha_1\right)}\sigma_++	ie^{2i\left(t\beta^-+\alpha_1\right)}\sigma_-,\quad\,\, z\in L_2,\\
I	-e^{-ig(z)\hat\sigma_3}\tilde V_1,z \in L_4,\qquad\quad\,\,\,\,\,\,\, 	I+e^{-ig(z)\hat\sigma_3}V_2\tilde V_1,\, z \in L_5,\quad
	I-e^{-ig(z)\hat\sigma_3}V_1,\, z \in L_6,\\
	I+	V_2\delta^2(z)\sigma_-, \,z \in \Sigma^1_{\kappa_1^{\mathrm{R}}}\cup\Sigma^1_{\kappa_2^{\mathrm{R}}},\,\,\,\,\,
I+	V_4\delta^{-2}(z)\sigma_+ \, z \in\Sigma^2_{\kappa_1^{\mathrm{R}}}\cup\Sigma^2_{\kappa_2^{\mathrm{R}}}, \\
ie^{-2\pi ic_k-2\alpha_2}\sigma_++
			ie^{2\pi ic_k+2\alpha_2}\sigma_-,z\in\tilde\Gamma,\qquad\, 	e^{2i\alpha_2\hat\sigma_3},\qquad\,\,\,\,\,\,\,\,\,\, z \in L_7, \\
	ie^{-2\pi ic_k-2\alpha_3}\sigma_++
			ie^{2\pi ic_k+2\alpha_3}\sigma_-,z\in\Gamma\setminus\tilde\Gamma,\quad 	
			-e^{-ig(z)\hat\sigma_3}V_3,\,\,\, z \in L_{8},\\
			\sigma_2	\overline{ J^{(2)}(\bar z)}\sigma_2,\quad z\in\bar \Sigma^{(2)},
	\end{cases}
\end{align*}
		\item The function $N^{(2)}(z)$ has singularities at the endpoints of $\,\,\Gamma_k$ of order at most $|z-E_k|^{-1/4}$ or $|z-\bar{E}_k|^{-1/4}$.
	\end{enumerate}
\end{problem}
Let $\mathcal{D}\subset\mathbb{C}$ denote the union of six open disks,
$$
\mathcal{D}=D_\epsilon(E_{00})\cup D_\epsilon(E_{01})\cup D_\epsilon(\kappa^{\mathrm{R}}_{1})\cup D_\epsilon(\bar E_{00})\cup D_\epsilon(\bar E_{01})\cup\mathcal{U},
$$
where $D_\epsilon(\cdot)$ denotes the disk with radius $\epsilon\equiv\epsilon(\xi)$ chosen sufficiently small so that these disks are pairwise disjoint. On $\Sigma^{(2)}\setminus\mathcal{D}$, the jump matrix $J^{(2)}(z)$ approaches the identity matrix as $t\to\infty$. We will construct the solution $N^{(2)}(z)$ by seeking a solution of the form
\begin{equation}\label{eq:N-2-dis}
	N^{(2)}(z)=\begin{cases}
		\mathcal{E}(z)N^{alg}(z), & z\in\mathbb{C}\setminus\mathcal{D},\\[6pt]
		\mathcal{E}(z)N^{alg}(z)N^{\tilde{\kappa}}(z), & z\in\mathcal{D},
	\end{cases}
\end{equation}
where $\tilde{\kappa}=\left\{E_{00},E_{01},\bar E_{00},\bar E_{01},\kappa^{\mathrm{R}}_{1},\kappa^{\mathrm{R}}_{2}\right\}$, with $N^{alg}(z)$ and $N^{\tilde{\kappa}}(z)$ being the models constructed in RH problem \ref{RH-mod} below, and where the error function $\mathcal{E}(z)$, which solves the small-norm RH problem \ref{RH-EE}, is proven to exist and to be asymptotically bounded.
\subsection{Solution of the Model Problem}
The matrix $N^{(2)}(z)$ is meromorphic away from the contour $\Sigma^{(2)}$, on which its boundary values satisfy the jump relation \eqref{eq:J(z)-aasy2}. Away from the stationary phase points $\kappa^{\mathrm{R}}_1,\kappa^{\mathrm{R}}_2$, however, the jump is uniformly close to the identity. Recalling the definition \eqref{eq:theta(z)} of $\theta(z)$, we have
\begin{equation}\label{eq:mo-pro}
	\left\|V^{(2)}-I\right\|_{L^{\infty}\left(\Sigma^{(2)}\setminus\mathcal{D}\right)}=\mathcal{O}\left(\mathrm{e}^{-ct}\right),
\end{equation}
where $c=c(\xi)>0$, implying exponential decay on $\Sigma^{(2)}\setminus\mathcal{D}$. This justifies constructing a model solution outside $\mathcal{D}$ that ignores the jumps completely.

Define
\begin{equation}
	\Sigma^{alg}=\Gamma\cup L_2\cup L_7\cup\bar{L}_2\cup\bar{L}_7.
\end{equation}
For the function $N^{alg}(z)$ defined by \eqref{eq:N-2-dis}, we have the following RH problem:
\begin{problem}\label{RH-mod}
	Construct a meromorphic function $N^{alg}(z) :\mathbb{C}\setminus \Sigma^{alg}\to SL_2(\mathbb{C})$  such that:
	\begin{enumerate}
		\item $N^{alg}(z)=I+\mathcal{O}(z^{-1}),$\quad $|z|\to\infty$.
		\item 	For each $z\in\Sigma^{alg}$, the boundary values $N^{alg}_\pm(z)$ satisfy the jump relation
		$$N^{alg}_+(z)=N^{alg}_-(z)J^{alg}(z).$$
		where
\begin{equation}\label{eq:J-alg}
J^{alg}(z)=J^{(2)}(z)|_{z\in\Sigma^{alg}.}
\end{equation}
		\item
The function $N^{alg}(z)$ has singularities at the endpoints of $\,\,\Gamma_k$ of order at most $|z-E_k|^{-1/4}$ or $|z-\bar{E}_k|^{-1/4}$.
	\end{enumerate}
\end{problem}
Let $\{a_k,b_k\}$ denote the appropriate homology basis. Then the Abelian integrals $\varphi(z)$, when considered on the upper sheet of $\mathcal{X}$, satisfy the following jump conditions:
\begin{equation}
	\varphi(z_{+})+\varphi(z_{-})=\left(0,\tau e_1,\dots,\tau e_{\frac{n}{2}}, \tau_{\frac{n}{2}}+\sum_{k=1}^{\frac{n}{2}-1}e_{k}, \tau e_{\frac{n}{2}+1}, \dots, \tau e_{n+2}\right),
\end{equation}
where the components correspond to $z\in\Gamma^\dagger$, and
\begin{equation}
	\varphi(z_{+})-\varphi(z_{-})=\sum^{\frac{n}{2}}_{k=1} e_k+\tau_{\frac{n}{2}}-\tau_{\frac{n}{2}+1}, \quad z \in \gamma_{(\bar E_{01},E_{01})}.
\end{equation}
For each choice of the real constants $\alpha_j$ $(j=1,2,3,4)$, $\beta^\pm$, and for each $t\geq0$, the RH problem~\ref{RH-mod} admits a unique solution $N^{alg}(z)$ given explicitly by
\begin{equation}
	N^{alg}(z)=\bigl(\tilde{N}^{alg}(\infty)\bigr)^{-1}\tilde{N}^{alg}(z),
\end{equation}
where $\tilde{N}^{alg}(z)=\bigl(\tilde{N}_{ij}^{alg}(z)\bigr)_{i,j=1,2}$ with components
	\begin{align}
		\tilde{N}_{s1}^{alg}(z)&=\frac{1}{2}\bigl(\hat{\nu}_0(z)+\hat{\nu}_0(z)^{-1}\bigr)\frac{\Theta\bigl(\varphi(z)+\hat{c}+d_s\bigr)}{\Theta\bigl(\varphi(z)+d_s\bigr)},\\
		\tilde{N}_{s2}^{alg}(z)&=\frac{1}{2}\bigl(\hat{\nu}_0(z)-\hat{\nu}_0(z)^{-1}\bigr)\frac{\Theta\bigl(\varphi(z)-\hat{c}-d_s\bigr)}{\Theta\bigl(\varphi(z)-d_s\bigr)},	\qquad s=1,2,
	\end{align}
and $\Theta(z)$ denotes the Riemann theta function defined in \eqref{eq:THeze}.
Here, the function $\hat{\nu}_0(z)$ is defined as
\begin{equation}
	\hat{\nu}_0(z)=\nu(z)\sqrt[4]{\left(\frac{z-E_{00}}{z-\bar{E}_{00}}\right)\left(\frac{z-E_{01}}{z-\bar{E}_{01}}\right)},
\end{equation}
and the vector-valued function $\hat{c}=\hat{c}(\xi)$ is defined by
\begin{equation}
	\hat{c}=\left(-\frac{\alpha_3}{\pi},\dots,-c_{\frac{n}{2}-1}-\frac{\alpha_3}{\pi},-\frac{t\beta^-+\alpha_1}{\pi},\left(-\frac{\alpha_2}{\pi}-\frac{t\beta^-+\alpha_1}{\pi}\right),-c_{\frac{n}{2}}-\frac{\alpha_2}{\pi},\dots,-c_{n}-\frac{\alpha_2}{\pi}\right).
\end{equation}
Let $\varphi\colon\mathcal{X}\to\mathbb{C}^{n+2}$ denote the Abel map defined in \eqref{eq:abel}.
Define vectors $d_1, d_2 \in \mathbb{C}^{n+2}$ by
\begin{equation}
	d_1 = \varphi(\mathcal{D}) + K, \qquad d_2 = -\varphi(\mathcal{D}) - K,
\end{equation}
where $\mathcal{D}$ denotes the pole divisor on $\mathcal{X}$, and $K \in \mathbb{C}^{n+2}$ denotes the vector of Riemann constants given by
\begin{equation}
	K = \frac{1}{2}\left(\sum_{k=1}^{n+2}\tau_{kk} + \sum_{k=1}^{n+2}e_k - e_{\frac{n}{2}}\right).
\end{equation}
Moreover, the solution $N^{alg}(z)$ admits the asymptotic expansion
\begin{equation}\label{eq:Nglo-solution}
	N^{alg}(z) = I + \frac{N_1^{alg}}{z} + \mathcal{O}\bigl(z^{-2}\bigr), \qquad z\to\infty,
\end{equation}
and the quantity $\hat{Q}(x,t)$ is given explicitly by
\begin{equation}\label{eq:Nglo-solution1}
	\hat{Q}(x,t) = \bigl(N_1^{alg}\bigr)_{12} = -\frac{i}{2}\operatorname{Im}\left(\sum_k E_k + E_{00} + E_{01}\right) \cdot \frac{\Theta\bigl(\varphi(\infty^+) + d_1\bigr)\Theta\bigl(\varphi(\infty^+) - \hat{c} - d_1\bigr)}{\Theta\bigl(\varphi(\infty^+) - d_1\bigr)\Theta\bigl(\varphi(\infty^+) + \hat{c} + d_1\bigr)}.
\end{equation}
\subsection{Solution of the Local Problem}
Recall that the matrix $J^{(2)}(z)-I$ decays to zero for $z\in\Sigma^{(2)}\setminus\Sigma^{alg}$ as $t\to\infty$. However, this decay is not uniform in $z$ as $z$ approaches $\Sigma^{alg}$. Thus, in neighborhoods of $\Sigma^{alg}$ within the contour $\Sigma^{(2)}$, we introduce local solutions that approximate $N^{(2)}(z)$ more accurately than $N^{alg}(z)$. These local approximations enable us to derive sharp error estimates and higher-order asymptotic corrections beyond the leading-order $\mathcal{O}(1)$ term.

In the following three subsections, we define three local solutions, denoted by $N^{E_{00}}(z)$, $N^{E_{01}}(z)$, and $N^{\kappa^{\mathrm{R}}}(z)$, which provide good approximations of $N^{(2)}(z)$ for $z$ in the three disks $D_\epsilon(E_{00})$, $D_\epsilon(E_{01})$, and $D_\epsilon(\kappa^{\mathrm{R}})$, respectively.
\subsection*{Local Parameter near $E_{00}$}
We define the function $N_0^{E_{00}}(z)$ for $z\to E_{00}$ by
\begin{equation}
	N^{E_{00}}(z)=N^{(2)}(z) e^{-i\left(\frac{i}{2} \ln \left(- \frac{\overline{r(\bar z)}\delta^2(z)}{1+|r(z)|^2}\right)+t\theta(E_{00})+g(z)\right) \sigma_3}, \quad z \in D_\epsilon(E_{00}) \cap \Sigma^{(2)} .
\end{equation}
\begin{figure}[htp]
	{
		\begin{minipage}{16cm}\centering
			\includegraphics[scale=0.55]{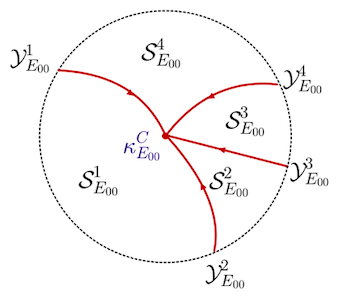}
	\end{minipage}}
	\caption{The contour  $\mathcal{Y}^j_{\kappa^C_1}$  in the  $D_\epsilon(E_{00})$ of $\kappa^C_{j}$ and the region $\{S^1_{\kappa_{s_0-1}}\}$ for $j=1,2,3,4$}
	\label{F4.11}
\end{figure}
It follows from the expression \eqref{eq:J(z)-aasy2} for the jump matrix $J^{(2)}(z)$ that $N^{E_{00}}(z)$ satisfies the following jump condition on $D_\epsilon(E_{00})$ (see Figure \ref{F4.11}):
\begin{equation}
	N_{+}^{E_{00}}(z)=N_{-}^{E_{00}}(z)J^{E_{00}}(z),\quad z\in\Sigma^{(2)}\cap D_\epsilon(E_{00}).
\end{equation}
where
\begin{equation}
	J^{E_{00}}(z)=
	\begin{cases}
		-\mathrm{e}^{-2it\theta_{E_{00}}(z)}\sigma_+, & z\in\mathcal{Y}^1_{E_{00}},\\[6pt]
		\mathrm{e}^{2it\theta_{E_{00}}(z)}\sigma_-, & z\in\mathcal{Y}^2_{E_{00}}\cup\mathcal{Y}^4_{E_{00}},\\[6pt]
	i\sigma_2, & z\in\mathcal{Y}^3_{E_{00}},
	\end{cases}
\end{equation}
where $\theta_{E_{00}}(z)=\theta(z)-\theta(E_{00})=\int_{E_{00}}^{z}\sqrt{s-E_{00}}f(s)\,ds$, and both $f(s)$ and $\int_{E_{00}}^{z}f(s)\,ds$ are analytic and nonzero at $s=E_{00}$.

In order to relate $N^{E_{00}}(z)$ to the Airy solution $N^{\mathrm{Ai}}(z)$ of Appendix \ref{App-C}, we make a local change of variables for $z\to E_{00}$ and introduce the new variable
\begin{equation}\label{eq:zeta-2}
\zeta=	\zeta(z)=\left(\frac{3it}{2}\left(\theta(z)-\theta(E_{00})\right)\right)^{2/3}=\left(\frac{3t}{2}\int_{E_{00}}^{z}f(s)\,ds\right)^{2/3}(z-E_{00}).
\end{equation}
We see that $N^{E_{00}}(\zeta)$ has the same jumps as $N^{\mathrm{Ai}}(z)$ near $E_{00}$, with the jump contours $\mathcal{Y}^j_{E_{00}}\mapsto\mathcal{Y}_j$ and $\mathcal{S}^s_{E_{00}}\mapsto\mathcal{S}_j$ for $j=1,2,3,4$. Hence we seek a parametrix $N^{E_{00}}(z)$ for $N^{(2)}(z)$ near $E_{00}$ of the form
\begin{equation}
	N^{E_{00}}(z)=N^{\mathrm{Ai}}(\zeta(z))\,\mathrm{e}^{i\left(\frac{i}{2}\ln\left(-\frac{\overline{r(\bar{z})}\,\delta^2(z)}{1+|r(z)|^2}\right)+t\theta(E_{00})+h(z)\right)\sigma_3},\quad z\in D_\epsilon(E_{00})\cap\Sigma^{(2)},
\end{equation}
where $N^{\mathrm{Ai}}(\zeta(z))$ is given in \eqref{App-as-11} with $\zeta$ given by \eqref{eq:zeta-2}.
\subsection*{Local model near $E_{01}$}
For $z$ near $E_{01}$, we first introduce the local phase function $\theta_{E_{01}}(z)$ defined by
$$
\theta_{E_{01}}(z):=\int_{E_{01}}^z d\theta=
\begin{cases}
	\theta(z)-\theta_{-}(E_{01}), & z\in S^1_{E_{01}}\cup S^2_{E_{01}}\cup S^5_{E_{01}},\\[6pt]
	\theta(z)-\theta_{+}(E_{01}), & z\in S^3_{E_{01}}\cup S^4_{E_{01}},
\end{cases}
$$
where the regions $S^j_{E_{01}}$ for $j=1,\dots,5$ are depicted in Figure \ref{F4.1}.

Next, we define the sectionally holomorphic function $\mathcal{A}(z)\equiv\mathcal{A}(z,\xi)$ on these regions:
\begin{equation}
	\mathcal{A}(z)=
	\begin{cases}
		\left(1+|r(z)|^2\right)^{\sigma_3/2}\mathrm{e}^{-it\theta_{-}(E_{01})\sigma_3}, & z\in S^1_{E_{01}}\cup S^2_{E_{01}}\cup S^5_{E_{01}},\\[8pt]
		\left(1+|r(z)|^2\right)^{-\sigma_3/2}\mathrm{e}^{-it\theta_{+}(E_{01})\sigma_3}, & z\in S^3_{E_{01}}\cup S^4_{E_{01}}.
	\end{cases}
\end{equation}
Finally, we define the transformed function $N_0^{E_{01}}(z)$ for $z\to E_{01}$ by
\begin{equation}\label{N-j-kappa-2}
	N^{E_{01}}(z):=N^{(2)}(z)\mathrm{e}^{-i\left(\frac{i}{2}\ln\left(-\frac{\overline{r(\bar{z})}\,\delta^2(z)}{1+|r(z)|^2}\right)+g(z)\right)\sigma_3}\mathcal{A}(z),\quad z\in D_\epsilon(E_{01})\cap\Sigma^{(2)}.
\end{equation}
Under this transformation, the jump condition $J^{(2)}(z)$  in equation  \eqref{eq:J(z)-aasy2} on $D_\epsilon(E_{01})\cap\Sigma^{(2)}$ reduces to
$$
J^{E_{01}}(z)=\begin{cases}
	-\mathrm{e}^{-2it\theta_{E_{01}}(z)}\sigma_+, & z\in\mathcal{Y}^1_{E_{01}},\\[6pt]
	\mathrm{e}^{2it\theta_{E_{01}}(z)}\sigma_-, & z\in\mathcal{Y}^2_{E_{01}}\cup\mathcal{Y}^4_{E_{01}},\\[6pt]
	i\sigma_2, & z\in\mathcal{Y}^3_{E_{01}}.
\end{cases}
$$
where $\theta_{E_{01}}(z):=\theta(z)-\theta(E_{01})$, and
\begin{figure}[htp]
	{
		\begin{minipage}{16.3cm}\centering			\includegraphics[scale=0.5]{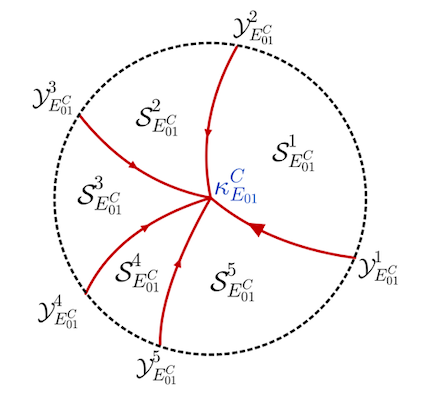}
	\end{minipage}}
	\caption{The contour in the $\epsilon$ neighborhood $D_\epsilon(E_{00})$ of $E_{00}$ and the sets $\{S^1_{\kappa_{2}}\}$}
	\label{F4.1}
\end{figure}
To construct the local solution near $E_{01}$, we first introduce the local phase function and the coordinate transformation.Define the local coordinate $\zeta(z)$ which maps the neighborhood of $z=E_{01}$ to $\zeta=0$ by
\begin{equation}
	\zeta(z)=\left(\frac{3it}{2}\left(\theta(z)-\theta(E_{01})\right)\right)^{2/3}=\left(\frac{3t}{2}\right)^{2/3}(z-E_{01})\psi_{E_{01}}(z),
\end{equation}
 $\psi_{E_{01}}(z)\equiv\psi_{E_{01}}(z,\xi)$ is analytic in $D_\epsilon(E_{01})$ with $\psi_{E_{01}}(E_{01})\neq 0$.

Observe that these jump conditions coincide with the standard jumps of the Airy function $N^{\mathrm{Ai}}(\zeta)$ (see Appendix \ref{App-C}). We therefore seek a parametrix of the form
\begin{equation}
	N^{E_{01}}(z)=N^{\mathrm{Ai}}(\zeta(z))Q^{-1}(z)\mathrm{e}^{i\left(\frac{i}{2}\ln\left(-\frac{\overline{r(\bar{z})}\,\delta^2(z)}{1+|r(z)|^2}\right)+g(z)\right)\sigma_3},\quad z\in D_\epsilon(E_{01})\setminus\Sigma^{(2)},
\end{equation}
where $N^{\mathrm{Ai}}$ is given in \eqref{App-as-11} and $\zeta(z)$ is defined by \eqref{eq:zeta-2}.
\subsection*{Local model near $\kappa^{\mathrm{R}}_{1}$ and $\kappa^{\mathrm{R}}_2$}
In this section, we consider the local parametrix near the real stationary phase point $\kappa^{\mathrm{R}}_1$ for $z\to\kappa^{\mathrm{R}}_1$.
Define the piecewise constant function $\mathcal{F}(z)\equiv\mathcal{F}(z,\xi)$ by
$$
\mathcal{F}(z)=\begin{cases}
	\mathrm{e}^{-it\theta_{-}(\kappa^{\mathrm{R}}_{1})\sigma_3}, & z\in S^1_{\kappa^{\mathrm{R}}_{1}}\cup S^2_{\kappa^{\mathrm{R}}_{1}}\cup S^6_{\kappa^{\mathrm{R}}_{1}},\\
	\mathrm{e}^{-it\theta_{+}(\kappa^{\mathrm{R}}_{1})\sigma_3}, & z\in S^3_{\kappa^{\mathrm{R}}_{1}}\cup S^4_{\kappa^{\mathrm{R}}_{1}}\cup S^5_{\kappa^{\mathrm{R}}_{1}},
\end{cases}
$$
and define the function $N_0^{\kappa^{\mathrm{R}}_{1}}(z)$ by
\begin{equation}
	N_0^{\kappa^{\mathrm{R}}_{1}}(z)=N^{(2)}(z)\mathrm{e}^{-ih(z)\sigma_3}\mathcal{F}(z),\quad z\in D_\epsilon(\kappa^{\mathrm{R}}_{1})\setminus\Sigma^{(2)}.
\end{equation}
Across the part of $\Sigma^{(2)}$ that lies in $D_\epsilon(\kappa^{\mathrm{R}}_{1})$, $N_0^{\kappa^{\mathrm{R}}_{1}}(z)$ satisfies the jump condition
$$
N_{0+}^{\kappa^{\mathrm{R}}_{1}}(z)=N_{0-}^{\kappa^{\mathrm{R}}_{1}}(z)J_0^{\kappa^{\mathrm{R}}_{1}}(z),
$$
where
\begin{align*}
	J_0^{\kappa^{\mathrm{R}}_{1}}(z)=I+\begin{cases}
		-r(z)\delta^{-2}(z)\mathrm{e}^{2it\theta_{\kappa^{\mathrm{R}}_{1}}(z)}\sigma_-, & z\in\mathcal{Y}^1_{\kappa^{\mathrm{R}}_{1}},\\
		-(1+|r(z)|^2)r(z)\delta^{2}(z)\mathrm{e}^{-2it\theta_{\kappa^{\mathrm{R}}_{1}}(z)}\sigma_+, & z\in\mathcal{Y}^5_{\kappa^{\mathrm{R}}_{1}},\\
		-\overline{r(\bar{z})}\delta^{2}(z)\mathrm{e}^{-2it\theta_{\kappa^{\mathrm{R}}_{1}}(z)}\sigma_+, & z\in\mathcal{Y}^2_{\kappa^{\mathrm{R}}_{1}},\\
		-(1+|r(z)|^2)r(z)\delta^{-2}(z)\mathrm{e}^{2it\theta_{\kappa^{\mathrm{R}}_{1}}(z)}\sigma_-, & z\in\mathcal{Y}^4_{\kappa^{\mathrm{R}}_{1}},\\
		V_r^{-1}(z)-I, & z\in\mathcal{Y}^6_{\kappa^{\mathrm{R}}_{1}},\\
		V_r(z)-I, & z\in\mathcal{Y}^3_{\kappa^{\mathrm{R}}_{1}},
	\end{cases}
\end{align*}
and all contours are oriented upward as in Figure \ref{Figure-3}.
\begin{figure}[htp]
	{
		\begin{minipage}{16cm}\centering
			\includegraphics[scale=0.4]{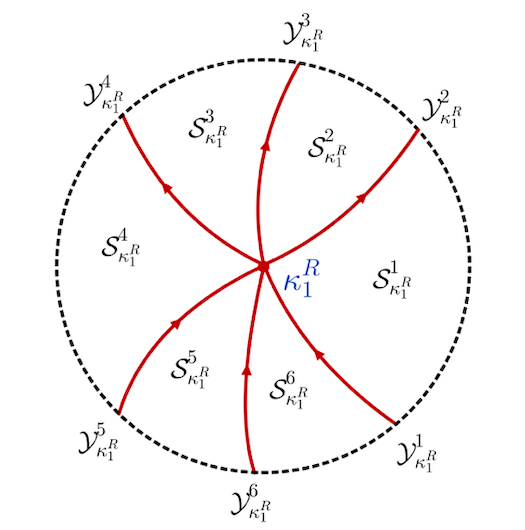}
	\end{minipage}}
	\caption{The contour and region  in the $\epsilon$ neighborhood $D_\epsilon(\kappa^{\mathrm{R}}_{1})$ of $\kappa^{\mathrm{R}}_{1}$}
	\label{Figure-3}
\end{figure}

We want to eliminate the jumps across $\mathcal{Y}^6_{\kappa^{\mathrm{R}}_{1}}\cup\mathcal{Y}^3_{\kappa^{\mathrm{R}}_{1}}$. Hence we define the complex-valued function $\tilde{\delta}(z)\equiv\tilde{\delta}(\xi,z)$ by
\begin{equation}\label{eq:defin-delta}
	\tilde{\delta}(z)=\exp\left\{\frac{1}{2\pi i}\int_{\mathcal{Y}^3_{\kappa^{\mathrm{R}}_{1}}}\frac{\ln\left(1+|r(s)|^2\right)}{s-z}\,ds-\frac{1}{2\pi i}\int_{\mathcal{Y}^6_{\kappa^{\mathrm{R}}_{1}}}\frac{\ln\left(1+|r(s)|^2\right)}{s-z}\,ds\right\},\quad z\in\mathbb{C}\setminus\left(\mathcal{Y}^3_{\kappa^{\mathrm{R}}_{1}}\cup\mathcal{Y}^6_{\kappa^{\mathrm{R}}_{1}}\right),
\end{equation}
where the branch of the logarithm is chosen such that $\ln\left(1+|r(s)|^2\right)$ is continuous on each contour and real at $s=\kappa^{\mathrm{R}}_{1}$.
\begin{lemma}
The function $\tilde\delta(z)$ has the following properties:
\begin{enumerate}
\item $\tilde{\delta}(z)$ and $\tilde{\delta}^{-1}(z)$ are bounded and analytic functions of $z \in D_\epsilon(\kappa^{\mathrm{R}}_{1}) \setminus \left(\mathcal{Y}^3_{\kappa^{\mathrm{R}}_{1}}\,\cup\,y\mathcal{Y}^6_{\kappa^{\mathrm{R}}_{1}}
\right)$.
    \item  $\tilde{\delta}(z)$ obeys the symmetry
$$
\tilde{\delta}=\left(\overline{\tilde{\delta}(\bar z}\right)^{-1}, \quad z \in \mathbb{C} \setminus\left(\mathcal{Y}^3_{\kappa^{\mathrm{R}}_{1}}\,\cup\,\mathcal{Y}^6_{\kappa^{\mathrm{R}}_{1}}
\right)
$$
\item $\tilde{\delta}(z)$ satisfies the jump condition
$$
\tilde{\delta}_{+}(z)=\tilde{\delta}_{-}(z) \times \begin{cases}1+|r(z)|^2, & z \in \mathcal{Y}^3_{\kappa^{\mathrm{R}}_{1}}
, \\ \left(1+|r(z)|^2\right)^{-1}, & z \in \mathcal{Y}^6_{\kappa^{\mathrm{R}}_{1}}.\end{cases}
$$\end{enumerate}
\end{lemma}
Define
\begin{equation}\label{eq:N1-kappaR}
	N_1^{\kappa^{\mathrm{R}}_{1}}(z)=N_0^{\kappa^{\mathrm{R}}_{1}}(z)\tilde{\delta}(z)^{-\sigma_3},\quad z\in D_\epsilon(\kappa^{\mathrm{R}}_{1}),
\end{equation}
which satisfies the jump condition $N_{1+}^{\kappa^{\mathrm{R}}_{1}}(z)=N_{1-}^{\kappa^{\mathrm{R}}_{1}}(z)J_1^{\kappa^{\mathrm{R}}_{1}}(z)$ with
\begin{equation}\label{PC_jj}
	J_1^{\kappa^{\mathrm{R}}_{1}}(z)=I+
	\begin{cases}
		-r(z)\hat{\delta}^{-2}(z)\mathrm{e}^{2it\theta_{\kappa^{\mathrm{R}}_{1}}(z)}\sigma_-, & z\in\mathcal{Y}^1_{\kappa^{\mathrm{R}}_{1}},\\
		-\overline{r(\bar{z})}\hat{\delta}^{2}(z)\mathrm{e}^{-2it\theta_{\kappa^{\mathrm{R}}_{1}}(z)}\sigma_+, & z\in\mathcal{Y}^2_{\kappa^{\mathrm{R}}_{1}},\\
		-(1+|r(z)|^2)r(z)\hat{\delta}^{-2}(z)\mathrm{e}^{2it\theta_{\kappa^{\mathrm{R}}_{1}}(z)}\sigma_-, & z\in\mathcal{Y}^4_{\kappa^{\mathrm{R}}_{1}},\\
		(1+|r(z)|^2)r(z)\hat{\delta}^{2}(z)\mathrm{e}^{-2it\theta_{\kappa^{\mathrm{R}}_{1}}(z)}\sigma_+, & z\in\mathcal{Y}^5_{\kappa^{\mathrm{R}}_{1}},
	\end{cases}
\end{equation}
where $\hat{\delta}(z)=\tilde{\delta}(z)\delta(z)$.

To relate $N_1^{\kappa^{\mathrm{R}}_{1}}(z)$ to the parabolic cylinder solution $M^{(\mathrm{PC})}(\zeta,r)$ of Appendix~A in \cite{MR18}, we first observe that $\theta_{\kappa^{\mathrm{R}}_{1}}(z)$ has a double zero at $z=\kappa^{\mathrm{R}}_{1}$. This motivates the local change of variables near $\kappa^{\mathrm{R}}_1$:
\begin{equation}\label{eq:dsdss}
	\tilde{\zeta}=2\sqrt{t\theta_{\kappa^{\mathrm{R}}_{1}}(z)}=2\sqrt{t}(z-\kappa^{\mathrm{R}}_{1})\psi_{\kappa^{\mathrm{R}}_{1}}(z),
\end{equation}
where $\psi_{\kappa^{\mathrm{R}}_{1}}(z)$ is analytic in $D_\epsilon(\kappa^{\mathrm{R}}_{1})$. The branch of the square root is fixed by requiring $\psi_{\kappa^{\mathrm{R}}_{1}}(\kappa^{\mathrm{R}}_{1})>0$. Shrinking $\epsilon>0$ if necessary, we may assume $\operatorname{Re}\psi_{\kappa^{\mathrm{R}}_{1}}(z)>0$ for $z\in\overline{D_\epsilon(\kappa^{\mathrm{R}}_{1})}$.

Next we consider the behavior of $\tilde{\delta}(z)$ as $z\to\kappa^{\mathrm{R}}_{1}$. Integrating by parts, we find
\begin{align}
	\int_{\mathcal{Y}^3_{\kappa^{\mathrm{R}}_{1}}}\frac{\ln\left(1+|r(s)|^2\right)}{s-z}\,ds&=\ln\left(z-E_{01}\right)\ln\left(1+|r(E_{01})|^2\right)-\ln_+\left(z-\kappa^{\mathrm{R}}_{1}\right)\ln\left(1+|r(\kappa^{\mathrm{R}}_{1})|^2\right)\nonumber\\
	&\quad-\int_{\mathcal{Y}^3_{\kappa^{\mathrm{R}}_{1}}}\ln(s-z)\,d\ln\left(1+|r(s)|^2\right),\\
	\int_{\mathcal{Y}^6_{\kappa^{\mathrm{R}}_{1}}}\frac{\ln\left(1+|r(s)|^2\right)}{s-z}\,ds&=-\ln\left(z-\bar{\kappa}^{\mathrm{C}}_2\right)\ln\left(1+|r(\bar{\kappa}^{\mathrm{C}}_2)|^2\right)+\ln_-\left(z-\kappa^{\mathrm{R}}_{1}\right)\ln\left(1+|r(\kappa^{\mathrm{R}}_{1})|^2\right)\nonumber\\
	&\quad-\int_{\mathcal{Y}^6_{\kappa^{\mathrm{R}}_{1}}}\ln(s-z)\,d\ln\left(1+|r(s)|^2\right),
\end{align}
where the branches are fixed by requiring that $\ln(z-\kappa^{\mathrm{R}}_{1})>0$ for $z>\kappa^{\mathrm{R}}_{1}$ with $s\in\mathcal{Y}^3_{\kappa^{\mathrm{R}}_{1}}\cup\mathcal{Y}^6_{\kappa^{\mathrm{R}}_{1}}$, and that $\ln(z-s)$ is continuous for $s\in\mathcal{Y}^3_{\kappa^{\mathrm{R}}_{1}}\cup\mathcal{Y}^6_{\kappa^{\mathrm{R}}_{1}}$ and $z\in D_\epsilon(\kappa^{\mathrm{R}}_{1})\setminus\left(\mathcal{Y}^3_{\kappa^{\mathrm{R}}_{1}}\cup\mathcal{Y}^6_{\kappa^{\mathrm{R}}_{1}}\right)$.
Hence we can write
$$
\tilde{\delta}(z)=\exp\left\{i\rho\left(\ln_+(z-\kappa^{\mathrm{R}}_{1})+\ln_-(z-\kappa^{\mathrm{R}}_{1})\right)+\tilde{\chi}(z)\right\},\quad z\in\mathbb{C}\setminus\left(\mathcal{Y}^3_{\kappa^{\mathrm{R}}_{1}}\cup\mathcal{Y}^6_{\kappa^{\mathrm{R}}_{1}}\right),
$$
where the constant $\rho\equiv\rho(\xi)>0$ is defined by
$$
\rho=\frac{1}{2\pi}\ln\left(1+|r(\kappa^{\mathrm{R}}_{1})|^2\right),
$$
and the function $\tilde{\chi}(z)\equiv\tilde{\chi}(z,\xi)$ is defined by
\begin{align}
	\tilde{\chi}(z):=\frac{1}{2\pi i}&\Bigl(\ln\left(z-E_{01}\right)\ln\left(1+|r(E_{01})|^2\right)+\ln\left(z-\bar{\kappa}^{\mathrm{C}}_2\right)\ln\left(1+|r(\bar{\kappa}^{\mathrm{C}}_2)|^2\right)\nonumber\\
	&-\int_{\mathcal{Y}^3_{\kappa^{\mathrm{R}}_{1}}}\ln(s-z)\,d\ln\left(1+|r(s)|^2\right)+\int_{\mathcal{Y}^6_{\kappa^{\mathrm{R}}_{1}}}\ln(s-z)\,d\ln\left(1+|r(s)|^2\right)\Bigr).
\end{align}
We also consider the behavior of $\delta(z)$ as $z$ approaches $\kappa^C_1$.
Hence we can write $\delta(z)$ as
$$
\delta(z)=e^{-i \rho \ln (z-\kappa^R_1)+\chi(z)}, \quad k \in \mathbb{C} \backslash I,
$$
where the function $\chi(z) \equiv \chi(z,\xi)$ is defined by
\begin{align}
	\chi(z):=&  \frac{1}{2 \pi i} \sum_{k<\frac{n}{2}}\ln (z-B_k) \ln \frac{1+\left|r_{+}(B_k)\right|^2}{1+\left|r_{-}(B_k)\right|^2} -\frac{w(z)}{2 \pi i}\int_{I} \frac{\ln (z-s)}{w(s)} d\ln \left(1+|r(s)|^2\right) \\
	&+\frac{w(z)}{2\pi i}\int_{B_{k_0}}^{\kappa^R_2}\ln\left(1+|r(s)|^2\right)d \frac{\ln (z-s)}{w(s)} +\ln\tilde\nu(z),\quad z \in \mathbb{C} \backslash I,
\end{align}
Define
\begin{equation}
	\hat{\delta}(z)=\tilde{\delta}(z)\delta(z)=\hat{\delta}_0(\tilde{\zeta})\hat{\delta}_1(t)\hat{\delta}_2(z),\quad z\in\mathbb{C}\setminus\left(\mathcal{Y}^3_{\kappa^{\mathrm{R}}_{1}}\cup\mathcal{Y}^6_{\kappa^{\mathrm{R}}_{1}}\cup I\right),
\end{equation}
where
\begin{align}
	\hat{\delta}_0(\tilde{\zeta})&=\mathrm{e}^{-i\nu\left(\ln_{-\pi}(\tilde{\zeta})-2\ln_0(\tilde{\zeta})\right)},\quad\tilde{\zeta}\in\mathbb{C}\setminus\left(\mathbb{R}\cup i\mathbb{R}_-\right),\\
	\hat{\delta}_1(t)&=\mathrm{e}^{\frac{\pi\nu}{2}}2^{-i\nu}t^{-\frac{i\nu}{2}}\mathrm{e}^{-i\nu\ln\psi_{\kappa^{\mathrm{R}}_{1}}(\kappa^{\mathrm{R}}_{1})}\mathrm{e}^{\chi(\kappa^{\mathrm{R}}_{1})+\tilde{\chi}(\kappa^{\mathrm{R}}_{1})},\quad t>0,\\
	\hat{\delta}_2(z)&=\mathrm{e}^{-i\nu\ln\tilde{\psi}_{\kappa^{\mathrm{R}}_{1}}(z)}\mathrm{e}^{\chi(z)-\chi(\kappa^{\mathrm{R}}_{1})+\tilde{\chi}(z)-\tilde{\chi}(\kappa^{\mathrm{R}}_{1})},\quad z\in D_\epsilon(\kappa^{\mathrm{R}}_{1}),
\end{align}
with $\tilde{\psi}_{\kappa^{\mathrm{R}}_{1}}(z)=\psi_{\kappa^{\mathrm{R}}_{1}}(z)\psi^{-1}_{\kappa^{\mathrm{R}}_{1}}(\kappa^{\mathrm{R}}_{1})$, and
\begin{equation}
	\ln(z-\kappa^{\mathrm{R}}_{1})=\ln\frac{-i\tilde{\zeta}}{2\sqrt{t}\psi_{\kappa^{\mathrm{R}}_{1}}(z)}=-\frac{\pi i}{2}-\ln(2\sqrt{t})+\ln_{-\pi/2}(\tilde{\zeta})-\ln\psi_{\kappa^{\mathrm{R}}_{1}}(z),\quad z\in D_\epsilon(\kappa^{\mathrm{R}}_{1})\setminus(-\infty,\kappa^{\mathrm{R}}_{1}].
\end{equation}
Define $N_2^{\kappa^{\mathrm{R}}_{1}}(z)$ by
\begin{equation}
	N_2^{\kappa^{\mathrm{R}}_{1}}(z)=N_1^{\kappa^{\mathrm{R}}_{1}}(z)\hat{\delta}_1(t)^{\sigma_3},\quad z\in D_\epsilon(\kappa^{\mathrm{R}}_{1})\setminus\Sigma^{(2)}.
\end{equation}
Then the jumps of \eqref{PC_jj} in $D_\epsilon(\kappa^{\mathrm{R}}_{1})$ can be expressed as
\begin{equation}
	J_2^{\kappa^{\mathrm{R}}_{1}}(z)=I+\begin{cases}
		-r(z)\left(\hat{\delta}_0\hat{\delta}_2\right)^{-2}\mathrm{e}^{i\frac{\tilde{\zeta}^2}{2}}\sigma_-, & z\in\mathcal{Y}^1_{\kappa^{\mathrm{R}}_{1}},\\
		-(1+|r(z)|^2)r(z)\left(\hat{\delta}_0\hat{\delta}_2\right)^{-2}\mathrm{e}^{i\frac{\tilde{\zeta}^2}{2}}\sigma_-, & z\in\mathcal{Y}^4_{\kappa^{\mathrm{R}}_{1}},\\
		-\overline{r(\bar{z})}\left(\hat{\delta}_0\hat{\delta}_2\right)^2\mathrm{e}^{-i\frac{\tilde{\zeta}^2}{2}}\sigma_+, & z\in\mathcal{Y}^2_{\kappa^{\mathrm{R}}_{1}},\\
		(1+|r(z)|^2)r(z)\left(\hat{\delta}_0\hat{\delta}_2\right)^2\mathrm{e}^{-i\frac{\tilde{\zeta}^2}{2}}\sigma_+, & z\in\mathcal{Y}^5_{\kappa^{\mathrm{R}}_{1}},
	\end{cases}
\end{equation}
where all contours are oriented toward the origin as in Figure \ref{Figure-3}.

Define the function $\hat{\delta}_3(z)\equiv\hat{\delta}_3(z,\xi)$ by
\begin{equation}
	\hat{\delta}_3(\tilde{\zeta})=\mathrm{e}^{i\rho\ln_{-\pi}(\tilde{\zeta})},
\end{equation}
that is, $\hat{\delta}_3(z)$ equals the function $z^{i\rho}$ with branch cut along the negative imaginary axis. Since
\begin{equation}
	\hat{\delta}_0=\hat{\delta}_3\times\begin{cases}
		1, & \arg z\in(0,\pi),\\[4pt]
		\mathrm{e}^{2\pi\rho}, & \arg z\in(-\pi,-2\pi),
	\end{cases}
\end{equation}
and $\mathrm{e}^{2\pi\rho}=1+|r(\kappa^{\mathrm{R}}_{1})|^2$, we can write
\begin{equation}
	J_2^{\kappa^{\mathrm{R}}_{1}}(z)=\begin{cases}
		r_{\kappa^{\mathrm{R}}_{1}}\tilde{\zeta}^{-2i\nu}\mathrm{e}^{i\frac{\tilde{\zeta}^2}{2}}\sigma_-, & z\in\mathcal{Y}^1_{\kappa^{\mathrm{R}}_{1}},\\[6pt]
		\dfrac{r_{\kappa^{\mathrm{R}}_{1}}}{1+|r(\kappa^{\mathrm{R}}_{1})|^2}\tilde{\zeta}^{-2i\nu}\mathrm{e}^{i\frac{\tilde{\zeta}^2}{2}}\sigma_-, & z\in\mathcal{Y}^4_{\kappa^{\mathrm{R}}_{1}},\\[8pt]
		\bar{r}_{\kappa^{\mathrm{R}}_{1}}\tilde{\zeta}^{2i\nu}\mathrm{e}^{-i\frac{\tilde{\zeta}^2}{2}}\sigma_+, & z\in\mathcal{Y}^2_{\kappa^{\mathrm{R}}_{1}},\\[6pt]
		-\dfrac{\bar{r}_{\kappa^{\mathrm{R}}_{1}}}{1+|r(\kappa^{\mathrm{R}}_{1})|^2}\tilde{\zeta}^{2i\nu}\mathrm{e}^{-i\frac{\tilde{\zeta}^2}{2}}\sigma_+, & z\in\mathcal{Y}^5_{\kappa^{\mathrm{R}}_{1}},
	\end{cases}
\end{equation}
where $r_{\kappa^{\mathrm{R}}_{1}}=r(\kappa^{\mathrm{R}}_{1})\exp\left\{2i\rho\ln_{i\pi}\tilde{\zeta}-2i\rho\ln_0\tilde{\zeta}\right\}\hat{\delta}^{-2}_2(\kappa^{\mathrm{R}}_{1})$, which are exactly the jumps of the parabolic cylinder model problem in Appendix~A in \cite{MR18}.
Then we consider the case $z\to\kappa^{\mathrm{R}}_{2}$. Using a similar method, we obtain
\begin{equation}
	N^{\kappa^{\mathrm{R}}_{2}}(z)=I+\frac{1}{\hat\zeta}\begin{pmatrix}
		0&-i\beta_{12}(r_{\kappa^{\mathrm{R}}_{2})}\\[4pt]
			i\beta_{21}(r_{\kappa^{\mathrm{R}}_{2}})&0
		\end{pmatrix}+\mathcal{O}(\hat\zeta^{-2}),
	\end{equation}
	where
	\begin{equation}\label{k--222}
		r_{\kappa^{\mathrm{R}}_{2}}=r(\kappa^{\mathrm{R}}_{2})\exp\!\left\{4it\rho \psi^2_{\kappa^{\mathrm{R}}_{2}}(z)+2\hat\chi(z)\right\},
	\end{equation}
	and the local scaling parameter is defined by
	\begin{equation}\label{eq:dsds}
		\hat\zeta=2\sqrt{t\theta_{\kappa^{\mathrm{R}}_{2}}(z)}=2\sqrt{t}(z-\kappa^{\mathrm{R}}_{2})\psi_{\kappa^{\mathrm{R}}_{2}}(z),
	\end{equation}
	with $\psi_{\kappa^{\mathrm{R}}_{2}}(z)$ analytic in $D_\epsilon(\kappa^{\mathrm{R}}_{2})$. The function $\hat\chi(z)$ is given by
	\begin{align}
		\hat\chi(z):={}& - \frac{1}{2\pi i} \ln(z-B_{k_0}) \frac{\log(1+|r(B_{k_0})|^2)}{w(B_{k_0})} -\frac{w(z)}{2\pi i}\int_{B_{k_0}}^{\kappa^{\mathrm{R}}_{2}} \ln(z-s)\,\mathrm{d}\!\left(\frac{\ln(1+|r(s)|^2)}{w(s)}\right)\notag\\
		&+\frac{w(z)}{2\pi i}\int_{-\infty}^{\kappa^{\mathrm{R}}_{1}}\frac{\log(1+|r(s)|^2)}{w(s)(s-z)}\,\mathrm{d}s+\ln\tilde\nu(z), \quad z \in \mathbb{C} \setminus I.
	\end{align}
\subsection{Solution of the Small norm problem}
	Utilizing the function $N^{alg}(z)$ defined by the Riemann-Hilbert (RH) problem \ref{RH-mod}, together with the local model functions $N^{E_{00}}(z)$, $N^{\kappa^{\mathbb{R}}_{s_0}}(z)$ and $N^{\kappa^{\mathrm{R}}_{1}}(z)$, we implicitly define an unknown function $\mathcal{E}(z)$ that is analytic in the domain $\mathbb{C}\setminus\Sigma^{\mathcal{E}}$, where the contour $\Sigma^{\mathcal{E}}$ is given by
	\begin{equation}
		\Sigma ^{\mathcal{E}}=\left(\Sigma^{(2)} \setminus\left(\Sigma^{alg} \cup \overline{D}\right)\right) \cup \partial D_{\kappa^{\mathrm{R}}_{1}} \cup \partial D_{E_{00}} \cup \partial D_{\kappa^C_{s_0}} \cup \partial D_{\bar{\kappa}^C_{s_0-1}} \cup \partial D_{\bar{\kappa}^C_{s_0}}.
	\end{equation}
	We endow the boundary $\partial\Sigma^{\mathcal{E}}$ with the clockwise orientation. It is straightforward to verify that the function $\mathcal{E}(z)$ satisfies the following RH problem.
	
	\begin{problem}\label{RH-EE}
		Construct a meromorphic function $\mathcal{E}(z):\mathbb{C}\setminus \Sigma^{\mathcal{E}}\to \mathrm{SL}_2(\mathbb{C})$ such that:
		\begin{enumerate}
			\item As $|z|\to\infty$, $\mathcal{E}(z)=I+\mathcal{O}(z^{-1})$.
			\item For every $z\in\Sigma^{\mathcal{E}}$, the boundary values $\mathcal{E}_\pm(z)$ satisfy the jump condition
			$$\mathcal{E}_+(z)=\mathcal{E}_-(z)J_{\mathcal{E}}(z),$$
			where the jump matrix $J_{\mathcal{E}}(z)$ is defined as
			\begin{equation}\label{jump-Ee}
				J_{\mathcal{E}}(z)=
				\begin{cases}
					N^{alg}(z) J^{(2)}(z)\left(N^{alg}(z)\right)^{-1}, & z \in \Sigma^{(2)} \setminus \overline{\mathcal{D}}, \\
					N^{alg}(z)N^{(\tilde\kappa)}(z)\left(N^{alg}(z)\right)^{-1}, & z \in  \partial D_\epsilon(\tilde\kappa).
				\end{cases}
			\end{equation}
		\end{enumerate}
			Here, $\tilde\kappa$ is defined in equation \eqref{eq:N-2-dis}.
	\end{problem}
	
 Starting from the jump matrix \eqref{eq:J(z)-aasy2}, and applying identities \eqref{eq:mo-pro} (for $z\in\mathbb{C}\setminus\Sigma_{\kappa^{\mathrm{R}}_{1}}$), \eqref{eq:zeta-2}, combined with the boundedness of $N^{alg}(z)$ on $D_{\kappa^{\mathrm{R}}_{1}}$, we derive the asymptotic estimates for the jump matrix:
	\begin{align}
		\left|J_{\mathcal{E}}(z)-I\right|=
		\begin{cases}
			\mathcal{O}\left(e^{-ct}\right),& z\in\Sigma^{(2)}\setminus D,\\
			\mathcal{O}(t^{-N}), &z\in D\setminus D_\epsilon(\kappa^{\mathrm{R}}_{1}),\ N\ge1,\\
			\mathcal{O}(t^{-1/2}),& z\in D_\epsilon(\kappa^{\mathrm{R}}_{1}),
		\end{cases}
	\end{align}
	where $c=c(\xi)>0$ is a positive constant as $t\to\infty$.
	
	To reconstruct the solution $q(x, t)$ of equation \eqref{eq:fNLS}, we require the large-$z$ asymptotic behavior of the solution to the RH problem \ref{RH-EE}. We now present the large-$z$ expansion of $\mathcal{E}(z)$:
	\begin{equation}
		\mathcal{E}(z)=I+\frac{\mathcal{E}_1(z)}{z}+\mathcal{O}(z^{-2}),
	\end{equation}
	where the coefficient $\mathcal{E}_1(z)$ takes the form
	\begin{align}
		\mathcal{E}_1(z)=&\frac{1}{2i\sqrt{t}(z-\kappa^{\mathrm{R}}_{1}) \psi_{\kappa^{\mathrm{R}}_{1}}(z)} N^{alg}(\kappa^{\mathrm{R}}_{1})
		\begin{pmatrix}
			0 & \beta_{12}\left(r_{\kappa^{\mathrm{R}}_{1}}\right) \\
			-\beta_{21}\left(r_{\kappa^{\mathrm{R}}_{1}}\right) & 0
		\end{pmatrix}
		\left(N^{alg}(\kappa^{\mathrm{R}}_{1})\right)^{-1}\\
	&	+\frac{1}{2i\sqrt{t}(z-\kappa^{\mathrm{R}}_{2}) \psi_{\kappa^{\mathrm{R}}_{2}}(z)} N^{alg}(\kappa^{\mathrm{R}}_{2})
		\begin{pmatrix}
			0 & \beta_{12}\left(r_{\kappa^{\mathrm{R}}_{2}}\right) \\
			-\beta_{21}\left(r_{\kappa^{\mathrm{R}}_{2}}\right) & 0
		\end{pmatrix}
		\left(N^{alg}(\kappa^{\mathrm{R}}_{2})\right)^{-1}+
		\mathcal{O}\left(t^{-1}\right),
	\end{align}
	and the constants $=\beta_{12}(r_{\kappa^{\mathrm{R}}_{1}}),\beta_{21}(r_{\kappa^{\mathrm{R}}_{t}})$ are explicitly given by
	\begin{equation}\label{beta-12-21}
		\beta_{12}=\beta_{12}(r_{\kappa^{\mathrm{R}}_{t}})=\frac{\sqrt{2 \pi} e^{i \pi / 4} e^{-\pi \rho / 2}}{r_{\kappa^{\mathrm{R}}_{1}} \Gamma(-i \kappa^{\mathrm R}_t)},\quad
		\beta_{21}=\beta_{21}(r_{\kappa^{\mathrm{R}}_{t}})=\frac{-\sqrt{2 \pi} e^{-i \pi / 4} e^{-\pi \rho / 2}}{\overline{r}_{\kappa^{\mathrm{R}}_{t}} \Gamma(i \kappa_t^{\mathrm R})},\quad t=1,2.
	\end{equation}

Reversing the sequence of transformations \eqref{tran-1} and \eqref{eq:trans-2},
the solution to the RH problem \ref{RH2-2} admits the representation
\[
N(z)=e^{\sigma_3\left(\ln\delta(\infty)+ig(\infty)\right)}\mathcal{E}(z)N^{alg}(z)e^{-ig(z)\sigma_3}
\bigl(G^{(1)}(z)\bigr)^{-1}G^{-1}(z)\delta^{-\sigma_3}(z),\quad z\in\mathbb{C}\setminus U.
\]
Taking the asymptotic expansion as $z\to\infty$, we obtain
\begin{equation}
	N(z)=e^{\left(\ln\delta(\infty)+ig(\infty)\right)\hat\sigma_3}
	\left(I+\frac{\mathcal{E}_1}{z}+\mathcal{O}(z^{-2})\right)\left(I+\frac{N_1^{alg}}{z}+\mathcal{O}(z^{-2})\right)\left(I-\frac{I_2}{z}+\mathcal{O}(z^{-2})\right)
	\left(I-\frac{I_1}{z}+\mathcal{O}(z^{-2})\right),
\end{equation}
where $I_1$ and $I_2$ are given in \eqref{delt-ing} and \eqref{ginfty}. Recall the reconstruction formula for the solution of the fNLS equation:
\begin{align}
	q(x,t)
	&= 2i\lim_{z\to\infty}z\,N_{12}(z) \nonumber\\
	&= 2ie^{2\left(\ln\delta(\infty)+ig(\infty)\right)}
	\left(\left(\mathcal{E}_1\right)_{12}+\left(N_1^{alg}\right)_{12}\right).
\end{align}
This completes the proof of Theorem \ref{theorem-1}.
	\begin{appendices}

\section{Riemann manifolds}\label{APP-A}
In order to present an explicit solution of the RH problem to express $M^{alg}(z)$, we need several ingredients from the theory of the $n$ genus of Riemann manifolds $\mathcal{X}$. Let a set of points $\left\{E_k, \bar{E}_k\right\}_{k=0}^n$ in the complex plane be given,
\begin{equation}\label{genus-surface}
w^2=P(z)=\prod_{k=0}^n\left(z-E_k\right)\left(z-\bar{E}_k\right)
\end{equation}
The upper and lower sheets of $\mathcal{X}$ are denoted by $\mathcal{X}_{+}$and $\mathcal{X}_{-}$respectively; they are fixed by the relations
$$
\sqrt{P(z)}= \pm z^{n+1}\left(1+\mathrm{O}\left(z^{-1}\right)\right), \quad z=\pi(\mathcal{P}) \rightarrow \infty, \quad P \in \mathcal{X}_{ \pm}
$$
where $z=\pi(P)$ is the standard projection of $P=(w, z) \in \mathcal{X}$ on the Riemann sphere $\mathbb{C P}^1$. Thus each point on the $z$-plane has two preimages $P_{ \pm}=\mathcal{X}_{ \pm}$, except for the branch points. Denote the preimage of $z=\infty$ on $\mathcal{X}_{ \pm}$by, respectively, $\infty^{ \pm}$. With the inclusion of two points $\left(\infty^{+}, \infty^{-}\right), \mathcal{X}$ becomes a compact Riemann surface of genus $n$. The square root $\sqrt{P(z)}$ turns into a meromorphic function on its own compact Riemann surface $\mathcal{X}$, which have $2 n+2$ zeros at $E_k$ and $\bar E_k, k=0,1, \ldots, n$, and two poles at $\infty^{+}$and $\infty^{-}$, each of multiplicity $n+1$.

We next discuss the choice of branch cuts and the choice of a canonical homology basis $\{a_j,b_j\}_j$ for $\mathcal{X}$. One possibility is to let all branch cuts be vertical segments and adopt the standard choice of canonical homology basis $\{a_j,b_j\}_j$. Here $a_j$ is a closed curve on the upper sheet going conterclockwise around the interval $\bar E_k,E_k$, and $b_j$ start from $\bar E_0,E_0$, goes on the upper sheet to $(\bar E_j,E_j)$, and returns ont he lower sheet tot he starting point. the cycles are chosen so that their intersection matrix read
\begin{equation}
	a_j\circ b_j=\delta_{ij},\quad i,j=1\dots,n,
\end{equation}
Here $\delta_{ij}$ denotes the Kronecker delta.

We let $\zeta=(\zeta_1,\zeta_2,\dots,\zeta_{n})$ be the normalized basis of $\mathcal{H}^1(\mathcal{X})$ due to the canonical homology basis $\{a_j,b_j\}$ in the sense that $\zeta_j$ are holomorphic differentials such that
\begin{equation}
\int_{a_i} \zeta_j=\delta_{i j}
\end{equation}
Since $w$ satisfies $w(z^+)=-w(z^-)$, we have the same symmetry for $\zeta$:
\begin{equation}
\zeta\left(z^{+}\right)=-\zeta\left(z^{-}\right) .
\end{equation}
We let $\tau=\left(\tau_{jl}\right)$ denote the $n\times n$ period matrix, defined by
\begin{equation}
\tau_{j l}:=\int_{b_j} \zeta_l .
\end{equation}
It is standard result that $\tau$ then  is symmetric and has positive definite imaginary part,
\begin{equation}
\int_{a_j^{+}} \zeta=\frac{1}{2} \int_{a_j} \zeta=\frac{1}{2} \mathrm{e}_j \quad \text { and } \quad \int_{b_j^{+}} \zeta=\frac{1}{2} \int_{b_j} \zeta=\frac{1}{2} \tau_j,
\end{equation}
where $a_j^{+}, b_j^{+}$denote the restrictions of $a_j, b_j$ to the upper sheet, $e_j$ denotes the $j$ th column of the $n \times n$ identity matrix $I$, and $\tau^{(j)}$ denotes the $j$ th column of $\tau=\left(\tau_{j l}\right)$.

 The theta function $\Theta$ associated to $\tau$ is defined by
\begin{equation}\label{eq:THeze}
\Theta(z):=\sum_{N \in \mathbb{Z}^n} \mathrm{e}^{2 \pi i\left(\frac{1}{2} N^T \tau N+N^T z\right)}, \quad z \in \mathbb{C}^n
\end{equation}
It is holomorphic since the sum converges nicely due to \cite{AAA06} . We will now show some simple properties of $\Theta(z)$. First observe that
\begin{equation}
	\Theta(-z)=\Theta(z)
\end{equation}
and second let $n,n'\in\mathbb{Z}^n$, then
\begin{equation}
	\begin{aligned}
		\Theta\left({z}+{n}+{\tau}{n}^{\prime}\right)=\exp 2 \pi i\left(-\left\langle{n}^{\prime}, {z}\right\rangle-\frac{\left\langle{n}^{\prime}, {\tau} {n}^{\prime}\right\rangle}{2}\right) \Theta({z}) .
	\end{aligned}
\end{equation}
We let $\varphi: M \rightarrow \mathbb{C}^n$ denote the Abel-type map with base point $\bar E_0$, i.e.,
\begin{equation}\label{eq:abel}
\varphi(P)=\int_{\bar{E}_0}^z \zeta, \quad z \in M
\end{equation}
Besides, for any non-special integral divisor ${D}=P_1+\ldots+P_n$ of degree $n$, there exists a vector $w({D})$ such that the Riemann theta function $\Theta(\varphi(z)+w(D))$ defined on $\mathcal{X}$ with cuts along of the cycles $a_j$ and $b_j$ has precisely $n$ zeros at $P_j, j=1, \ldots, n$. The vector $w({D})$ is defined by
\begin{equation}
	w(D)=-\varphi(D)-K
\end{equation}
where the Riemann constant vector $K$ is defined, component-wise by
\begin{equation}\label{eq:const}
	K_j=\frac{1}{2}-\frac{\tau_{jj}}{2}+\sum_{i=1,i\ne j}\int_{a_i}\left(\varphi_j\right)\varphi_i.
\end{equation}
\section{The Painlevé II Model RH Problem }\label{App-B}
The Painlev\'{e} II equation takes the form
\begin{equation}\label{PE-II}
	u_{qq}(q)=2u^3(q)+qu(q), \quad s \in \mathbb{R}.
\end{equation}
\begin{problem}\label{RHBB}
	Construct a meromorphic function $N^P(\lambda):\mathbb{C} \setminus L$ such that:
with
$$
L_j:=\left\{z \in \mathbb{C} \left\lvert\, \arg z=\frac{(2k-1)\pi}{4} \pi\right.\right\}, \quad j=1, 2,3,4.
$$
as shown in Figure \ref{fig-a-a}.
\begin{figure}
	{
		\begin{minipage}{15cm}\centering
			\includegraphics[scale=0.35]{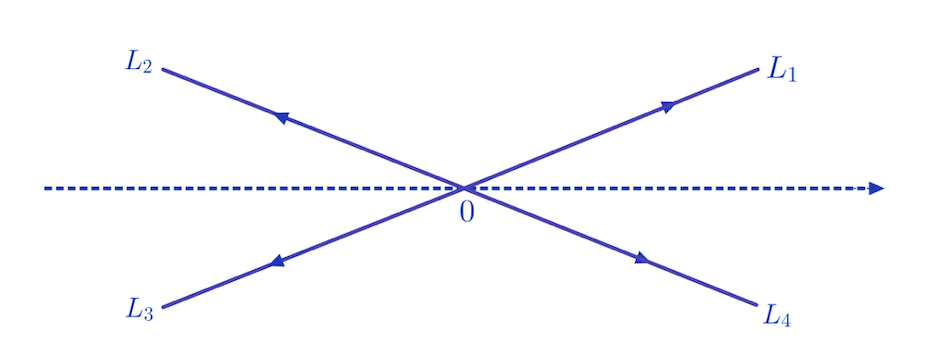}
	\end{minipage}}
	\caption{The contour $L_j, j=1,2,3,4$ in the for the RH problem \ref{RHBB}.}
	\label{fig-a-a}
\end{figure}
\begin{enumerate}
	\item As $\lambda \rightarrow \infty$, we have
	\begin{equation}\label{App-B-psi}
		N^P(\lambda)=I+\lambda^{-1} N_1^P+\mathcal{O}\left(\lambda^{-2}\right),
	\end{equation}
	where
	\begin{equation}\label{eq:Pe-e}
		N_1^P=\frac{1}{2}\begin{pmatrix}
			-i \int_p^{\infty} u(\lambda)^2 d \lambda & u(p) \\
			u(p) & i \int_p^{\infty} u(\lambda)^2 \mathrm{~d} \lambda
	\end{pmatrix}.
	\end{equation}
\item For each $z\in L$, the boundary values $N^P_\pm(\lambda)$ satisfy the jump relation
$$
N^P_{+}(\lambda)=N^P_{-}(\lambda) e^{-i\left(\frac{4}{3} \lambda^3+s \lambda\right) \hat{\sigma}_3}{P}(\lambda),
$$
where
$$
J_{P}(\lambda)= \begin{cases}\begin{pmatrix}
		1 & 0 \\
		\rho & 1
	\end{pmatrix}, & \lambda \in L_1, \qquad\,\,\,
\begin{pmatrix}
		1 & 0 \\
		-\rho & 1
\end{pmatrix},  \qquad \lambda \in L_{2}, \\
\begin{pmatrix}
		1 &- \rho \\
		0 & 1
\end{pmatrix} & \lambda \in L_{3}, \qquad\,\,
\begin{pmatrix}
		1 & \rho \\
		0 & 1
	\end{pmatrix},\qquad\quad  \lambda \in L_4.\end{cases}
$$
\item $N^P(\lambda)$ is bounded near the origin.
\end{enumerate}
\end{problem}
Then
$$
u(p)=2\left(N_1^P\right)_{12}=2\left(N_1^P\right)_{21}
$$
solves the Painlev\'{e} II equation \eqref{PE-II}. Also, a result due to Hastings and McLeod \cite{HAS} states that, for any $a=\operatorname{Im}\rho$, there exists a unique solution to the homogeneous Painlev\'{e} II equation \eqref{PE-II} that behaves as
\begin{equation}\label{app-b-as}
	u(p)=a \operatorname{Ai}(p)+\mathcal{O}\left(p^{-\frac{1}{4}} e^{-\frac{4}{3} p^{3 / 2}}\right), \quad s \rightarrow+\infty,
\end{equation}
where $\operatorname{Ai}(s)$ denotes the Airy function.

\section{Exact Solution in Terms of Airy Functions}\label{App-C}
Let $Y:=Y_1 \cup \cdots \cup Y_4 \subset \mathbb{C}$ denote the union of the four rays
$$
\begin{array}{ll}
	Y_1:=\{s \mid 0 \leq s<\infty\}, & Y_2:=\left\{\left.s \mathrm{e}^{\frac{2 i \pi}{3}} \right\rvert\, 0 \leq s<\infty\right\}, \\
	Y_3:=\{-s \mid 0 \leq s<\infty\}, & Y_4:=\left\{\left.s \mathrm{e}^{-\frac{2 i \pi}{3}} \right\rvert\, 0 \leq s<\infty\right\},
\end{array}
$$
oriented toward the origin as in Figure \ref{F4.111}.
Define the open sectors $\left\{S_j\right\}_1^4$ by
$$
\begin{array}{ll}
	S_1:=\left\{\arg k \in\left(0, \frac{2 \pi}{3}\right)\right\}, & S_2:=\left\{\arg k \in\left(\frac{2 \pi}{3}, \pi\right)\right\} \\
	S_3:=\left\{\arg k \in\left(\pi, \frac{4 \pi}{3}\right)\right\}, & S_4:=\left\{\arg k \in\left(\frac{4 \pi}{3}, 2 \pi\right)\right\}
\end{array}
$$
\begin{figure}[htp]
	{
		\begin{minipage}{15cm}\centering
			\includegraphics[scale=0.3]{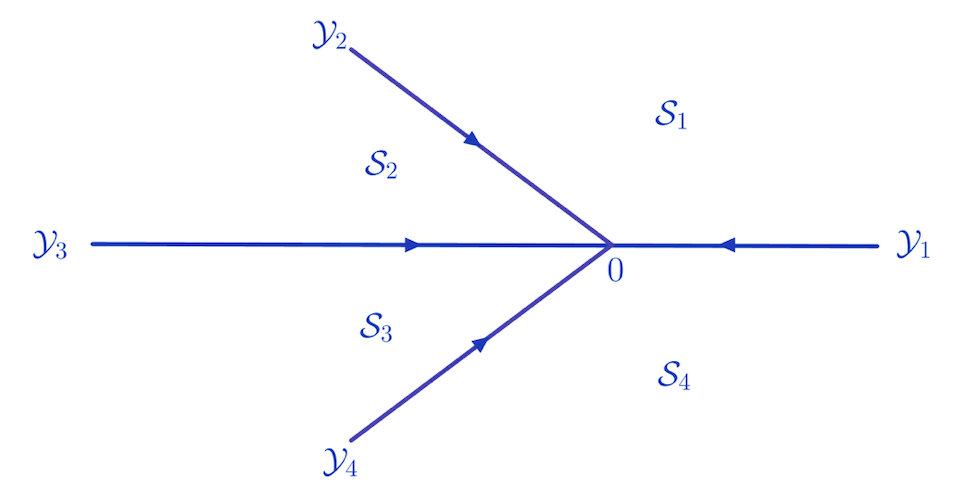}
	\end{minipage}}
	\caption{The jump contour $\mathcal{Y}_j, j=1,2,3,4$ for the RH problem \ref{RHCC12}.}
	\label{F4.111}
\end{figure}
\begin{problem}\label{RHCC12}
	Construct a meromorphic function $N^{Ai}(\zeta) :\mathbb{C}\setminus \Sigma^{Ai}\to SL_2(\mathbb{C})$  such that:
\begin{enumerate}
	\item $N^{Ai}(\zeta)=I+\mathcal{O}(z^{-1}),$\quad $|z|\to\infty$.
	\item 	For each $\zeta\in\Sigma^{Ai}$, the boundary values $N^{Ai}_\pm(\zeta)$ satisfy the jump relation
	$$N^{Ai}_+(\zeta)=N^{Ai}_-(\zeta)J^{Ai}(\zeta),\quad \zeta \in Y \backslash\{0\},$$
	where
$$
J^{Ai}(\zeta):= \begin{cases}\begin{pmatrix}
		1 & -\mathrm{e}^{-\frac{4}{3} \zeta^{3 / 2}} \\
		0 & 1
\end{pmatrix}, & \zeta \in Y_1, \\
\begin{pmatrix}
		1 & 0 \\
		\mathrm{e}^{\frac{4}{3} \zeta^{3 / 2}} & 1
\end{pmatrix}, & \zeta \in Y_2 \cup Y_4, \\
\begin{pmatrix}
		0 & 1 \\
		-1 & 0
\end{pmatrix}, & \zeta \in Y_3 .\end{cases}
$$
\item  For each integer $N \geq 0$, the functions $N_{as, N}^{Ai}(\zeta)$ and $N_{as, N}^{Ai, \text { inv }}(\zeta)$ are analytic for $\zeta \in \mathbb{C} \backslash(-\infty, 0]$ and satisfy the following jump relations on the negative real axis:
\begin{align}
&	N_{as, N+}^{Ai}(\zeta)=N_{as, N-}^{Ai}(\zeta)\begin{pmatrix}
		0 & 1 \\
		-1 & 0
\end{pmatrix}, & \zeta<0 ,\\\label{eq:App-jm}
&	N_{as, N+}^{Ai, \text { inv }}(\zeta)=\begin{pmatrix}
		0 & -1 \\
		1 & 0
\end{pmatrix} N_{as, N-}^{Ai, \text { inv }}(\zeta), & \zeta<0.
\end{align}
 For each integer $N \geq 0$, define the asymptotic approximations $N_{as, N}^{Ai}(\zeta)$ and $N_{as, N}^{Ai, \text { inv }}(\zeta)$ of $N^{Ai}(\zeta)$ and $N^{Ai}(\zeta)$, respectively, by
\begin{align}
&
N_{as, N}^{Ai}(\zeta)=\frac{\mathrm{e}^{\frac{i \pi}{12}}}{2 \sqrt{\pi}} \sum_{k=0}^N \frac{1}{\left(\frac{2}{3} \zeta^{3 / 2}\right)^k} \zeta^{-\frac{1}{4} \sigma_3}\begin{pmatrix}
		(-1)^k u_k & u_k \\
		-(-1)^k v_k & v_k
\end{pmatrix} e^{-\frac{\pi i}{4} \sigma_3}, & \zeta \in \mathbb{C} \backslash Y, \\\label{eq:app-inv}
&N_{as, N}^{Ai, \text { inv }}(\zeta)=\sqrt{\pi} \mathrm{e}^{-\frac{i \pi}{12}} \sum_{k=0}^N \frac{1}{\left(\frac{2}{3} \zeta^{3 / 2}\right)^k} \mathrm{e}^{\frac{\pi i}{4} \sigma_3}\begin{pmatrix}
		v_k & -u_k \\
		(-1)^k v_k(-1)^k u_k
\end{pmatrix} \zeta^{\frac{1}{4} \sigma_3}, & \zeta \in \mathbb{C} \backslash Y,
\end{align}
where the real constants $\left\{u_j, v_j\right\}_0^{\infty}$ are defined by $u_0:=v_0:=1$ and
$$
u_k:=\frac{(2 k+1)(2 k+3) \ldots(6 k-1)}{(216)^k k!}, \quad v_k:=\frac{6 k+1}{1-6 k} u_k, \quad k=1,2, \ldots
$$
\item The functions  $N_{as, N}^{Ai}(\zeta)$ and $N_{as, N}^{Ai, \text { inv }}(\zeta)$ approximate $N^{Ai}(z)$ and its inverse as $\zeta \rightarrow \infty$ in the sense that
\begin{align}\label{App-as-1}
&	\left(N^{Ai}(z)\right)^{-1} N_{as, N}^{Ai}(\zeta) =I+\mathrm{O}\left(\zeta^{-\frac{3(N+1)}{2}}\right), & \zeta \rightarrow \infty ,\\\label{App-as-11}
&	N^{Ai}(z) =N_{as, N}^{Ai}(\zeta) +\mathrm{O}\left(\zeta^{-\frac{3(N+1)}{2}}\right), & \zeta \rightarrow \infty,
\end{align}

where the error terms are uniform with respect to $\arg \zeta \in[0,2 \pi]$.
\end{enumerate}
\end{problem}
\end{appendices}
\section*{acknowledgments}
The authors express their sincere gratitude to the anonymous reviewers for their careful reading of the manuscript and their insightful comments and suggestions. The authors are also grateful to Dr. Yiling Yang and Taiyang Xu for their valuable suggestions on this work. This research was supported by the National Natural Science Foundation of China under Grant Nos. 12271104.
\subsection*{Ethical Statement}
This manuscript is original and has not been published previously, nor is it under consideration for publication elsewhere. The study has not been divided into several parts to increase the number of submissions, and no portion of this work has been submitted to multiple journals simultaneously or sequentially.
\section*{Data Availability Statements}
All data supporting this study are included in the article.
\medskip
\let\em=\it

\makeatletter
\def\@biblabel#1{#1.}
\def\doibase{http://dx.doi.org/}
\def\reftitle#1{``#1''}
\def\booktitle#1{\textit{#1}}
\makeatother

\end{document}